\documentclass{journalA4}

\usepackage{microtype}
\usepackage{amsmath,amssymb}
\usepackage{graphicx,color}
\usepackage{tabularx}
\usepackage{xspace}
\usepackage{algo}
\usepackage{authblk}

\graphicspath{{figures/}{figures/graphs/}}

\expandafter\ifx\csname pdfoptionalwaysusepdfpagebox\endcsname\relax\else
\fi

\newif\ifcomments
\commentsfalse
\ifcomments
  \newcommand{\mdb}[1]{\textcolor{blue}{MdB: #1}}
  \newcommand{\marcel}[1]{\textcolor{blue}{Marcel: #1}}
  \newcommand{\mcomment}[1]{\marginpar{{\footnotesize #1}}}
\else
  \newcommand{\mdb}[1]{}
  \newcommand{\marcel}[1]{}
  \newcommand{\mcomment}[1]{}
\fi

\newcommand{\BeginMyItemize}{\begin{itemize}\setlength{\itemsep}{-\parskip}}
\newcommand{\EndMyItemize}{\end{itemize}}
\newcommand{\myitemize}[1]{\BeginMyItemize #1 \EndMyItemize}

\newcommand{\BeginMyEnumerate}{\begin{enumerate}\setlength{\itemsep}{-\parskip}}
\newcommand{\EndMyEnumerate}{\end{enumerate}}
\newcommand{\myenumerate}[1]{\BeginMyEnumerate #1 \EndMyEnumerate}

\newcommand{\myitpara}[1]{\vspace{10pt} \noindent \textit{#1}}

\newcommand{\mypara}[1]{\vspace{10pt} \noindent \textbf{\sffamily #1}}

\renewcommand{\leq}{\leqslant}
\renewcommand{\geq}{\geqslant}

\newcommand{\eps}{\varepsilon}

\newcommand{\mydef}{:=}

\newcommand{\etal}{{\emph{et al.}}\xspace}
\newtheorem{defin}{Definition}
  \newenvironment{definition}{\begin{defin} \sl}{\end{defin}}

\newtheorem{propo}[defin]{Proposition}
  
\newtheorem{coro}[defin]{Corollary}

\newtheorem{rem}[defin]{Remark}
  
\newtheorem{myfact}[defin]{Fact}

\newcommand{\Reals}{{\mathbb{R}}}

\newcommand{\dbs}{{\sc dbscan}\xspace}
\newcommand{\dbss}{{\dbs}$^*$\xspace}
\newcommand{\hdbs}{{\sc hdbscan}\xspace}
\newcommand{\optics}{{\sc optics}\xspace}
\newcommand{\mst}{{\sc mst}\xspace}

\newcommand{\minp}{\mbox{{\sc MinPts}}}

\newcommand{\epsclusters}[1]{\mathcal{C}_{#1}(D)}
\newcommand{\refine}{\prec}
\newcommand{\neps}{N_{\eps}}
\newcommand{\dist}{\mathrm{dist}}
\newcommand{\core}{\mathrm{core}}
\newcommand{\dcore}{D_{\core}}
\newcommand{\graph}{\mathcal{G}}
\newcommand{\tree}{\mathcal{T}}

\newcommand{\gnb}{\graph}
\newcommand{\enb}{E}
\newcommand{\gcore}{\graph_{\core}}
\newcommand{\ecore}{E_{\core}}
\newcommand{\border}{\mathrm{border}}
\newcommand{\dborder}{D_{\border}}

\newcommand{\cdist}{d_{\core}}
\newcommand{\mreach}{\mathrm{mr}}
\newcommand{\mdist}{d_{\mreach}}
\newcommand{\Gmr}{\graph_{\mreach}}
\newcommand{\GmrRed}{\overline{\graph}_{\mreach}}

\newcommand{\B}{\mathcal{B}}
\newcommand{\C}{\mathcal{C}}
\newcommand{\mybox}{\mathrm{box}}
\newcommand{\gbox}{\graph_{\mybox}}
\newcommand{\ebox}{E_{\mybox}}
\newcommand{\cone}{c}

\author[1]{Mark de Berg}
\author[1]{Ade Gunawan}
\author[2,3]{Marcel Roeloffzen}

\affil[1]{Department of Computing Science, TU Eindhoven \\
          P.O.~Box 513, 5600 MB Eindhoven, the Netherlands \\
          \texttt{mdberg@win.tue.nl}}
\affil[2]{National institute of informatics, Tokyo, Japan \\
          \texttt{marcel@nii.ac.jp}}
\affil[3]{JST ERATO, Kawarabayashi Large Graph Project}

\title{Faster \dbs and \hdbs in Low-Dimensional Euclidean Spaces\thanks{%
MdB is supported by the Netherlands Organization for Scientific Research under grant 024.002.003.}}

\begin{document}
\maketitle

\begin{abstract}
We present a new algorithm for the widely used density-based clustering method~\dbs.
Our algorithm computes the \dbs-clustering in $O(n\log n)$ time in~$\Reals^2$,
irrespective of the scale parameter~$\eps$
(and assuming the second parameter $\minp$ is set to a fixed constant, as is the case in practice).
Experiments show that the new algorithm is not
only fast in theory, but that a slightly simplified version is 
competitive in practice and much less sensitive to the choice of $\eps$ than
the original \dbs algorithm.
We also present an $O(n\log n)$ randomized algorithm for \hdbs in the plane---\hdbs
is a hierarchical version of \dbs
introduced recently---and we show how to compute an approximate version of
\hdbs in near-linear time in any fixed dimension.
\end{abstract}

\section{Introduction}
Clustering is one of the most fundamental tasks in data mining.
Due to the wide variety of applications where clustering is important,
the clustering problem comes in many variants. These
variants differ for example in the dimensionality of the data set~$D$
and in the underlying metric, but also in the objective of the clustering.
Thus a multitude of clustering algorithms has been developed~\cite{tsk-idm-06}, each with
their own strengths and weaknesses. We are interested in \emph{density-based clustering},
where clusters are defined by areas in which the
density of the data points is high and clusters are separated from each other
by areas of low density.

One of the most popular density-based clustering methods is \dbs;
the paper by Ester~\etal~\cite{eksx-dbadc-96} on \dbs has
been cited over 8,800 times, and in 2014 \dbs received the test-of-time award from KDD,
a leading data-mining conference.
\dbs has two parameters, $\eps$ and $\minp$, that together determine
when the density around a point $p\in D$ is high enough for~$p$
to be part of a cluster (as apposed to being noise); see Section~\ref{se:preli}
for a precise definition of the \dbs clustering.
Typically $\minp$ is a constant---in the original article~\cite{eksx-dbadc-96}
it is concluded that $\minp=4$ works well---but
finding the right value for $\eps$ is more difficult.
The worst-case running time of the original \dbs algorithm is $\Theta(n^2)$. It is often
stated that the running time is $O(n\log n)$ for
Euclidean spaces when a suitable indexing structure such as an R-tree
is used to support the \dbs algorithm. While this may be true
in certain practical cases, it is not true from a theoretical
point of view.

Several variants of \dbs algorithm have been proposed, often
with the goal to speed up the computation.
Some ({\sc idbscan}~\cite{bb-idbscan-04} and
{\sc fdbscan}~\cite{l-fdbscan-06}) do so at the expense of
computing a slightly different, and not clearly defined, clustering.
Others ({\sc g}ri\dbs~\cite{mm-uga-08}) compute the same
clustering as~\dbs, but without speeding up the worst-case running time.

A fundamental bottleneck of the original \dbs algorithm is that it
performs a query with each point $p\in D$ to find $\neps(p,D)$, the set
of points within distance $\eps$ of~$p$. Thus $\sum_{p\in D} |\neps(p,D)|$
is a lower bound on the running time of the \dbs algorithm.
In the worst case $\sum_{p\in D} |\neps(p,D)|=\Theta(n^2)$, so even with
a fast indexing structure the worst-case running time of the original \dbs
algorithm is $\Omega(n^2)$.
(Apart from this, the worst-case query time of R-trees and other
standard indexing structures is not logarithmic even if we disregard to time to report points.)
In most practical instances the \dbs algorithm is much faster than quadratic. The reason is that $\eps$ is typically
small so that the sets $\neps(p,D)$ do not contain many points and the range
queries can be answered quickly. However, the fact that the algorithm always explicitly reports the sets
$\neps(p,D)$ makes the running time sensitive to the choice of $\eps$
and the density of the point set~$D$. For example, suppose
we have a disk-shaped cluster with a Gaussian distribution around the disk center.
Then a suitable value of $\eps$ will lead to large sets $\neps(p,D)$
for points $p$ near the center of the cluster.

Chen~\etal~\cite{csx-gadbc-05} overcame the quadratic bottleneck of the standard
approach, and designed an algorithm\footnote{As described,
the algorithm actually computes a variation of the
\dbs clustering, but it is easily adapted to compute the true \dbs clustering.}
with $O(n^{2-\frac{2}{d+2}}\;\mbox{polylog}\;n)$ worst-case running time.
Note that for $d=2$ the running time of the exact algorithm is $O(n^{1.5}\;\mbox{polylog}\;n)$.
They also present an approximate algorithm that is more practical.
Chen~\etal remark that their exact algorithm is mainly of theoretical interest.
The natural question is then whether or not it is possible to to compute the
\dbs clustering in subquadratic time in the worst case, irrespective
of the value of~$\eps$, with a simple and practical algorithm?
\medskip

Although \dbs is used extensively and performs well in many situations,
it has its drawbacks. One is that it produces a flat (non-hierarchial)
clustering which heavily depends on the choice of the scale parameter~$\eps$. Ankerst~\etal~\cite{abks-optics-99}
therefore introduced \optics, which can be seen as a hierarchical version of~\dbs.
Recently Campello~\etal~\cite{cms-dbchde-13} proposed an improved density-based
hierarchical clustering method---similar to \optics but cleaner---together with a cluster-stability
measure that can be used to automatically extract relevant clusters. The new method,
called \hdbs, only needs the parameter~$\minp$, which is much easier to choose than~$\eps$.
(Campello~\etal used \minp=4  in all their experiments.)
While \hdbs is very powerful, the algorithm to compute
the \hdbs hierarchy runs in quadratic time; not only in the worst-case, but
actually also in the best-case.
There have been only few papers dealing with speeding up \hdbs or its predecessor \optics.
A notable recent exception is {\sc Poptics}~\cite{ppalmc-poptics-13}, a parallel algorithm
that computes a similar (though not the same) hierarchy as \optics.
We do not know of any algorithm that computes the \hdbs or \optics hierarchy in
subquadratic time.
Thus the second question we study is:
is it possible to compute the \hdbs hierarchy in subquadratic time?

\mypara{Our results.}
We present an $O(n\log n)$ algorithm to compute the \dbs clustering for a set~$D$
of $n$ points in the plane, irrespective of the setting of the parameter $\eps$
used to define the \dbs clustering. (Here, and in our other results, we assume that
the parameter~$\minp$ is a fixed constant. As mentioned this is the case in practice,
where one typically uses $\minp=4$.)
We remark that our algorithm is not only fast in theory, but a slightly simplified version is also
competitive in practice and much less sensitive to the choice of $\eps$ than
the original \dbs algorithm. Some basic experimental results are provided in Section~\ref{se:experiments}.

We also present a new algorithm for planar \hdbs: we show how to compute the
\hdbs hierarchy in $\Reals^2$ in $O(n\log n)$ expected time, thus obtaining the first
subquadratic algorithm for the problem.

Finally, we provide a slightly improved version of the approximate \dbs clustering algorithm by Chen~\etal~\cite{csx-gadbc-05} and by Gan and Tao~\cite{gt-dbs-15} (their results are discussed in more detail below). Specifically we improve the dependency on the approximation parameter $\delta$.
We then extend the concept of an approximate \dbs clustering
as defined to the hierarchical version. We thus obtain $\delta$-approximate \hdbs,
an approximate version of the \hdbs hierarchy of
Campello~\etal~\cite{cms-dbchde-13}, where the parameter $\delta$ specifies
the accuracy of the approximation. (Intuitively, a $\delta$-approximate \hdbs hierarchy
has the same clusters as the standard \hdbs hierarchy at any level~$\eps$,
except that clusters at distance $(1-\delta)\cdot \eps$ from each other may be merged.
See Section~\ref{se:approx} for  precise definition.)
We show that a $\delta$-approximate
\hdbs hierarchy in $\Reals^d$ can be computed in $O((n/\delta^{(d-1)/2})\log n)$ time.

\mypara{Further related work.}
This work should be viewed as the journal publication of the 
(so far unpublished) masters thesis of the second author~\cite{g-fadbs-13}, 
which contained the results on \dbs, extended with results on \hdbs.
In the meantime, Gan and Tao~\cite{gt-dbs-15}
published a paper in which they extend the work from the masters thesis
to $\Reals^d$, resulting in an algorithm for \dbs with a running time 
of $O(n^{2-\frac{2}{\lceil d/2\rceil +1} +\gamma})$;
we briefly comment on how this is done at the end of Section~\ref{se:dbs}.
Gan and Tao also prove that computing the \dbs clustering in $\Reals^d$
for $d\geq 3$ is at least as hard as the so-called unit-spherical emptiness problem,
which is believed to require $\Omega(n^{4/3})$ time~\cite{e-rcgp-95}.
Finally, Gan and Tao show that a $\delta$-approximate \dbs clustering
can be computed in $O(n/\delta^{d-1})$ expected time, using a modified version of the exact algorithm.
Their approximate clustering is the
same as the approximate clustering defined by Chen~\etal~\cite{csx-gadbc-05},
who already showed how to compute it in $O(n\log n + n/\delta^{d-1})$
time deterministically.
(Gan and Tao were unaware of the paper by Chen~\etal.)
As we show in Section~\ref{se:approx} our algorithm can also be used
to obtain a deterministic algorithm with $O(n\log n + n/\delta^{d/3})$ running time.

\section{Preliminaries on \dbs and \dbss} \label{se:preli}
Let $D$ be a set of points in $\Reals^d$.
\dbs distinguishes three types of points:
core points (points in the ``interior'' of a cluster),
border points (points on the boundary of a cluster),
and noise (points not in any cluster).
The distinction is based on two global
parameters, $\eps$ and $\minp$.
Define $\neps(p,D) \mydef \{ q \in D : |pq| \leq \eps \}$
to be the \emph{neighborhood} of a point~$p$,
where $|pq|$ denotes the (Euclidean) distance between~$p$~and~$q$;
the point $p$ itself is included in~$\neps(p,D)$.
A point $p\in D$ is a \emph{core point} if $|\neps(p,D)|\geq \minp$,
and a non-core point $q$ in
the neighborhood of a core point is a \emph{border point}.
We denote the set of core points by $\dcore$, and the set of border points by $\dborder$.
The remaining points are \emph{noise}.
In \dbss~\cite{cms-dbchde-13}
border points are not part of a cluster but are considered noise.

Ester~\etal~\cite{eksx-dbadc-96} define the \dbs clusters
based on the concept of density-reachability (a detailed description is given below).
Equivalently, we can
define the clusters as the connected components of a certain graph.
\begin{figure}[bt]
\begin{center}
\includegraphics{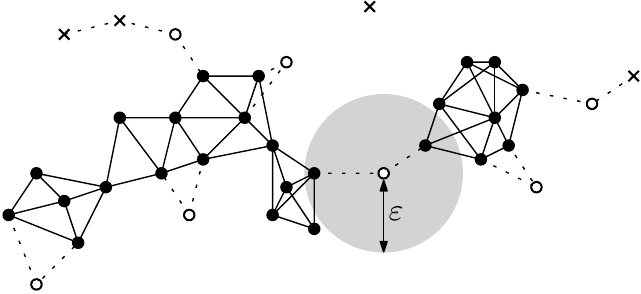}
\end{center}
\caption{A neighborhood graph (with $\minp=4$ and $\eps$ as indicated).
         Solid disks are core points, open circles
         are border points, and crosses are noise.
         Edges between core points are solid, other edges are dotted.
         The solid disks and edges form the core graph.}
\label{fi:core-graph}
\end{figure}
To this end, define the \emph{neighborhood graph} $\gnb(D,\enb)$
as the (undirected) graph  with node set $D$ and
edges connecting pairs of points within distance~$\eps$; see Fig.~\ref{fi:core-graph}.
In other words,
\[
\enb = \{ \ (p,q)\in D\times D \; : \; q\in \neps(p,D) \setminus\{p\} \ \}.
\]
Note that a point $p\in D$ is a core point if and only if its degree in $\gnb$
is at least $\minp-1$, since then its neighborhood contains at least $\minp$ points
(including $p$ itself).
Now consider the subgraph $\gcore(\dcore,\ecore)$
induced by the core points, that is, $\gcore$ is the graph whose nodes
are the core points and whose edges connect two
core points when they are within distance $\eps$ from each other.
We call $\gcore$ the \emph{core graph}.
The connected components of $\gcore$ are the clusters in \dbss.
The clusters in \dbs are the same, except that they also contain border points.
Formally, a border point~$q$ belongs to a cluster $C$ if $q$
has an edge (in $\gnb$) to a core point $p\in C$. Thus a border
point can belong to multiple clusters.
The original \dbs algorithm assigns
a border point~$p$ to the first cluster that finds $p$ (clusters are constructed one by one);
we assign border points to the cluster of their nearest core point.

\subsection{The original definition of the \dbs clustering.}
For comparison purposes only, we restate the original definition of the \dbs clustering.
Ester~\etal~\cite{eksx-dbadc-96} define when two points are in the same cluster
based on the concept of density-reachability, as explained next.
A point $q\in D$ is \emph{directly density-reachable} from a point $p\in D$
if $q\in \neps(p,D)$ and $p$ is a core point. We denote this by $p\rightarrow q$.
A core point is always directly density-reachable from itself,
since a point $p\in D$ is always in its own neighborhood.
A point $q$ is \emph{density-reachable} from a core point $p$, denoted by $p\leadsto q$,
if there is a sequence $p=r_0,\ldots,r_k=q$ (for some $k\geq 0$)
such that $r_0\rightarrow r_1 \rightarrow \cdots \rightarrow r_k$.
Observe that if two points $p$ and $q$ are both core points,
then $p\leadsto q$ if and only if $q\leadsto p$.
Two points $p$ and $q$ are \emph{density-connected} if there is a point~$r$
such that $r\leadsto p$ and $r\leadsto q$.

A cluster is now defined as
a subset~$C\subseteq D$ such that (i) if a core point $p$ is in $C$
then all points $q$ that are reachable from $p$ are in $C$, and
(ii) any two points in $C$ are density-connected to each other.
As observed by Ester~\etal, a cluster must contain at least one core point,
and a cluster is uniquely defined by any of its core points.
More precisely, if $p\in C$ is a core point then $C = \{ q\in D : p\leadsto q\}$.
Each core point belongs to exactly one \dbs cluster.
Under the above definition border points can belong to multiple clusters, however.
This is typically undesirable, so the original \dbs algorithm
assigns each border point~$q$ to only one cluster, namely the cluster from
which $q$ is discovered first by their algorithm. This implies that the computed clustering
depends on the order in which their algorithm happens to handle the points.

\section{A fast algorithm for \dbs} \label{se:dbs}
The original \dbs algorithm reports, while generating
and exploring the clusters, for each point $p\in D$ all
its neighbors. In other words, it spends time on every edge
in the neighborhood graph.
Our new algorithm avoids this by
working with a smaller graph, the \emph{box graph}~$\gbox$.
Its nodes are disjoint rectangular boxes with a diameter of at most $\eps$
that together contain all the points in~$D$, and its edges connect pairs of boxes within distance~$\eps$; see Fig.~\ref{fi:box-graph}.
\begin{figure}[b]
\begin{center}
\includegraphics{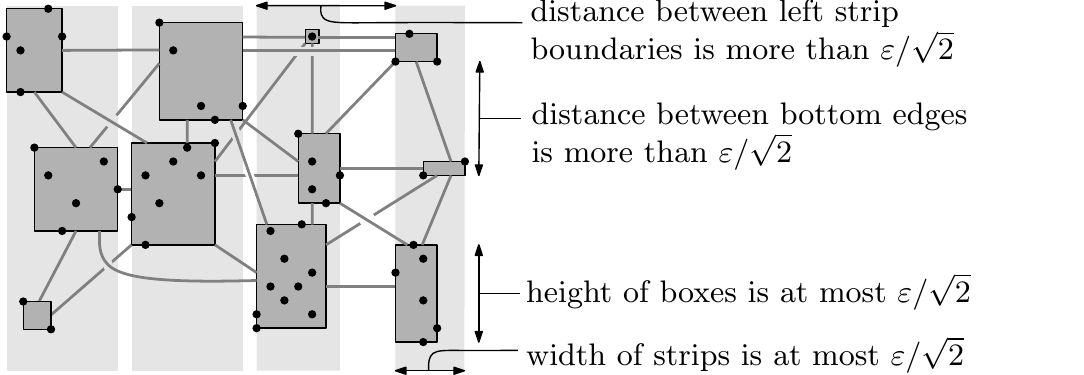}
\end{center}
\caption{Example of a box graph.}
\label{fi:box-graph}
\end{figure}
The boxes are generated such that (i) any two points in the same
box are in each other's neighborhood, and
(ii) the degree of any node in the box graph is~$O(1)$.
Property~(i) allows us to immediately
classify all points in a box as core points when it contains
at least $\minp$ points, and property~(ii) allows us to quickly
retrieve the neighbors of any given point in a box.
Next we describe the algorithm, which consists of four easy steps, in detail.

\mypara{Step~1: Compute the box graph~$\gbox$.}
To compute~$\gbox$, we first construct
a collection of vertical \emph{strips} that together cover all the
points. Let $p_1,\ldots,p_n$ be the points in $D$ sorted by $x$-coordinate, with ties broken arbitrarily.
The first strip has $p_1$ on its left boundary.
We go through the remaining points from left to right, adding them to
the first strip as we go, until we encounter a point $p_i$ whose distance
to the left strip boundary is more than~$\eps/\sqrt{2}$.
We then start a new strip with $p_i$ on its left boundary, and we
add points to that strip until we encounter a point
whose distance to the left strip boundary is more than~$\eps/\sqrt{2}$,
and so on, until we have handled all the points. Constructing
the strips takes $O(n)$ time, after sorting the points by $x$-coordinate.

Within each strip we perform a similar procedure,
going over the points within the strip in order of increasing $y$-coordinate
and creating boxes instead of strips.
Thus the first box in the strip has the lowest point
on its bottom edge, and we keep adding points to this box
(enlarging it so that the new point fits, ensuring a tight bounding box)
until we encounter a point whose vertical distance to the bottom edge
is more than~$\eps/\sqrt{2}$. We then start a new box, and so on,
until we handled all points in the strip.
If the number of points in
the $j$-th strip is $n_j$, then the time needed to handle all
the strips is $\sum_j O(n_j \log n_j) = O(n \log n)$.

Let $m$ be the number of strips and $\B_j$ the set of boxes
in the $j$-th strip. We sometimes refer to a set $\B_j$ as a strip,
even though formally $\B_j$ is a set of boxes.
Let $\B := \B_1\cup\cdots\cup\B_m$.
The nodes of the box graph~$\gbox$ are the boxes in $\B$
and there is an edge~$(b,b')$ when $\dist(b,b')\leq\eps$, where
$\dist(b,b')$ denote the minimum distance between $b$ and $b'$.
Two boxes $b,b'$ are \emph{neighbors}
when they are connected by an edge in~$\gbox$.
Let $\neps(b,\B)$ be the set of neighbors~$b$.
\begin{lemma}\label{le:box-graph-degree}
$\gbox$ has at most~$n$ nodes, each having $O(1)$~neighbors.
\end{lemma}
\begin{proof}
The box graph obviously has at most~$n$ nodes. Next we give a precise analysis
of the number of neighbors a box~$b$ in some strip~$ B_j$ can have.
Consider the node corresponding to the box~$b$.  If there are two
or more strips in between $b$ and some other box $b'$ then
$\dist(b,b') > 2 (\eps/\sqrt{2}) > \eps$, so $b$ can only have neighbors in
$\B_{j-2}$, $\B_{j-1}$, $\B_j$, $\B_{j+1}$, or~$\B_{j+2}$.
We bound the number of neighbors in each of these strips separately.
\myitemize{
\item
  The number of neighbors in $\B_j$ is at most four. Indeed,
  the boxes in $\B_j$ can be ordered vertically, and if there are
  more than two boxes in between $b$ and some other box~$b'\in\B_j$,
  then the vertical distance between $b$ and $b'$ is more
  than~$\eps$.
\item
  The number of neighbors in $\B_{j-1}$ (and similarly
  in $\B_{j+1}$) is at most five. Suppose for a contradiction that~$b$
  has six or more neighbors in~$B_{j-1}$. Let $b'$ and $b''$ be the
  lowest and highest of these neighbors, respectively.
  Then the vertical distance between the top boundary of~$b'$
  and the bottom boundary of~$b''$is more than~$4\eps/\sqrt{2}$.
  Since the height of~$b$ is at most $\eps/\sqrt{2}$, this implies
  that the vertical distance between $b$ and either $b'$ or $b''$
  is more than $(4\eps/\sqrt{2} -  \eps/\sqrt{2})/2 > \eps$,
  contradicting that both $b'$ and $b''$ are neighbors of~$b$.
\item
  The number of neighbors in $\B_{j-2}$ (and similarly
  in $\B_{j+2}$) is at most four. Suppose for a contradiction that~$b$
  has five or more neighbors in~$B_{j-2}$. Let $b'$ and $b''$ be the
  lowest and highest of these neighbors, respectively.
  Then the vertical distance between the top boundary of~$b'$
  and the bottom boundary of~$b''$ is more than~$3\eps/\sqrt{2}$.
  Since the height of~$b$ is at most $\eps/\sqrt{2}$, this implies
  the vertical distance between $b$ and either $b'$ or $b''$ is
  more than $(3\eps/\sqrt{2} -  \eps/\sqrt{2})/2 = \eps/\sqrt{2}$.
  The horizontal distance is at least $\eps/\sqrt{2}$ as well,
  because there is a strip in between. Hence, the total distance
  between $b$ and either $b'$ or $b''$ is more than~$\eps$,
  contradicting that both $b'$ and $b''$ are neighbors of~$b$.
}
Adding up the maximum number of neighbors in each of the strips gives
us a maximum of~22 neighbors in total.
\end{proof}

This also gives us an easy way
to compute the edge set $\ebox$ of the box graph, because the edges
between boxes in strips $\B_j$ and $\B_{j'}$ with $|j-j'|\leq 2$
can be computed in $O(|\B_j|+|\B_{j'}|)$ time in total by scanning the
boxes in $\B_{j}$ and $\B_{j'}$ in a coordinated manner.
The total time to compute all edges of the box graph is thus
$$
O\left( \sum_{j=1}^m \sum_{j'=\max(j-2,1)}^{\min(j+2,m)} \left(|\B_j|+|\B_{j'}|\right) \right)
 =
O\left(\sum_{j=1}^m |\B_j|  \right)
 =
O(n).
$$
Adding the time to construct the strips and boxes, we
see that Step~1 takes $O(n\log n)$ time and we
obtain the following lemma.
\begin{lemma}\label{le:box-graph-comp}
The box graph~$\gbox(\B,\ebox)$ can be computed in $O(n\log n)$ time.
\end{lemma}

\myitpara{An alternative for Step~1.}
An alternative approach is to define the boxes as the non-empty cells
in a grid whose cells have height and width~$\eps / \sqrt{2}$. If we store
the boxes in a hash-table based on the coordinates of their lower left corners,
then finding the neighbors of a box~$b$ can be done by checking each potential
neighbor cell for existence in the hash-table---we do not need to store
the box graph explicitly. Creating the boxes (with their corresponding point sets)
can be done in $O(n)$ time if the floor function can be computed in $O(1)$ time.

\mypara{Step~2: Find the core points.}
The graph $\gbox$ allows us to determine the core points
in a simple and efficient manner. The key observation is that the maximum
distance between any two points in the same box is at most~$\eps$.
Hence, if a box contains more than $\minp$ points, then all of them
are core points. Thus the following simple algorithm suffices to determine
the core points.

For a box $b\in \B$, let $D(b) := D\cap b$ be the set
of point inside $b$, and let $n_b := |D(b)|$.
If $n_b \geq \minp$ then label all points in $b$ as core points.
Otherwise, for each point $p\in D(b)$, count the number of points $q$ in neighboring
boxes of $b$ for which $|pq|\leq\eps$.
If this number is at least $\minp-n_b$, then label~$p$ as core point.
The counting is done brute-force, by checking all points
in neighboring boxes. Hence, this takes $O(\sum_{b'\in \neps(b,\B)} n_{b'})$ time
for each point $p\in b$.
%
\begin{lemma}\label{le:find-core}
Given~$\gbox$, we can find all core points in $D$ in $O(n)$ time.
\end{lemma}
\begin{proof}
The total time spent to handle boxes~$b$ with $n_b\geq \minp$ is clearly $O(n)$.
The time needed to handle a box~$b$ with $n_b< \minp$ is
\[
O\left( n_b \cdot \sum_{b'\in\neps(b,\B)} n_{b'} \right)
   =
O\left( \minp \cdot \sum_{b'\in\neps(b,\B)} n_{b'} \right).
\]
Now charge $O(\minp)=O(1)$ time to each point
in every $b'\in\neps(b,\B)$. Because any box~$b'$
is the neighbor of $O(1)$ other boxes by Lemma~\ref{le:box-graph-degree},
each point is charged $O(1)$ times, proving the lemma.
\end{proof}

\mypara{Step~3: Compute the cluster cores.}
In the previous step, we determined the core points. Next we wish to determine
the clusters or, more precisely the cluster cores.
The \emph{core of a cluster} is the set of core points in that cluster.
In Step~3 we assign to each core point a \emph{cluster-id}
so that core points in the same cluster have the same cluster-id.
Again, this can be done in an efficient manner using~$\gbox$.
To this end, we first remove certain boxes and edges from~$\gbox$ to obtain
a reduced box graph~$\gbox^*$. More precisely, we keep only
the boxes with at least one core point, and we keep only the edges $(b,b')$
for which there are core points $p\in b$, $p'\in b'$ with $|pp'|\leq \eps$.
Because any two core points in a given box~$b$ are connected in $\gcore$,
we have the following lemma.
\begin{lemma}
The connected components in $\gbox^*$ correspond one-to-one to the connected
components in the core graph $\gcore$ and, hence, to the \dbss clusters.
\end{lemma}
Thus the cluster cores can be computed by
computing the connected components in $\gbox^*$. The latter can be
done in $O(n)$ time using~DFS~\cite{clrs-ia-09}.
After computing the connected components we give every core point~$p$
a cluster-id corresponding to the connected component of the box~$b$ that contains~$p$.

To construct $\gbox^*$, we need to decide for two given boxes
$b,b'$ whether there are core points $p\in D(b)$, $p'\in D(b')$ with $|p p'|\leq \eps$.
If  $n_b<\minp$ or $n_{b'}<\minp$
then we simply check all pairs of core points.
We can argue as in the proof of Lemma~\ref{le:find-core}
that this takes $O(\minp \cdot n) = O(n)$ time in total.
If $n_b\geq\minp$ and $n_{b'}\geq\minp$ we
have to be more careful,since checking all pairs of points can lead to a
quadratic running time. For example, if both cells contain $n/2$ points,
then checking all pairs would take $\Omega(n^2)$ time.
When $n_b\geq\minp$ and $n_{b'}\geq\minp$ we therefore
compute the Delaunay triangulation of $D(b)\cup D(b')$ in
$O((n_{b}+n_{b'})\log (n_{b}+n_{b'}))$ time~\cite{bcko-cgaa-08},
and check whether it has an edge $(p,p')$
of length at most $\eps$ between points $p\in D(b), p'\in D(b')$.
This is sufficient because the pair of points defining the closest distance
between $D(b)$ and $D(b')$
must define an edge in the Delaunay triangulation.
This leads to the following lemma.
\begin{lemma}
Computing the cluster cores can be done in $O(n\log n)$ time.
\end{lemma}
\begin{proof}
The most time consuming part of the construction of $\gbox^*$ is
to determine for each pair of neighboring boxes in $\B$
whether there are core points $p\in b$, $p'\in b'$ with $|p p'|\leq \eps$.
As mentioned, the total time spent on pairs with $n_b<\minp$
or $n_{b'}<\minp$ is $O(n)$. Let $\B^*$ be the set of boxes containing
at least $\minp$~points. Then the total time spent on the pairs of boxes from~$\B^*$
is
\[
\sum_{b\in {\B^*}}\sum_{b'\in \neps(b,\B^*)} O((n_{b}+n_{b'})\log (n_{b}+n_{b'})),
\]
which is $O(n\log n)$ because $|\neps(b,\B^*)|=O(1)$ for any box $b$
and $\sum_{b\in \B^*} n_b \leq n$.
\end{proof}

\myitpara{Remark.} In practice we can also
use a brute-force algorithm when $n_b\geq\minp$ and $n_{b'}\geq\minp$, because
the number of points in boxes with more than $\minp$ points
is typically still not very large. Moreover, if both $b$ and $b'$ contain many points, then
there are often many pairs of points within distance~$\eps$ from each other,
and we can stop when we find such a pair.

\mypara{Step 4: Assigning border points to clusters.}
It remains to decide for non-core points~$p$ whether $p$ is a border point or noise.
(For \dbss we can skip Step~4, since in \dbss border points are considered noise.)
If $p$ is a border point, it has to be assigned to the
nearest cluster. Again, a brute-force method suffices:
for each box~$b\in B$ and each non-core point $p\in b$,
we check all points in $b$ and its neighboring boxes to find
$p$'s nearest core point, $p'$. If $|pp'|\leq\eps$,
then $p$ is a border point in the same cluster as~$p'$,
otherwise $p$ is noise. We only need to consider boxes $b$ with
$n_b <\minp$---otherwise all points in $b$ are core points---so the
argument from the proof of Lemma~\ref{le:find-core} shows
that this takes $O(n)$ time.

\mypara{Putting it all together.}
Steps~1 and~3  take $O(n\log n)$ time
and Steps~2 and~4 take $O(n)$ time. We thus obtain the following theorem.
\begin{theorem}\label{th:dbs-main}
Let $D$ be a set of $n$ points in $\Reals^2$, and $\eps$ and $\minp$ be given
constants. Then we can compute a \dbs clustering on $D$ according to $\eps$ and
$\minp$ for the Euclidean metric in $O(n \log n)$ time.
\end{theorem}

\myitpara{Remark: extension to higher dimensions.}
The algorithm just described can easily be extended to $\Reals^d$ for~$d>2$,
as already observed by Gan and Tao~\cite{gt-dbs-15}. For completeness we
describe the extension and the resulting bound. One trivial modification is
that in $\Reals^d$ we need to use boxes in~$\gbox$ of width at most $\eps/\sqrt{d}$ along each axis
to ensure their diameter is at most~$\eps$. The only other difference is
in Step~3, where we have to decide for all pairs~$b,b'$ of neighboring
boxes whether there are core points $p\in D(b)$ and $p'\in D(b')$
with $|p p'|\leq \eps$. When $n_b\geq\minp$ and $n_{b'}\geq\minp$
we can no longer use the Delaunay triangulation for this, as it may have quadratic size.
Instead we can use a known algorithm for bichromatic closest pair~\cite{aes-emstbcp-91},
which gives a running time of $O(n^{2-\frac{2}{\lceil d/2\rceil +1} +\gamma})$,
where $\gamma>0$ is an arbitrarily small constant. (For $d=3$ the $n^{\gamma}$
factor can be replaced by a polylogarithmic factor.)

\section{A fast algorithm for \hdbs in the plane} \label{se:hdbs}
Campello~\etal~\cite{cms-dbchde-13} introduced \hdbs, a hierarchical
version of \dbss (similar to \optics~\cite{abks-optics-99}).
The algorithm described by Campello~\etal to compute the \hdbs hierarchy runs in quadratic time.
We show that in $\Reals^2$ and under the Euclidean metric, the \hdbs hierarchy can be computed
in $O(n \log n)$ time.

\mypara{Preliminaries on \hdbs.}
Recall that \dbss is the version of \dbs in which border points are considered noise.
The \hdbs hierarchy is a tree structure encoding the clusterings of \dbss
(for a fixed~\minp) that arise as $\eps$ increases from $\eps=0$ to
$\eps=\infty$.
Initially, when $\eps=0$, all points are noise.
As~$\eps$ increases, three types of events can happen to the \dbss clustering:
\myitemize{
\item \emph{Type~(i): the status of a point changes.} In this event,
       a point changes from being noise to being a core point.
       The value of $\eps$ at which this happens
       for a point~$p$ is called the \emph{core distance} of $p$; we denote it by~$\cdist(p)$.
\item \emph{Type~(ii): a new cluster starts.}
            This event is triggered by a type~(i) event, when
            a point becoming a core point forms a new (singleton) cluster. 
\item \emph{Type~(iii): two clusters merge.}
            This event can be triggered by a type~(i) event or it can
            happen when $\eps = |pq|$ for core points $p,q$ from different
            clusters.
}
Note that all events happen at values of $\eps$
such that $\eps=|pq|$ for some pair of points $p,q\in D$.
Also note that several events may happen simultaneously, not only
when two pairwise distances are equal, but also when an event triggers other events.
This process can be modeled as a \emph{dendrogram}:
a tree whose leaves correspond to
the points in~$D$ and whose nodes correspond to clusters arising during the process.
This dendrogram, where each node stores the value of $\eps$ at which the corresponding
cluster was created, is the \hdbs hierarchy.
 Notice that we can easily extract the \dbss clustering for any desired value of~$\eps$
 (with respect to the given $\minp$) from the dendrogram in linear time.
Campello~\etal compute the \hdbs hierarchy as follows.

For two points $p,q\in D$, define
$
\mdist(p,q) := \max \left(\cdist(p),\cdist(q), |pq| \right)
$
to be the \emph{mutual reachability distance} of $p$ and $q$.
The \emph{mutual reachability graph} $\Gmr$ is defined as the complete
graph with node set~$D$ in which each edge $(p,q)$
has weight~$\mdist(p,q)$. Campello~\etal observe that \hdbs hierarchy
can easily be computed from a minimum spanning tree  (\mst) on $\Gmr$.
(Indeed, the cluster-growing process
corresponds to the process of computing an \mst on $\Gmr$
using Kruskal's algorithm~\cite{clrs-ia-09}.)
Hence, they compute the \hdbs hierarchy as follows.
\myenumerate{
\item \label{step:coredist} Compute the core distances $\cdist(p)$ for all points $p\in D$.
\item \label{step:mst} Compute an \mst $\tree$ of the mutual reachability graph~$\Gmr$.
\item \label{step:convert}
      Convert~$\tree$ into a dendrogram where each internal node stores the
      value of~$\eps$ at which the corresponding cluster is formed.
}
\mypara{Our planar algorithm.}
The most time-consuming parts in the algorithm above
are Steps~\ref{step:coredist}~and~\ref{step:mst}; Step~\ref{step:convert}
takes $O(n)$ time after sorting the edges of $\tree$ by weight.

For Step~\ref{step:coredist} we observe that $\cdist(p)$ is the distance of point~$p$ to
its $\ell$-th nearest neighbor for $\ell=\minp-1$. Hence, to compute all
core distances it suffices to compute for each point its $k$ nearest
neighbors. This can be done (in any fixed dimension) in $O(n\ell\log n)$ time~\cite{v-annp-89}.
Since $\ell=\minp-1=O(1)$ this implies that Step~\ref{step:coredist}
takes $O(n\log n)$ time.

Step~\ref{step:mst} is more difficult to do in subquadratic time.
The main problem is that we cannot afford to look at all edges of~$\Gmr$
when computing~$\tree$. To overcome this problem we need the following generalization
of Delaunay triangulations, introduced by Gudmundsson~\etal~\cite{ghk-hodt-02}.
Recall that a pair of points $p,q\in D$ forms an edge
in the Delaunay triangulation of~$D$ if and only if there is a circle with $p$ and
$q$ on its boundary and no points from~$D$ in its interior~\cite{bcko-cgaa-08}.
We say that the pair $p,q\in D$ forms a
\emph{$k$-th order Delaunay edge}, or \emph{$k$-OD edge} for short,
if and only if there exists a circle with $p$ and $q$ on its boundary
and at most $k$ points from $D$ in its interior~\cite{ghk-hodt-02}.
(Thus the 0-OD edges are precisely the edges of the Delaunay triangulation.)
The $k$-OD edges are useful for us because of the following lemma.
\begin{lemma}\label{le:kod-lemma}
Let $\GmrRed$ be the subgraph of $\Gmr$ that contains only the $k$-OD edges,
where $k=\max(\minp-3,0)$. Then an \mst of $\GmrRed$ is also an \mst of $\Gmr$.
\end{lemma}
\begin{proof}
Imagine computing an \mst $\tree$ on $\Gmr$ using Kruskal's algorithm~\cite{clrs-ia-09}. This algorithm treats the edges $(p,q)$ of $\Gmr$
in order of increasing weight, that is, increasing values of $\mdist(p,q)$.
When it processes $(p,q)$ it checks if $p$ and $q$ are already in the
same connected component---in our application this component corresponds to a cluster
at the current value of $\eps$---and, if not, merges these components.
We will argue that whenever we process an edge~$(p,q)$ that is not in $\GmrRed$,
that is, an edge that is not a $k$-OD edge, then $p$ and $q$ are already in the same connected
component. Hence, there is no need to process $(p,q)$, which proves that
an \mst of $\GmrRed$ is also an \mst of $\Gmr$.

Let $C_{pq}$ be the circle such that $p$ and $q$ form a diametrical pair of~$C$,
and let $D(C_{pq})\subset D$ be the set of points lying in the interior of~$C_{pq}$.
If $|D(C_{pq})|\leq k$, then $(p,q)$ is a $k$-OD edge, so assume $|D(C_{pq})|\geq k+1$.
Note that $\cdist(r)<|pq|$ for all $r\in D(C_{pq})$.
Indeed, since $p,q$ is a diametrical pair of $C_{pq}$,
the distance from $r$ to any other point in~$C_{pq}$ (including $p$ and $q$)
is smaller than~$|pq|$. Hence, for $\eps=|pq|$ we have
$|\neps(r,D)| \geq |D(C_{pq})| + 2 = k+3 = \minp$. Thus all points
$r\in C_{pq}$ are core points when we process $(p,q)$.
Moreover, for all edges $(s,t)$ with $s,t \in D(C_{pq})\cup \{p,q\}$
we have $\mdist(s,t) \leq |pq|$.
Hence, it suffices to prove the following.
\\[2mm]
\noindent \emph{Claim:}
Let $C$ be a circle with two points $p,q$ on its boundary
and let $D(C)\subset D$ be the set of points from $D$ in the interior of~$C$.
Then there is a path from $p$ to $q$ in $\GmrRed$ that uses only
points in $D(C)\cup \{p,q\}$.
\\[2mm]
We prove this claim by induction on~$|D(C)|$. If $|D(C)|\leq k$
then $(p,q)$ is a $k$-OD edge itself and we are done. Otherwise,
pick any point $r\in D(C)$. Now shrink $C$, while keeping $p$
in its boundary, until we obtain a circle~$C_1$ that also has $r$ on its boundary, as shown in Figure~\ref{fig:k-delaunay}.
By induction, there is a path from $p$ to $r$ in $\GmrRed$ that
uses only points in $D(C_1)\cup \{p,r\} \subset D(C)\cup \{p,q\}$.
A similar argument shows that there is a path from $r$ to $q$ that
uses only points in $D(C)\cup \{p,q\}$.
This proves the claim and, hence, the lemma.
\end{proof}
\begin{figure}
\centering
\includegraphics{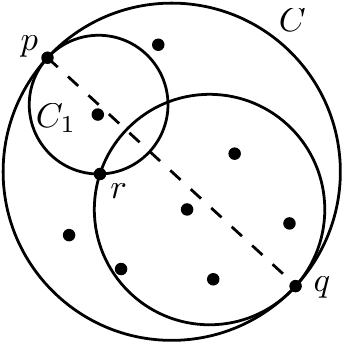}
\caption{Illustration of the recursive argument in the proof of Lemma~\ref{le:kod-lemma}.}
\label{fig:k-delaunay}
\end{figure}
Gudmundsson~\etal showed that the number of $k$-OD edges is $O(n(k+1))$
and that the set of all $k$-OD edges can be computed
in $O(n(k+1)\log n)$ time with a randomized incremental algorithm.
Lemma~\ref{le:kod-lemma} implies that after computing the core distances and
the $k$-OD edges in $O(n\log n)$ time---recall that $k=\max(\minp-3,0)=O(1)$---we
can compute the \mst for $\Gmr$ by considering only $O(n)$ edges.
Thus computing the \mst can be done in $O(n\log n)$ time~\cite{clrs-ia-09}.
Since the rest of the algorithm takes linear time, we obtain the following theorem.
\begin{theorem}
Let $D$ be a set of $n$ points in $\Reals^2$ and $\minp$ be a given
constant. We can compute the \hdbs hierarchy on $D$ for the Euclidean metric
with a randomized algorithm in $O(n \log n)$ expected time.
\end{theorem}

\section{Approximate \dbss and \hdbs} \label{se:approx}
 The approach from the previous section for computing the \hdbs hierarchy
 will not give a subquadratic bound in higher dimensions,
 since the number of Delaunay edges can be quadratic. The running
 time of the algorithm from Section~\ref{se:dbs} to compute a single
 \dbs clustering is subquadratic in any fixed dimension,
 but the running time quickly deteriorates as the dimension increases.
In this section we introduce an approximate version of \dbss and \hdbs, both of which can be computed in linear time. An approximation for \dbss that runs in expected linear time was already provided by Chen~\etal~\cite{csx-gadbc-05} Gan and Tao~\cite{gt-dbs-15}, however our approach runs has a slightly better dependency on the approximation factor $\delta$. To the best of our knowledge this is the first linear time approximation algorithm for $\hdbs$.

\subsection{Approximate \dbss}\label{se:approx-dbscan}
First we define what exactly is an approximate \dbss clustering.
Our definition of approximate \dbss is essentially the same as the definitions
by Chen~\etal~\cite{csx-gadbc-05} and Gan and Tao~\cite{gt-dbs-15}.
The main difference is that we base our definition on \dbss instead
of \dbs, which avoids some technical difficulties in the definition.

Let $\minp$ be a fixed constant. Let $\epsclusters{\eps}$
denote the set of clusters in the \dbss clustering for a given value of $\eps$.
We call a clustering $\C_1$ a \emph{refinement} of a clustering $\C_2$,
denoted by $\C_1\refine \C_2$, when
for every cluster $C_1\in\C_1$ there is a cluster $C_2\in \C_2$ with $C_1\subseteq C_2$.
Recall that, as $\eps$ increases, the \dbss clusters merge or expand and new singleton clusters may appear, but clusters do not shrink or disappear.
Hence, if $\eps<\eps'$ then\footnote{Here it is important that we consider
  \dbss and not \dbs. Indeed, in \dbs border points can ``flip'' between clusters
  as $\eps$ increases, and so we do not necessarily have $\epsclusters{\eps}\refine \epsclusters{\eps'}$.}
$\epsclusters{\eps}\refine \epsclusters{\eps'}$.
An approximate \dbss clustering is now defined as follows.
\begin{definition}
A \emph{$\delta$-approximate \dbss clustering} of a data set $D$, for given
parameters $\eps$ and $\minp$, and a given error $\delta>0$, is now defined as
a clustering $\C^*$ of $D$ into clusters and noise such that
$\epsclusters{(1-\delta)\eps} \refine \C^* \refine \epsclusters{\eps}$.
\end{definition}
Thus if we choose $\delta$ sufficiently small, then a $\delta$-approximate \dbss clustering
is very similar to the exact \dbss clustering for the given parameter values.

\mypara{The algorithm.}
As mentioned in the introduction, both Chen~\etal~\cite{csx-gadbc-05} and
Gan and Tao~\cite{gt-dbs-15} already presented algorithms for this.
We obtain a slightly better dependency on $\delta$ than Gan and Tao~\cite{gt-dbs-15}
by plugging in a better algorithm for approximate bichromatic closest pair.

Recall that the bottleneck in computing a \dbss clustering lies in checking,
for pairs $b,b'$ of neighboring boxes, whether there is a pair $(p,p')\in D(b)\times D(b')$
with $|p p'|\leq \eps$. We can perform this check
approximately by computing an \emph{approximate bichromatic closest pair}
$(p,p')\in D(b)\times D(b')$
such that $|p p'|\leq (1+\alpha) \cdot \dist(D(b),D(b'))$ for
$\alpha=\delta/(1-\delta)$, where
$\dist(D(b),D(b'))$ denotes the distance between the points of the true closest pair.
This can be done in $O((1/\alpha)^{d/3}(n_b+n_{b'}))=O((1/\delta)^{d/3}(n_b+n_{b'}))$
time~\cite{ac-bed-14}. We add the edge $(b,b')$ to the box graph~$\gbox$ when $|p p'|\leq \eps$.
This way we can obtain the following result.
\begin{theorem}
Let $D$ be a set of $n$ points in $\Reals^d$, and let $\eps$ and $\minp$ be given
constants. Then, for any given $\delta>0$, we can compute a $\delta$-approximate
\dbss clustering on $D$ with respect to $\eps$ and $\minp$ for the Euclidean metric
in $O(n\log n + (1/\delta)^{d/3}n)$ time.
\end{theorem}
\begin{proof}
Let $\C$ be the clustering computed using the approximate bichromatic closest pairs.
When $\dist(D(b),D(b')\leq \eps$, then the reported approximate bichromatic closest pair
$(p,p')$ has $|p p'|\leq \eps$, so the set of edges added to the box graph
is a subset of the edge set of the actual box graph at the given value of~$\eps$.
Hence, $\C \refine \epsclusters{\eps}$. On the other
hand, when $\dist(D(b),D(b'))\leq (1-\delta)\eps$ then the approximate closest
pair $(p,p')$ has $|p p'|\leq (1+\alpha)(1-\delta)\eps = \eps$.
Hence, we are guaranteed to add all edges of the box graph for $(1-\delta)\eps$
and so $\epsclusters{(1-\delta)\eps}\refine \C$.
Thus $\C$ is a $\delta$-approximate \dbss clustering.

Recall that the running time of our algorithm is $O(n\log n)$, plus
the time needed to compute the edges of the box graph. In the approximate
version the latter step takes time
\[
\sum_{b\in {\B^*}}\sum_{b'\in \neps(b,\B^*)} O((1/\delta)^{d/3}\cdot(n_{b}+n_{b'}))
= O((1/\delta)^{d/3}n).
\]
\end{proof}

\mypara{Approximate \hdbs.}
Our definition of an approximate \hdbs hierarchy is based on the definition
of $\delta$-approximate \dbss clusterings: we say that a hierarchy is
a \emph{$\delta$-approximate \hdbs hierarchy} if, for any value of $\eps$,
the clustering extracted from the hierarchy is a $\delta$-approximate \dbss clustering
for that value of~$\eps$. Next we show how to compute a $\delta$-approximate \hdbs hierarchy
in $O(n\log n)$ time, in any fixed dimension.
\medskip

As in Section~\ref{se:hdbs} we follow the algorithm by Campello~\etal~\cite{cms-dbchde-13},
and we speed up Step~\ref{step:mst} of the algorithm by computing an \mst on a subgraph of the
mutual reachability graph $\Gmr$ rather than on the whole graph.
(Steps~\ref{step:coredist} and~\ref{step:convert} can still be done
in $O(n\log n)$ and $O(n)$ time, respectively.) The difference with the
exact algorithm of Section~\ref{se:hdbs} is that we will select the edges of the subgraph
in a different manner, using ideas from so-called $\theta$-graphs~\cite{ns-gsn-07}.

Let $p\in D$ be a point. We partition $\Reals^d$ into simplicial
cones with apex~$p$ and whose angular diameter is~$\theta$,
where $\theta$ will be specified later.
(The angular diameter of a cone~$\cone$ with apex~$p$ is the maximum angle
between any two vectors emanating from $p$ and inside~$\cone$.)
Let $\Gamma_p$ be the resulting collection of cones and
consider a cone~$\cone\in\Gamma_p$. Let $D(\cone)\subseteq D$ denote the set of points
inside~$\cone$. (If a point lies on the boundaries of several cones we
can assign it to one of these cones arbitrarily.) Pick a half-line $\ell_{\cone}$
with endpoint~$p$ that lies inside~$\cone$. A $\theta$-graph would
now be obtained by projecting all points from $D(\cone)$ orthogonally onto $\ell_{\cone}$,
and adding an edge from $p$ to the point closest to $p$ in this
projection, with ties broken arbitrarily. We do the same,
except that we add edges to the $k$ closest points for~$k:=2\cdot\minp-3$.
If $\cone$ contains fewer than $k$ points, we simply connect $p$
to all points in $D(\cone)$.
Doing this for all the cones $\cone\in\Gamma_p$ gives us a set $E_p$ of $O(k/\theta)=O(1/\theta)$
edges for point~$p$. Let $E(\theta) := \bigcup_{p\in D} E_p$.
The set $E(\theta)$ can be computed
by making a straightforward adaptation to the algorithm to compute
a $\theta$-graph in $\Reals^d$~\cite[Chapter 5]{ns-gsn-07},
leading to the following result.
\begin{lemma}
$E(\theta)$ has $O(n/\theta^{d-1})$ edges and can be computed in $O((n/\theta^{d-1})\log^{d-1}n)$ time.
\end{lemma}
The set $E(\theta)$, where $\theta$ is chosen such that $\cos \theta \geq 1-\delta$,
defines the subgraph $\GmrRed(\delta)$ on which we compute the \mst in
Step~\ref{step:mst}. Since $\cos \theta > 1-\theta^2/2$, we have
$\cos \theta \geq 1-\delta$ when $\theta :=\sqrt{2\delta}$.
Next we show that an \mst on $\GmrRed(\delta)$ defines a $\delta$-approximate
\hdbs clustering.
\begin{lemma}\label{le:approx-theta}
Let $\tree$ be an \mst of $\GmrRed(\delta)$ and let $\eps>0$.
Let $\C(\tree,\eps)$ be the clustering induced by $\tree$.
Then $\C$ is a $\delta$-approximate \dbss clustering for the given~$\eps$.
\end{lemma}
\begin{proof}
For a weighted graph~$\graph$ and threshold weight $\tau$,
let $\graph[\tau]$ denote the subgraph obtained by removing all edges of weight
greater than~$\tau$. In order to show that
$\C(\tree,\eps)\refine \epsclusters{\eps}$ we must show that
any connected component of $\tree[\eps]$ is contained in a connected
component of $\Gmr[\eps]$.
Since $\tree$ is a subgraph of $\Gmr$ this is obviously the case.
\medskip

Next we prove that $\epsclusters{(1-\delta)\eps} \refine \C(\tree,\eps)$.
For this we must prove that any connected component of $\Gmr[(1-\delta)\eps]$
is contained in a connected component of $\tree[\eps]$. Since $\tree$
is an \mst of $\GmrRed(\delta)$, the connected components of $\tree[\eps]$
are the same as the connected components of $\GmrRed(\delta)[\eps]$.
It thus suffices to show the following: for any edge~$(p,q)\in \Gmr[(1-\delta)\eps]$,
there is a path from $p$ to $q$ in $\GmrRed(\delta)[\eps]$.
We show this by induction on $|pq|$, similarly to the way
in which it is shown that a $\theta$-graph has a small dilation.
\begin{figure}[t]
\begin{center}
\includegraphics{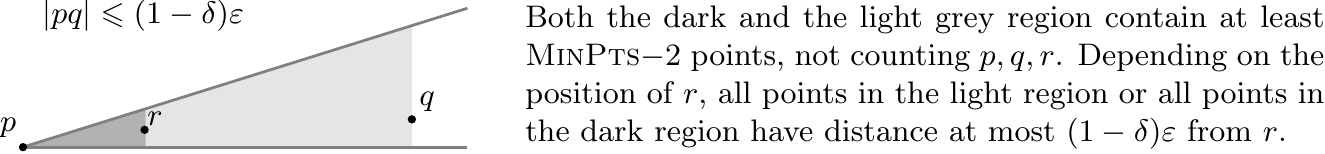}
\end{center}
\caption{Illustration for the proof of Lemma~\protect\ref{le:approx-theta}.}
\label{fi:approx-proof}
\end{figure}

Let $(p,q)$ be an edge in $\Gmr[(1-\delta)\eps]$.
Consider the set $\Gamma_p$ of cones with apex~$p$ that was used to define
the edge set $E_p$, and let $c\in \Gamma_p$ be the cone containing~$q$.
Recall that we added an edge from $p$ to the $k$ points in $c$ that
are closest to $p$ when projected onto the half-line~$\ell_c$, where
$k := 2\cdot\minp -3$. Hence, when $q$ is one of these $k$ closest
points we are done. Otherwise, let $r\in D(c)$ be the $(\minp-1)$-th closest point.
\\[2mm]
\emph{Claim:}
(i) $\cdist(r) \leq (1-\delta)\eps$,
(ii) $|pr| \leq \eps$, and
(iii) $|rq| < |pq|$.
\\[2mm]
Before we prove this claim, we first we argue that the claim allows us
to finish our inductive proof. Since $(p,q)$ is an edge
in $\Gmr[(1-\delta)\eps]$ we have $\mdist(p,q) \leq (1-\delta)\eps$.
Thus $|pq|\leq (1-\delta)\eps$ and $\cdist(q)\leq (1-\delta)\eps$.
Together with parts~(i)~and~(iii) of the claim this implies that $(r,q)$ is an edge in
$\Gmr[(1-\delta)\eps]$ with $|rq| < |pq|$.

In the base case of our inductive proof, where $(p,q)$ is the
shortest edge in~$\Gmr[(1-\delta)\eps]$, this cannot
occur. Thus $q$ must be one of the $k$ closest points in the
cone~$c$, and we have an edge between $p$ and $q$ in $\GmrRed(\delta)[\eps]$ by construction.

If we are not in the base case, then we have a path from
$r$ to $q$ in $\GmrRed(\delta)[\eps]$ by the induction hypothesis.
Moreover, $(p,r)$ is an edge in $\GmrRed(\delta)$ by construction.
Since $|pr|\leq \eps$ by part~(ii) of the claim, we have a path from $p$ to $q$
in $\GmrRed(\delta)[\eps]$.

It remains to prove the claim. For this we use the following
fact~\cite[Lemma 4.1.4]{ns-gsn-07}, which
is also used to prove that a $\theta$-graph has small dilation.
\\[2mm]
\emph{Fact:}
Let $s,t$ be any two points in a cone~$c\in \Gamma_p$ such that, when projected
onto the half-line~$\ell_c$, the distance from $p$ to $s$ is smaller than the
distance from $p$ to $t$. Then $|ps| \leq |pt|/\cos\theta$
and $|st| < |pt|$.
\\[2mm]
Part~(iii) of the claim immediately follows from this fact by taking $s:=r$ and $t:=q$.
Part~(ii) follows again by taking $s:=r$ and $t:=q$, using that
$|pq|\leq (1-\delta)\eps$ and that we have chosen $\delta$
such that $\cos\theta = 1-\delta$.
For part~(i) we must prove that there are at least $\minp-1$ points
within distance~$(1-\delta)\eps$ from~$r$.
Recall that $r$ is the $(\minp-1)$-th closest point to $p$
in the cone $c$, measured in the projection onto the half-line~$\ell_c$.
Let $r_1,\ldots,r_{k}$ be the $k$ closest points; thus $r=r_i$ for $i=\minp-1$.
We distinguish two cases:
$|pr|\leq (1-\delta)\eps$ and $|pr|> (1-\delta)\eps$.
See also Fig.~\ref{fi:approx-proof}.

In the former case we can conclude that $|r_i r|\leq (1-\delta)\eps$
for all $1\leq i\leq \minp-2$ by setting $s:= r_i$ and $t:=r$
and using $|pr|\leq (1-\delta)\eps$. Thus, including the
point~$p$, we know that $r$ has at least $\minp-1$ points
within distance~$(1-\delta)\eps$.

In the latter case we will argue that $|r_i r|\leq (1-\delta)\eps$
for all $\minp \leq i\leq 2\cdot\minp-3$. Since by part~(iii) of
the claim we have $|rq|\leq (1-\delta)\eps$, we conclude that
also in the latter case~$r$ has at least $\minp-1$ points
within distance~$(1-\delta)\eps$. To argue that $|r_i r|\leq (1-\delta)\eps$
we first note that for any point $s\in c$ we have
$|s s^*|\leq \sin\theta\cdot |p s|$, where
$s^*$ denotes the orthogonal projection of $s$ onto~$\ell_c$.
Thus
\[
\begin{array}{lll}
|r r_i| & \leq & |r r^*| + |r^* r_i^*| + |r_i r_i^*| \\[2mm]
             & \leq & \sin\theta\cdot |pr| + |r^* q^*| + \sin\theta\cdot |p r_i| \\[2mm]
             & \leq & 2\sin\theta \cdot |pq|/\cos \theta + |r^* q^*| \\[2mm]
             & = & 2\sin\theta \cdot |pq|/\cos \theta + |p q^*| - |p r^*| \\[2mm]
             & \leq & 2 \sin\theta\cdot|pq|/\cos \theta + |pq| - |pr|\cos\theta \\[2mm]
             & \leq & \left( 2\frac{\sin\theta}{\cos\theta} + 1-\cos \theta \right)\cdot(1-\delta)\eps
\end{array}
\]
where the last inequality uses $|pq|\leq (1-\delta)\eps$ and that
we are now considering the case $|pr|>(1-\delta)\eps$. Since we
can assume that $\theta$ is small enough to ensure $2\sin\theta<\cos^2\theta$,
we conclude that, indeed, $|r r_i| \leq (1-\delta)\eps$.
This finishes the proof for part~(i) of the claim and hence, of the lemma.
\end{proof}
Combining the previous two lemmas we obtain the following theorem.
\begin{theorem}\label{th:hdbs-approx}
Let $D$ be a set of $n$ points in $\Reals^d$, and let $\eps$ and $\minp$ be given
constants. Then, for any given $\delta>0$, we can compute a $\delta$-approximate
\hdbs clustering on $D$ with respect to $\eps$ and $\minp$ for the Euclidean metric
in $O((n/\delta^{(d-1)/2})\log^{d-1}n)$ time.
\end{theorem}

\section{Experimental evaluation} \label{se:experiments}
In this section we experimentally investigate the efficiency of our new algorithm
and compare it to the original algorithm. The only goal of these experiments is to serve as a proof of concept to illustrate that indeed for very basic point distributions the original algorithm has a bad running time, whereas the new algorithm performs much better.
We first describe some implementation details and we discuss our implementation of
the original algorithm. We then describe the data sets and parameters for the tests.
Finally we present the results and conclusions.

\mypara{Implementation details.}
We implemented two versions of our new algorithm: the strip-based approach
for Step~1 and the grid-based approach.
In Step~3, where we have to decide for all pairs~$b,b'$ of neighboring
boxes whether there are core points $p\in D(b)$ and $p'\in D(b')$
with $|p p'|\leq \eps$, we use the following randomized brute-force approach.
(This is instead
of the theoretically better Delaunay triangulations or spherical emptiness queries.)
Without loss of generality assume $n_{b} \leq n_{b'}$.
For each core point $p \in D(b)$ we first test if $\dist(p,b')\leq \eps$. If not,
no point of $D(b')$ can be within distance $\eps$ of~$p$. If so, we test each point
$p' \in D(b')$ to see if $|pp'|\leq\eps$. If during this procedure we find two
core points within distance $\eps$ from each other, we can stop.
The randomization is obtained by considering the points in each box in random order;
it ensures that if there are many pairs within distance~$\eps$, we expect
to find such a pair early.

The original \dbs algorithm performs a spherical range query for each point $p\in D$
to find all $q\in D$ with $|pq|\leq \eps$. To this end an
indexing structure such as an R-tree is typically used. In our implementation
we use the box-graph to answer the queries. Note that an R-tree also groups the points into boxes
at the leaf level; the tree structure is then used find the leaf boxes
intersecting the query range, after which the points inside these boxes are
tested. The boxes in our box-graph can be seen as being optimized for the
radius~$\eps$ of the query range, and so the box-graph should be at least as
efficient as a general-purpose R-tree.

\mypara{Experimental set-up.}
We ran the algorithms on several synthetic data sets in 2D and 4D,
each consisting of four clusters. (For other numbers of clusters the results are similar.)
The clusters either have a uniform or Gaussian distribution, and their centers
are placed roughly 700 units apart in a hypercube with edge length~1,000 units.
Uniform clusters
are generated within a ball of radius 300 around the cluster center, Gaussian clusters
are generated with a standard deviation of 100. For several data sets we added
5\% noise to the input, uniformly inside a slightly expanded bounding box around the clusters.

\mypara{Parameters.}
We analyse the efficiency of the algorithms with respect to two parameters:
the input size and the density within the clusters.
As a measure of the density we use $n (r/\eps)^d$, where $r$ is the cluster radius;
thus, for the uniform distribution, the density represents the expected number of points within distance
$\eps$ from a core point, and also the expected number of points in the boxes of the box-graph
within the clusters.

\mypara{Measurements.}
We compare the algorithms in two different ways: we measure the actual execution times,
and the number of pairs of points for which a distance computation is done.
The latter is the main operation for finding the clusters in both the new and the
original algorithm, so it provides a good implementation-independent measure.
For the original algorithm we also count the sum of neighborhood sizes of all points.
Following earlier work, we call this the number of \emph{seeds}.
This is a lower bound for the number of operations needed in the original algorithm,
independent of the indexing structure used to find the neighborhoods.

\mypara{Result for fixed input size.}
In these experiments we fix the input size and run the algorithms with different values of~$\eps$.
In 2D we ran the algorithms on a data set of each type---uniform or Gaussian and with or without noise---in which the clusters contain 500,000 points each. In 4D we use the same types of data sets, but with 200,000 points per cluster. The results are shown in Fig.~\ref{fig:2D-fixed-n}.

\begin{figure}[p]
\centering
\begin{tabular}{@{}c@{}c@{}c@{}c@{}}
\multicolumn{4}{c}{2D}\\
Gaussian & Gaussian + noise & uniform & uniform + noise\\
\includegraphics[trim=52 250 150 250 ,clip,scale=0.25]{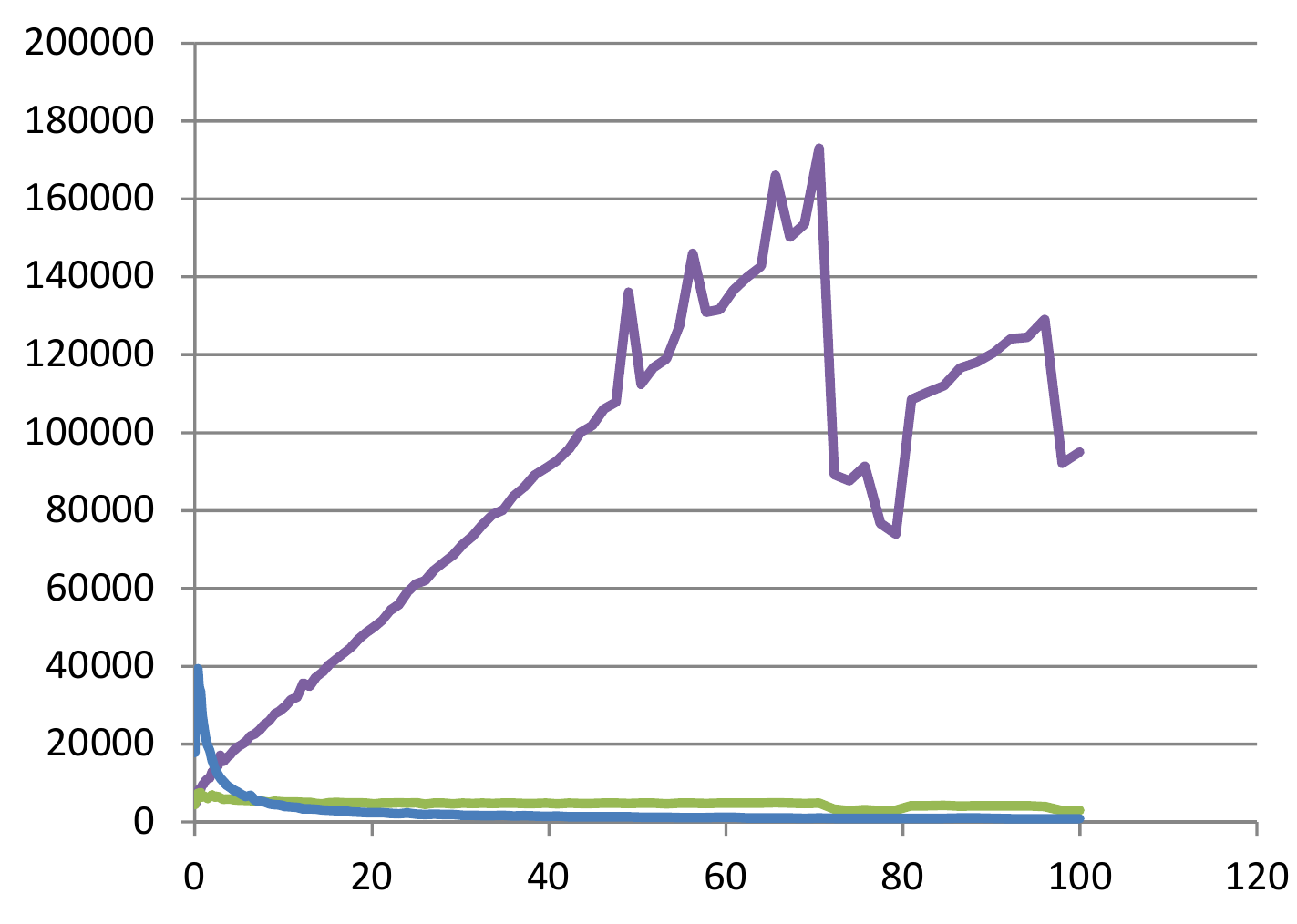} &
\includegraphics[trim=52 250 150 250 ,clip,scale=0.25]{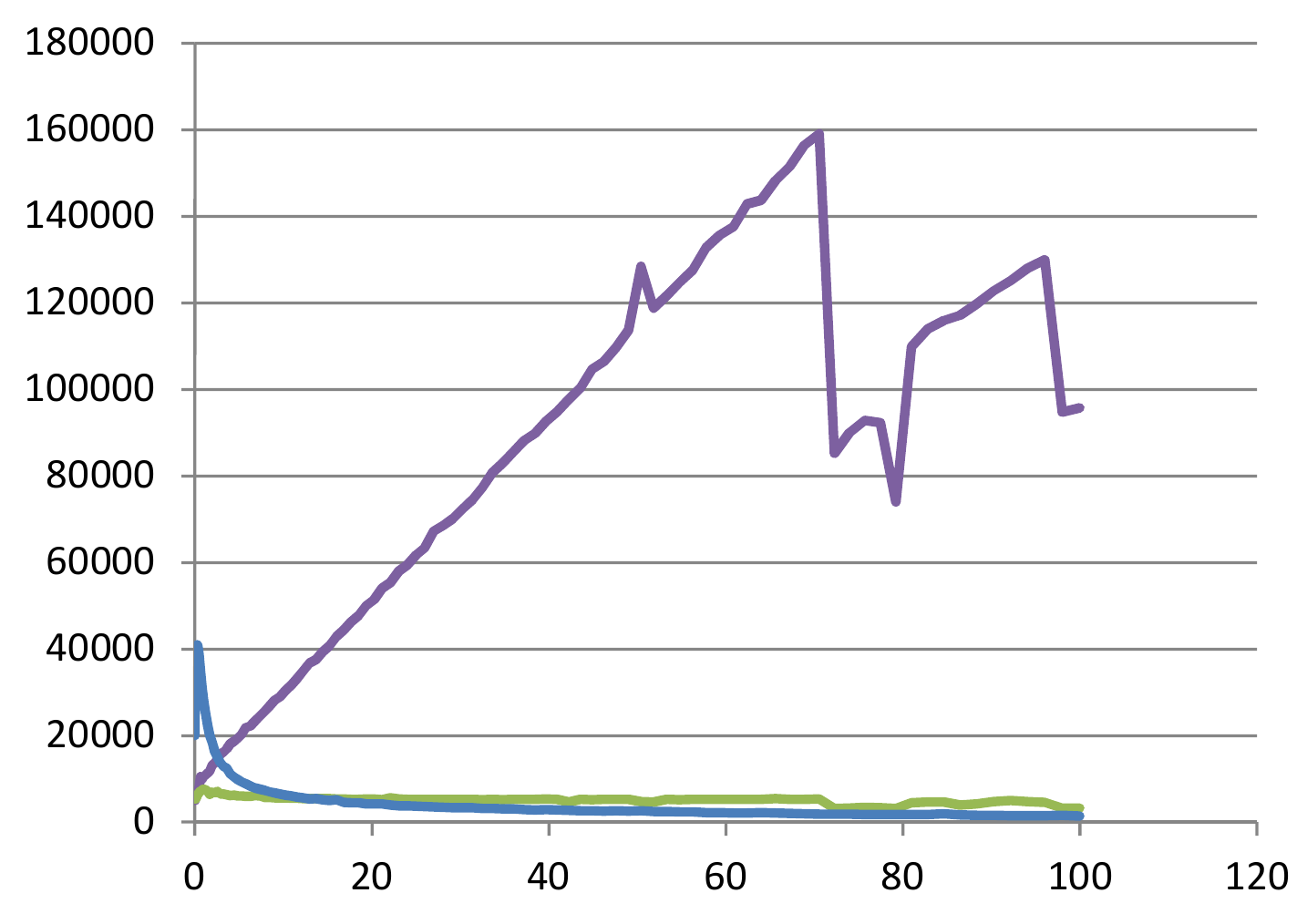} &
\includegraphics[trim=52 250 150 250 ,clip,scale=0.25]{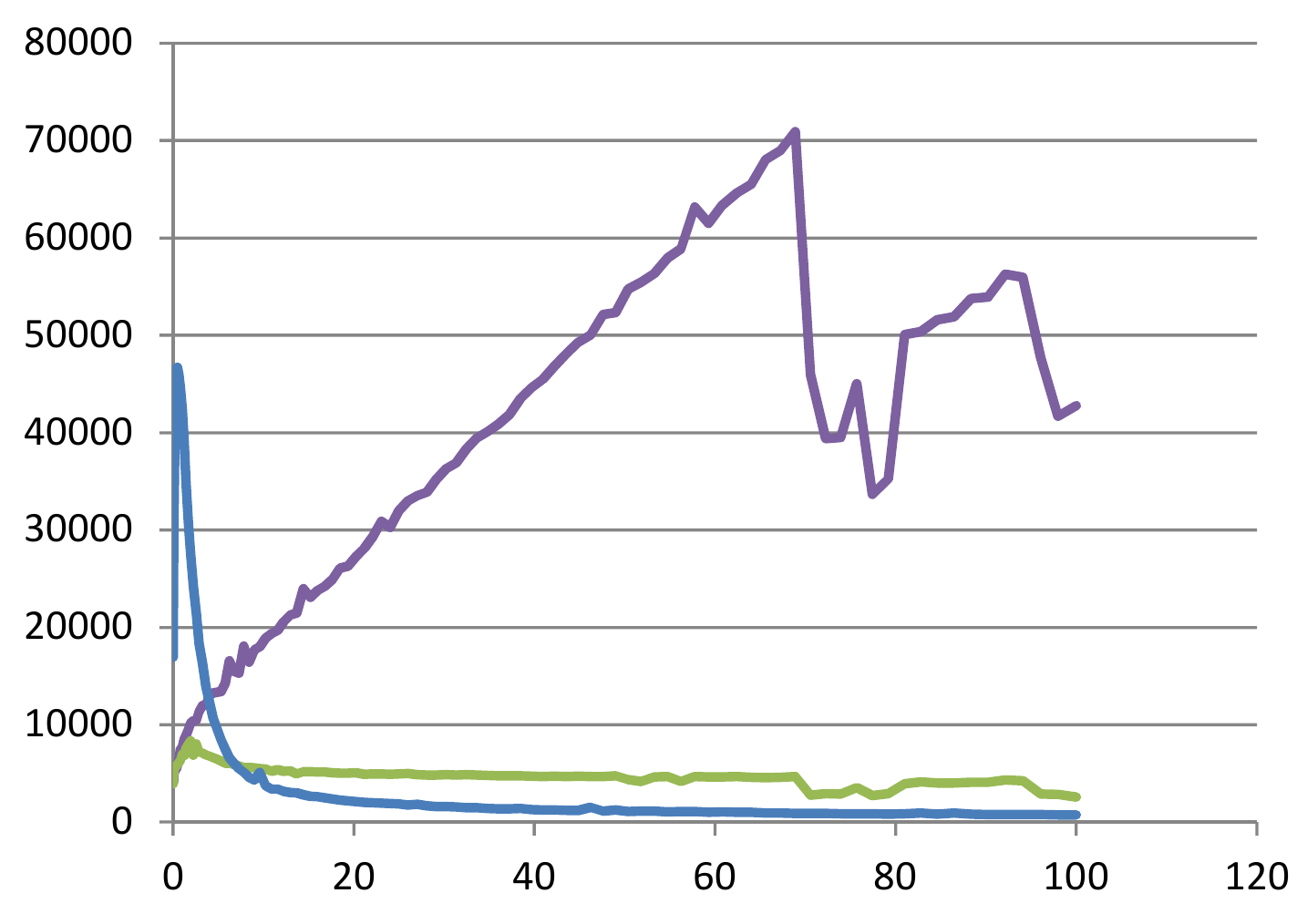} &
\includegraphics[trim=52 250 52 250 ,clip,scale=0.25]{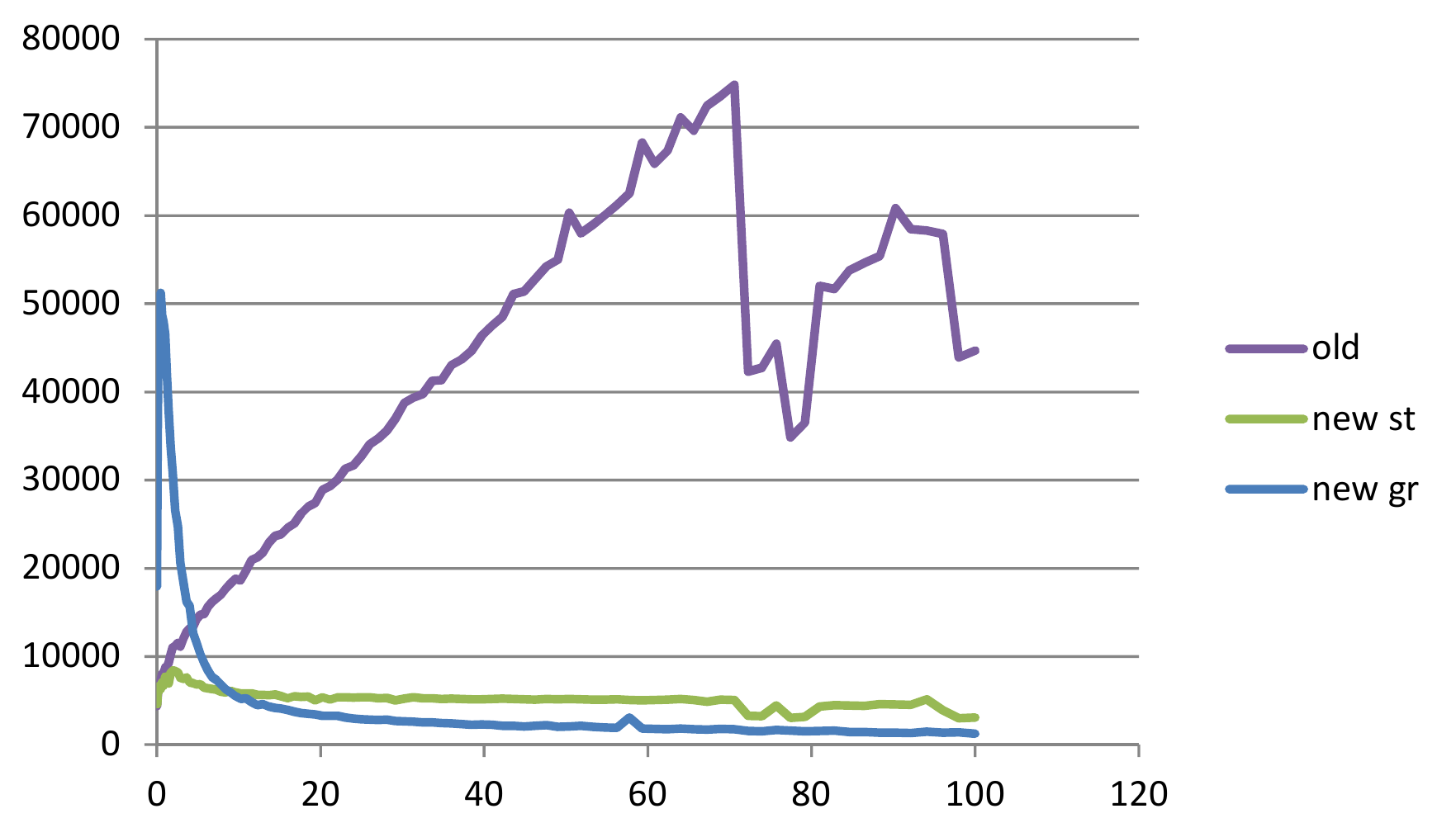}
\\
\includegraphics[trim=52 250 150 250 ,clip,scale=0.25]{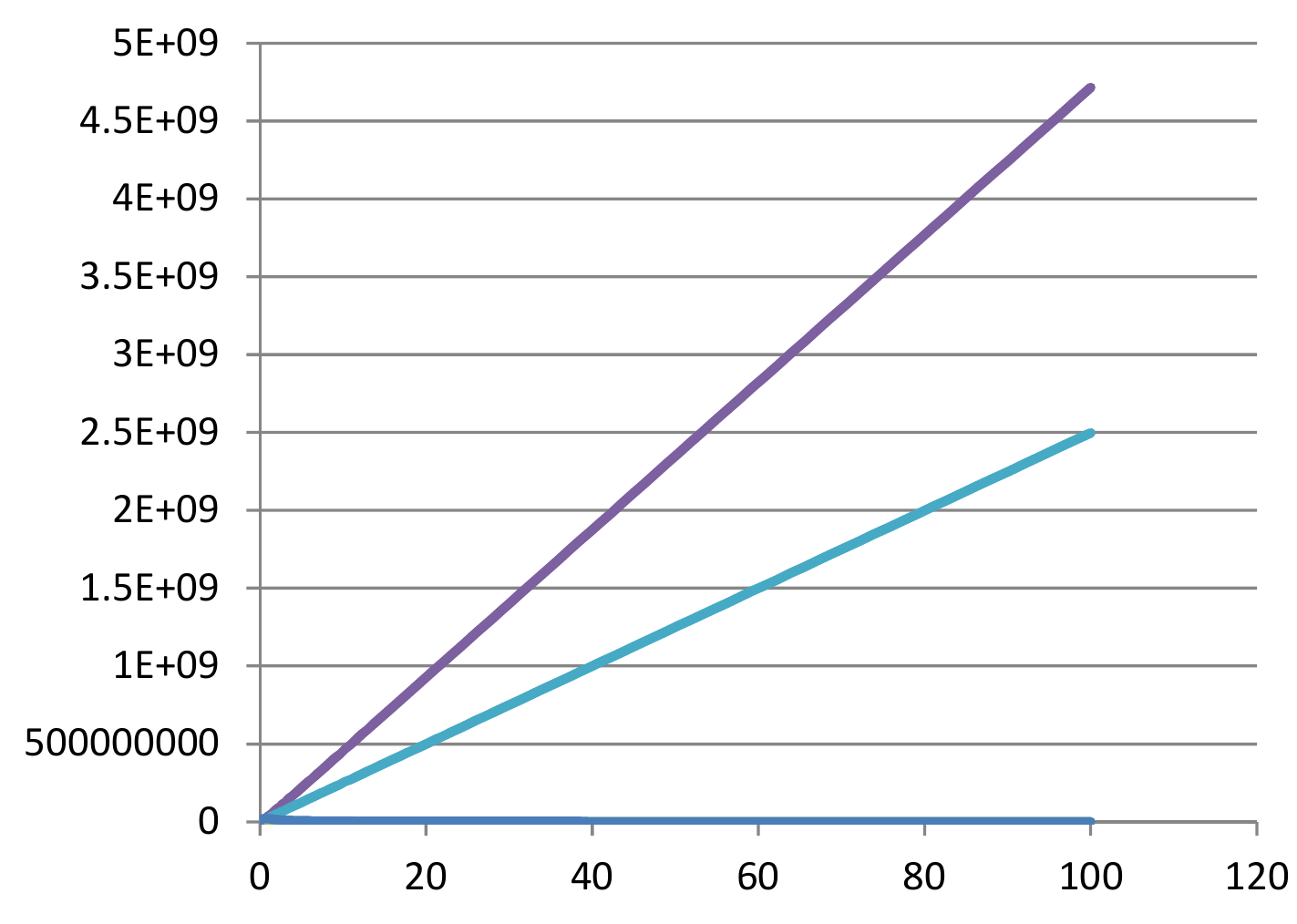} &
\includegraphics[trim=52 250 150 250 ,clip,scale=0.25]{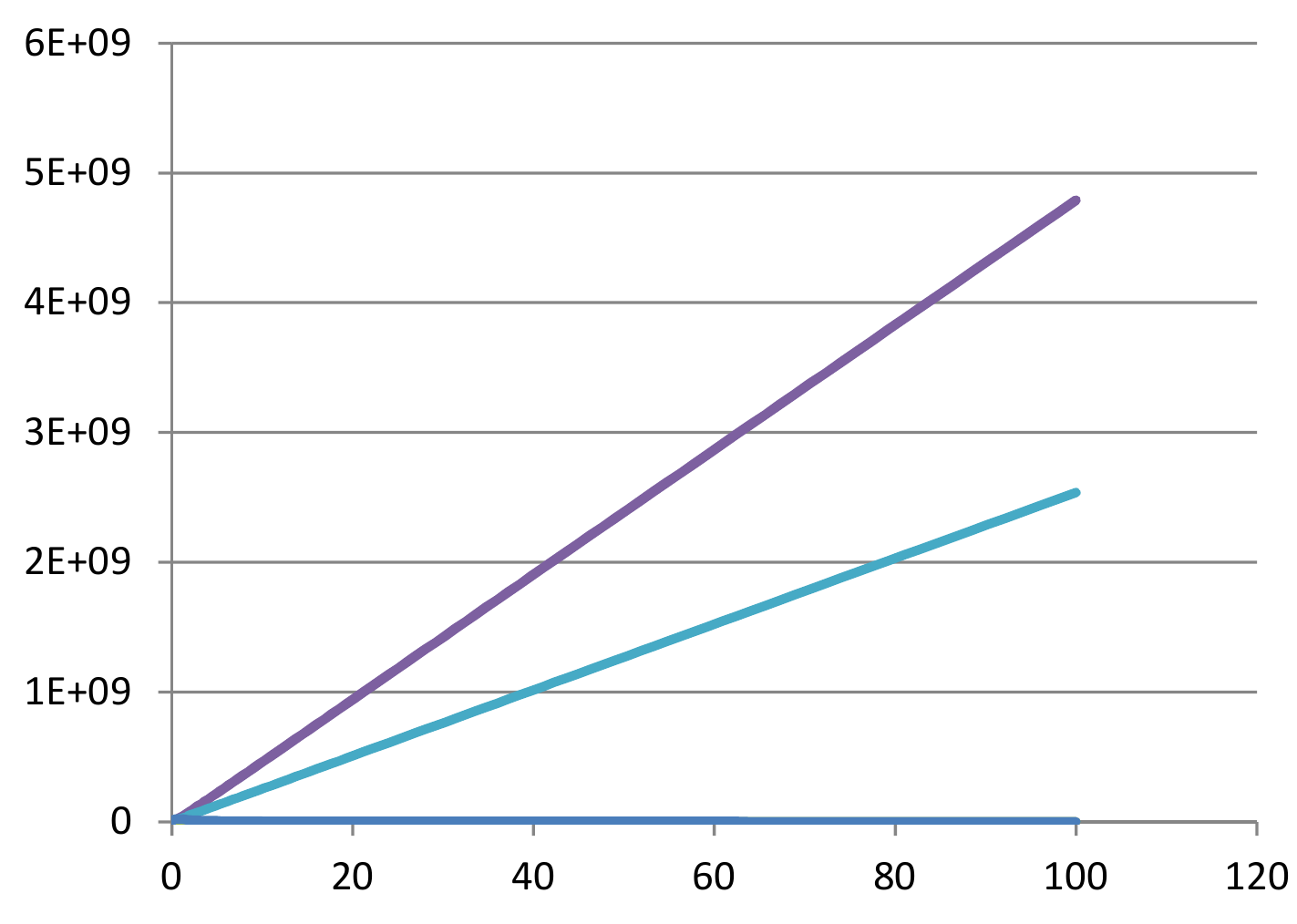} &
\includegraphics[trim=52 250 150 250 ,clip,scale=0.25]{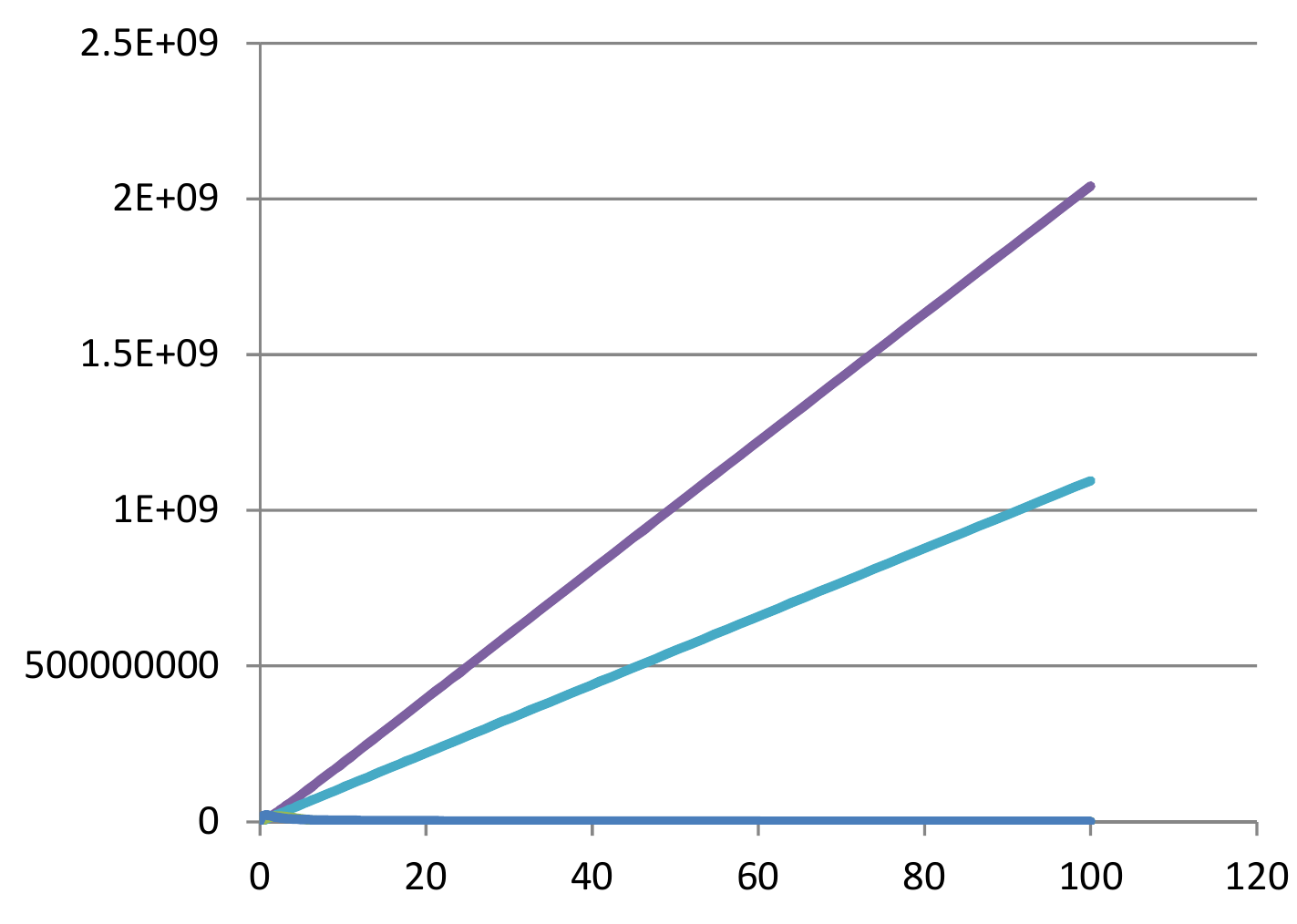} &
\includegraphics[trim=52 250 52 250 ,clip,scale=0.25]{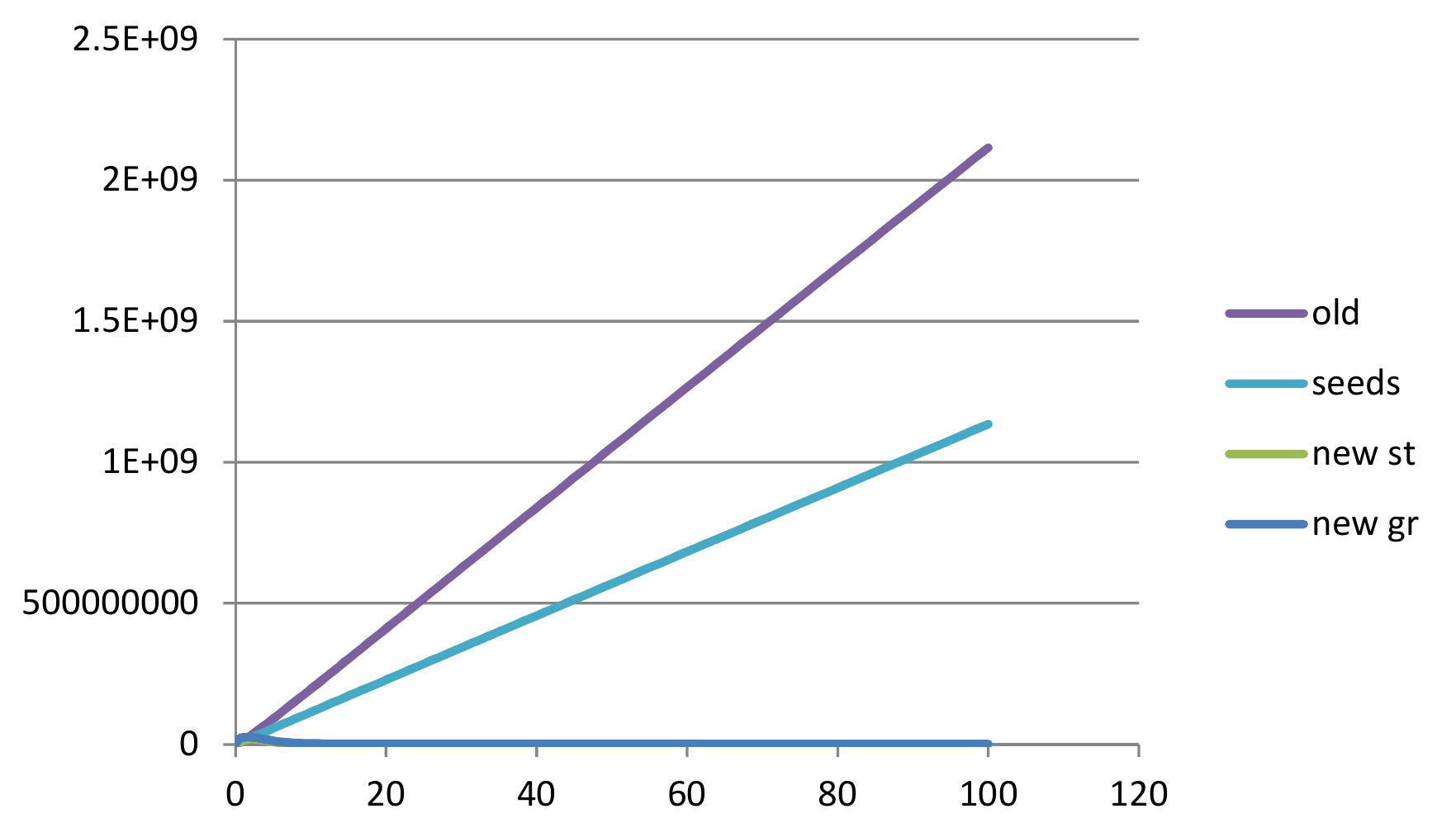}
\\
\includegraphics[trim=52 250 150 250 ,clip,scale=0.25]{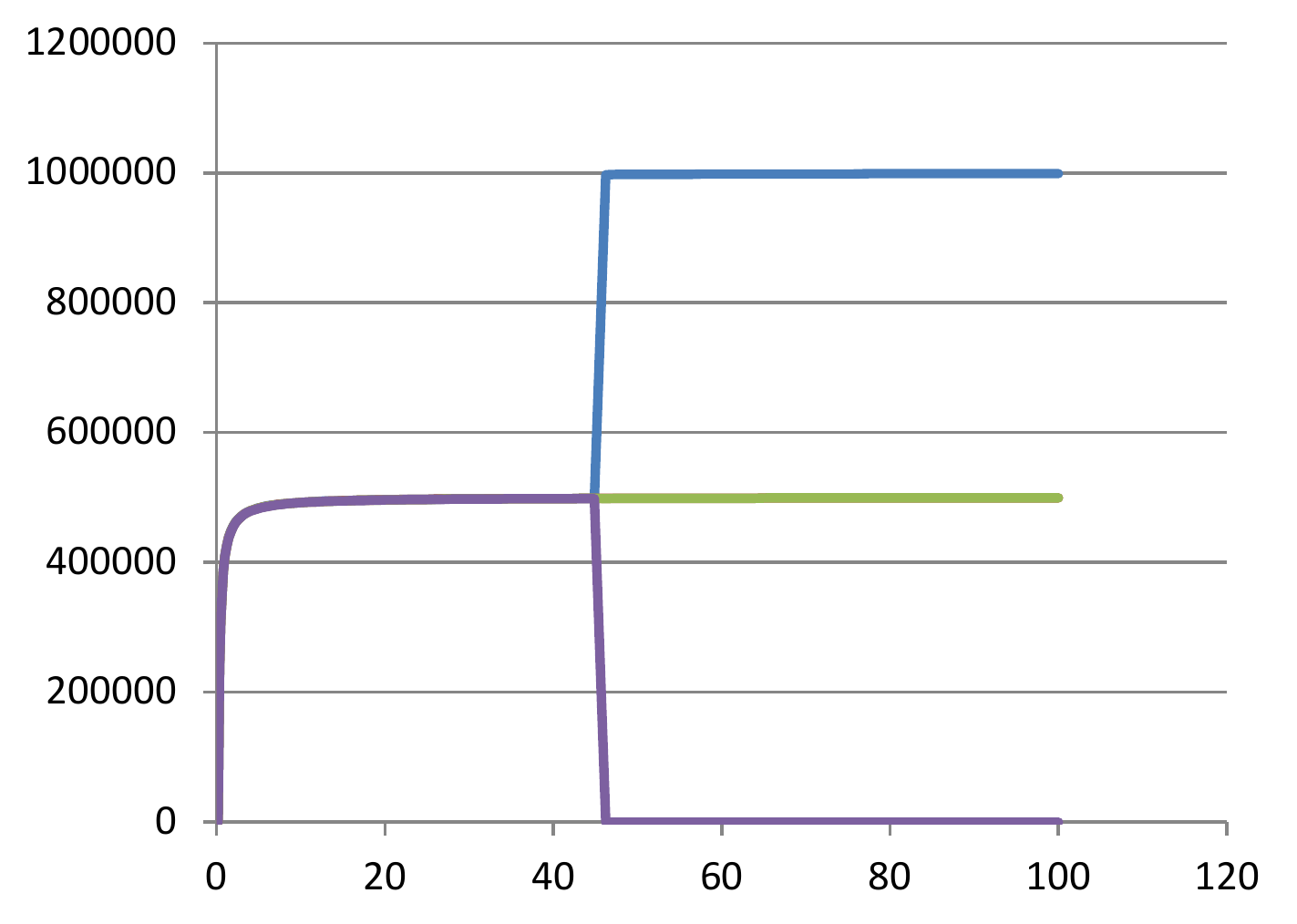} &
\includegraphics[trim=52 250 150 250 ,clip,scale=0.25]{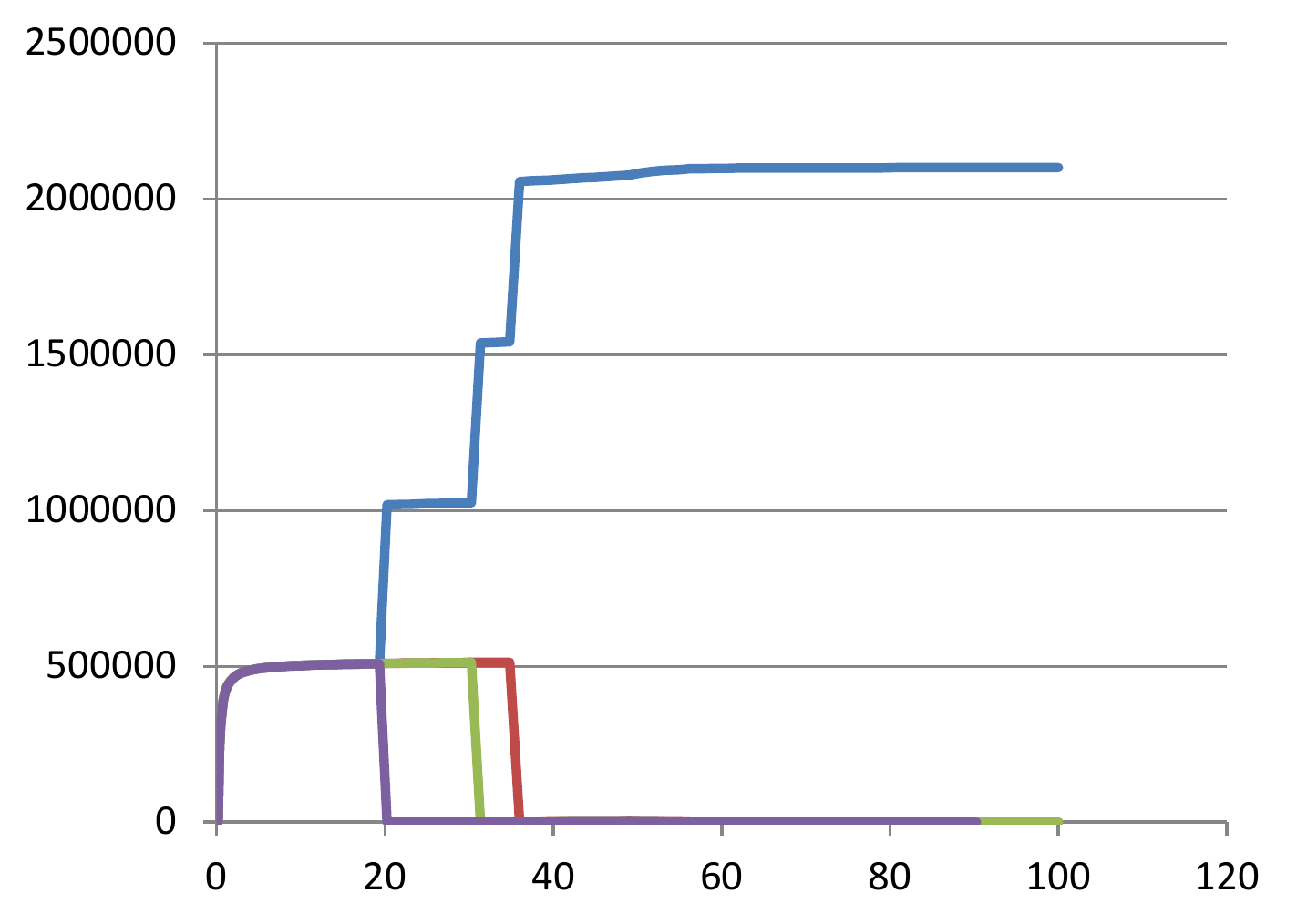} &
\includegraphics[trim=52 250 150 250 ,clip,scale=0.25]{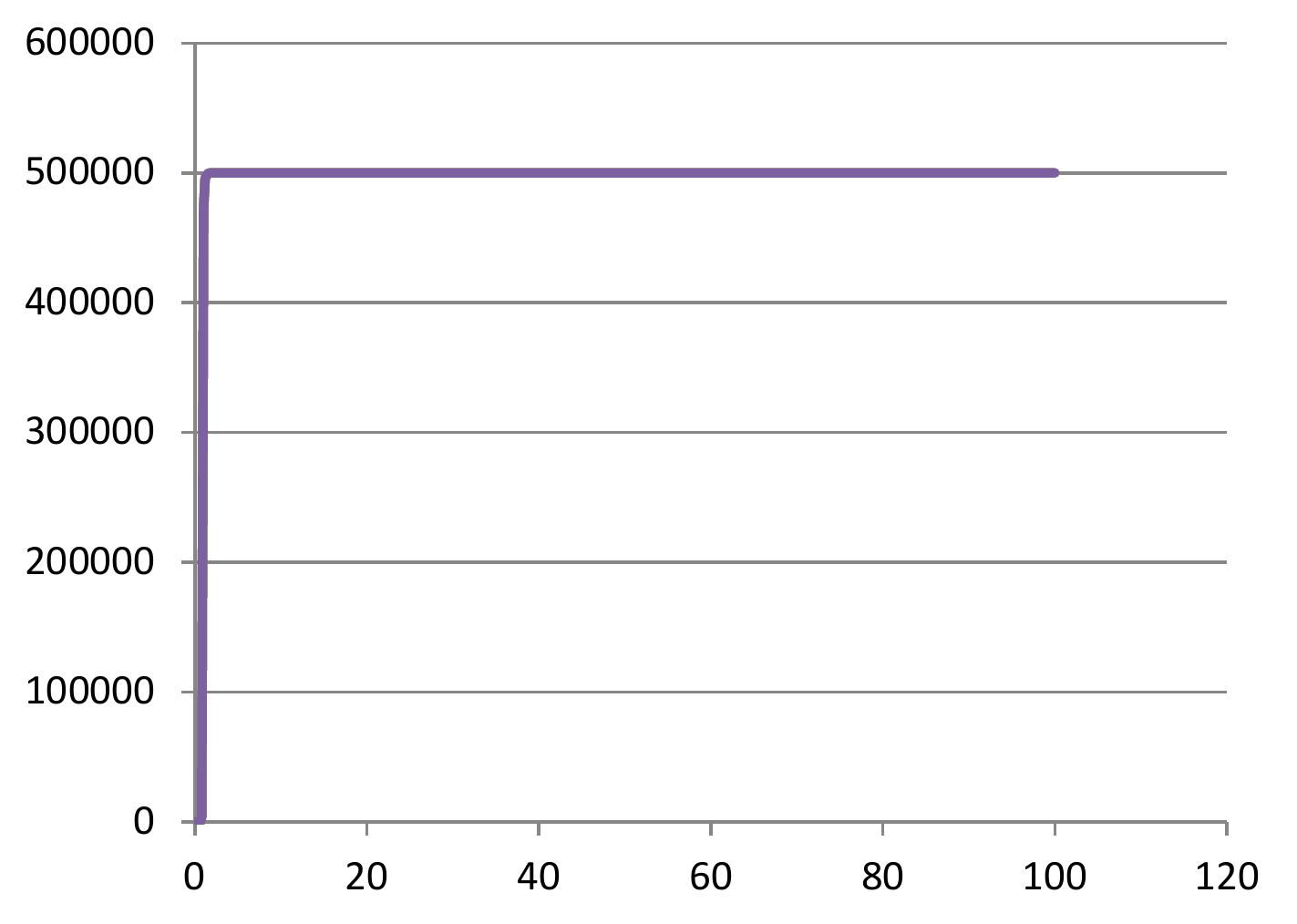} &
\includegraphics[trim=52 250 52 250 ,clip,scale=0.25]{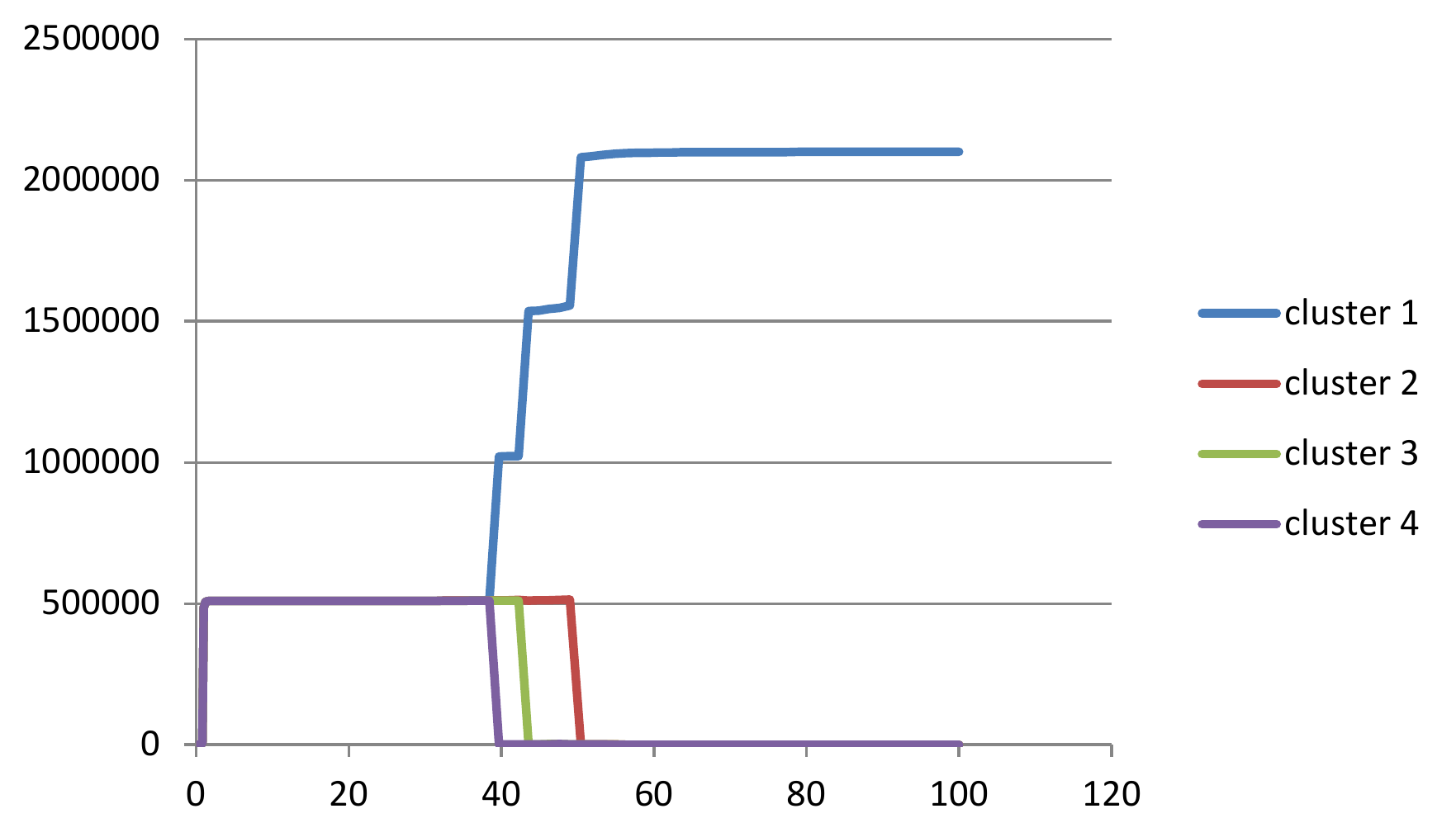}
\\ \\
\multicolumn{4}{c}{4D}\\
\includegraphics[trim=52 250 150 250 ,clip,scale=0.25]{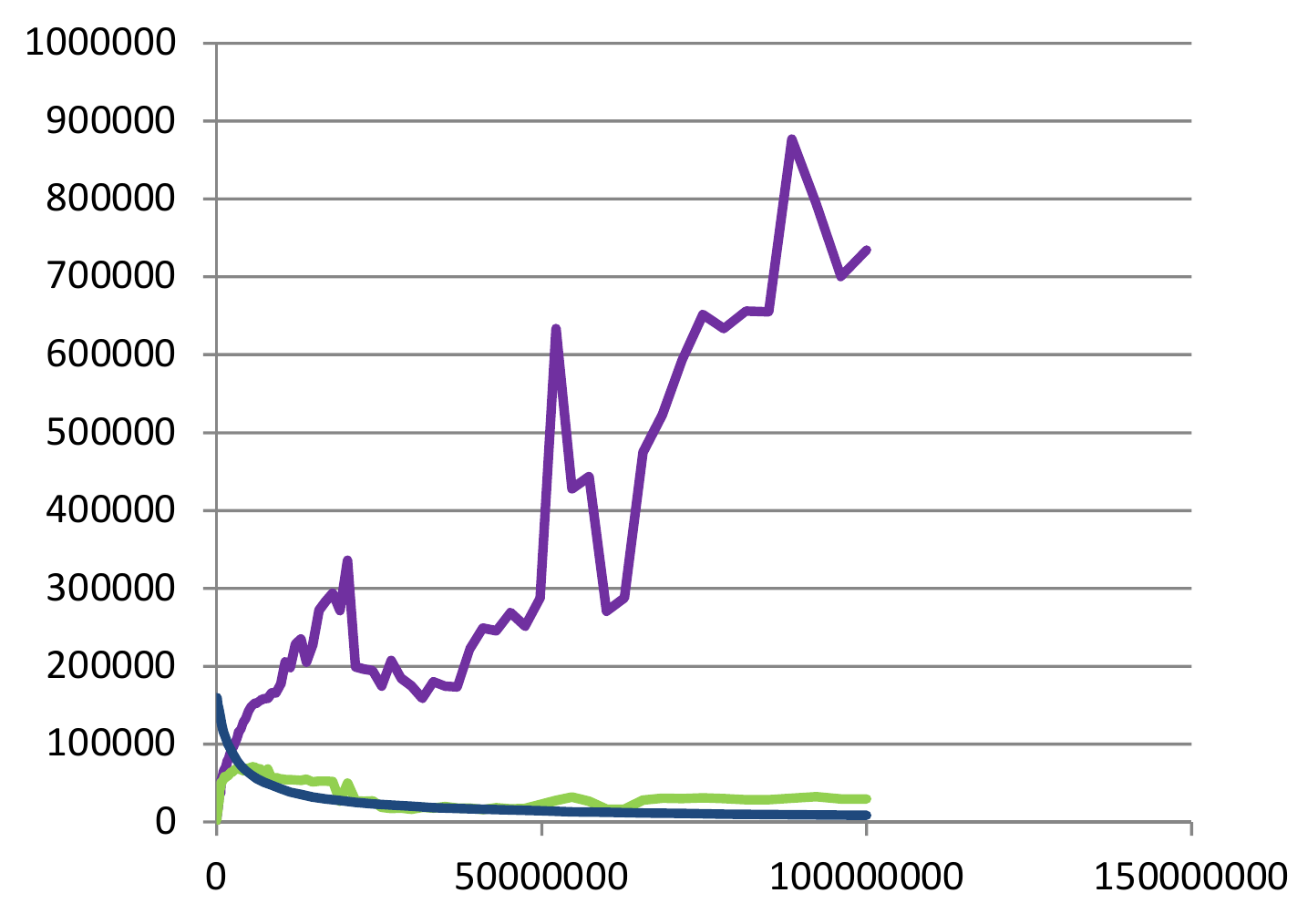} &
\includegraphics[trim=52 250 150 250 ,clip,scale=0.25]{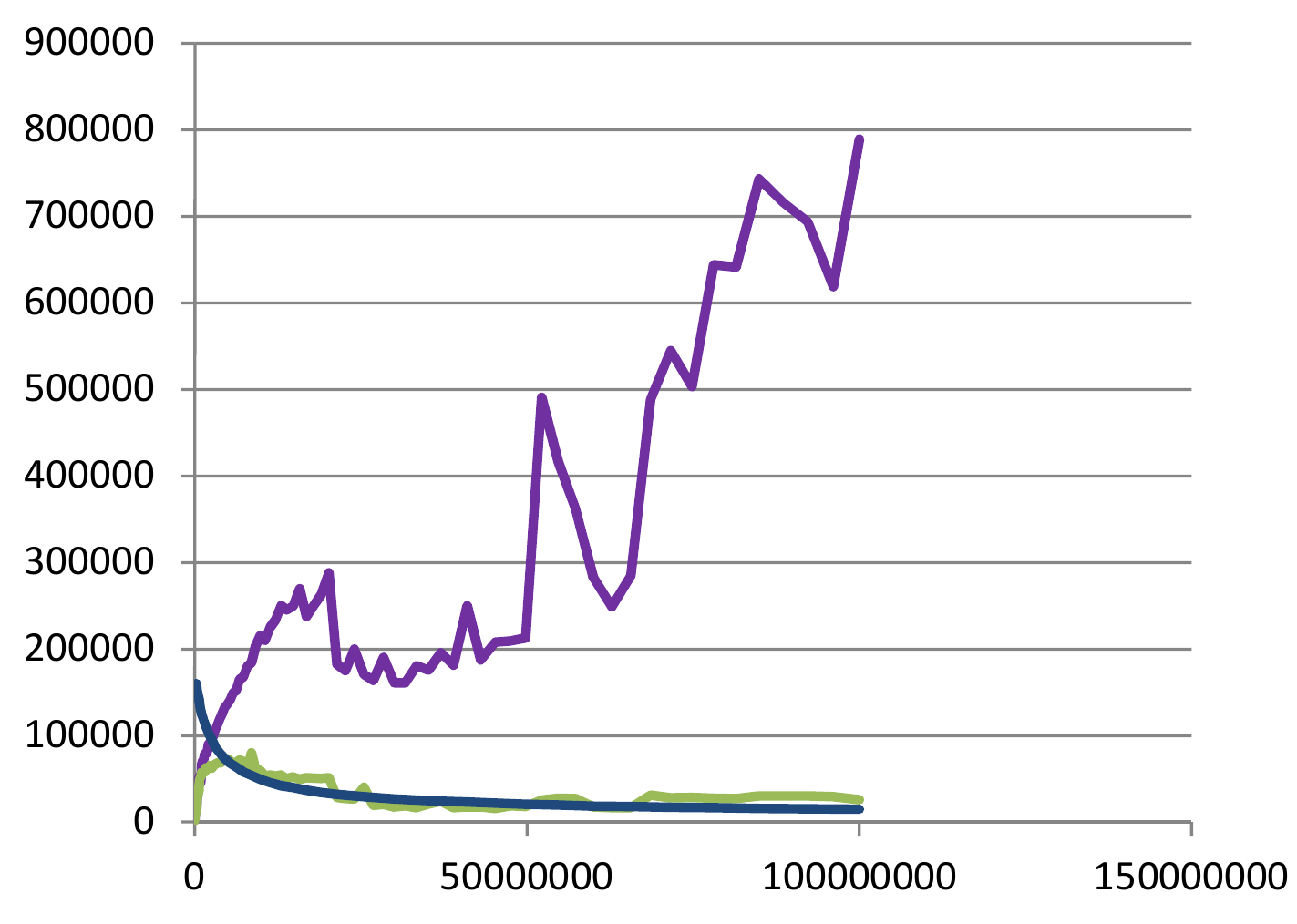} &
\includegraphics[trim=52 250 150 250 ,clip,scale=0.25]{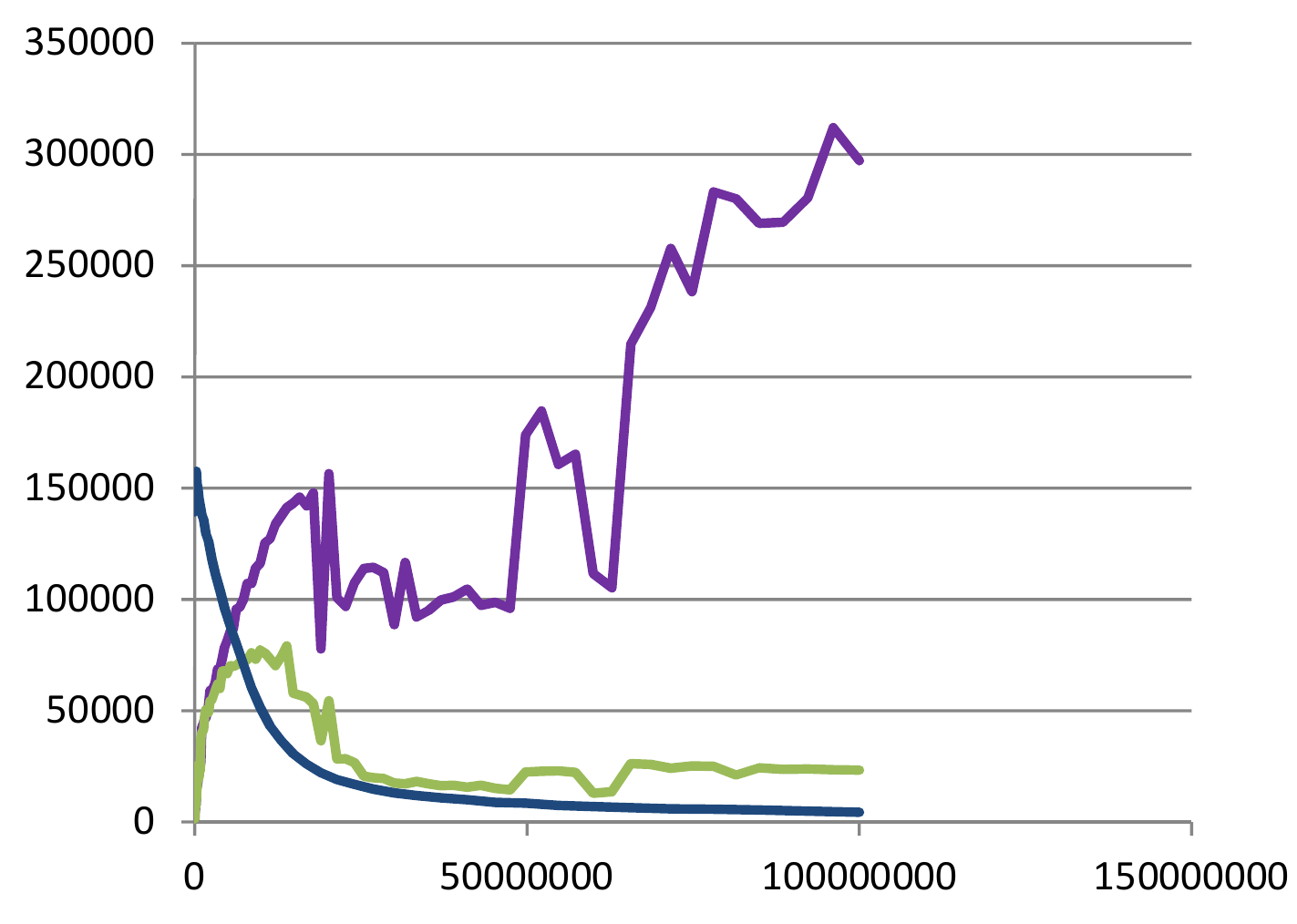} &
\includegraphics[trim=52 250 52 250 ,clip,scale=0.25]{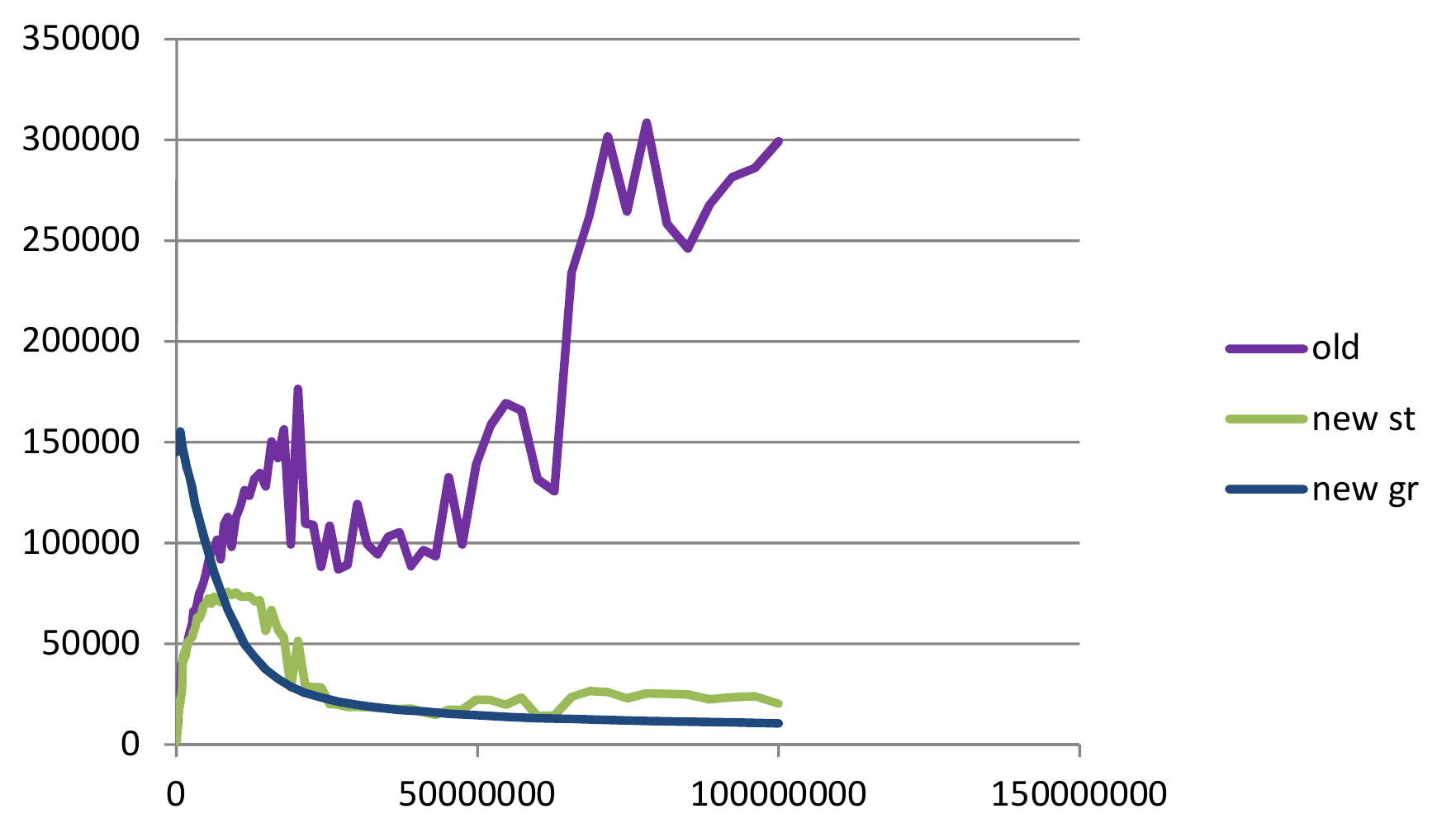}
\\
\includegraphics[trim=52 250 150 250 ,clip,scale=0.25]{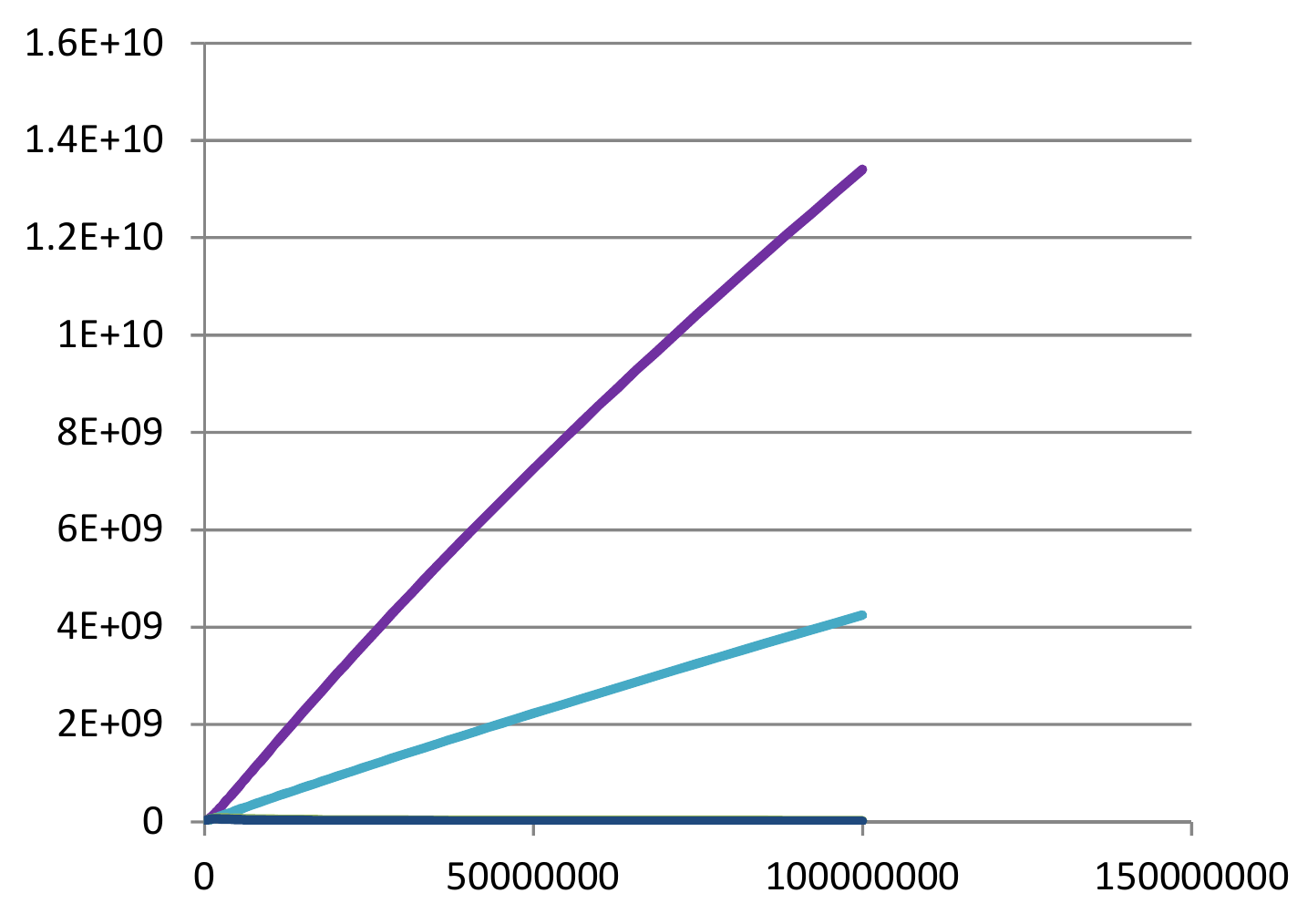} &
\includegraphics[trim=52 250 150 250 ,clip,scale=0.25]{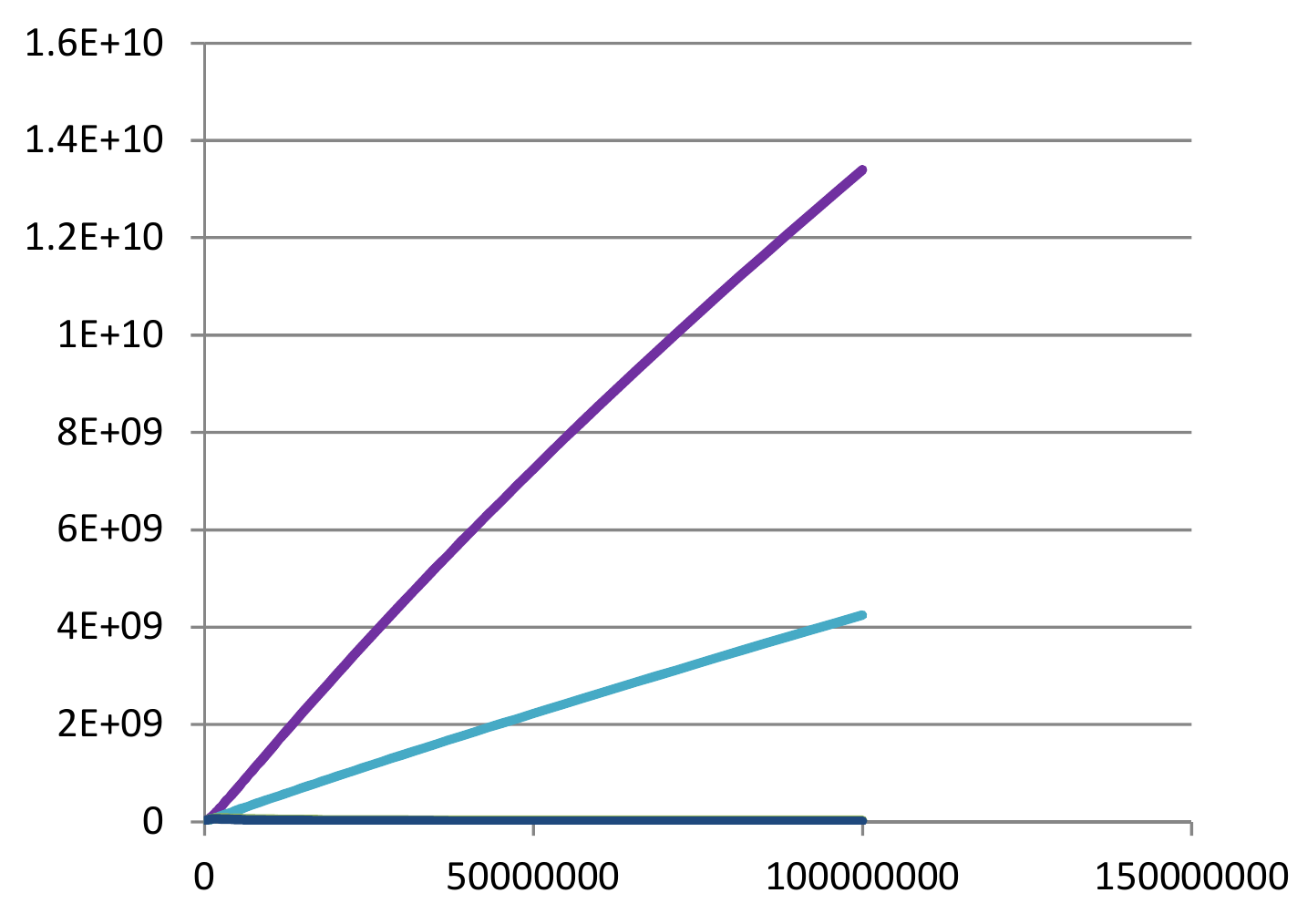} &
\includegraphics[trim=52 250 150 250 ,clip,scale=0.25]{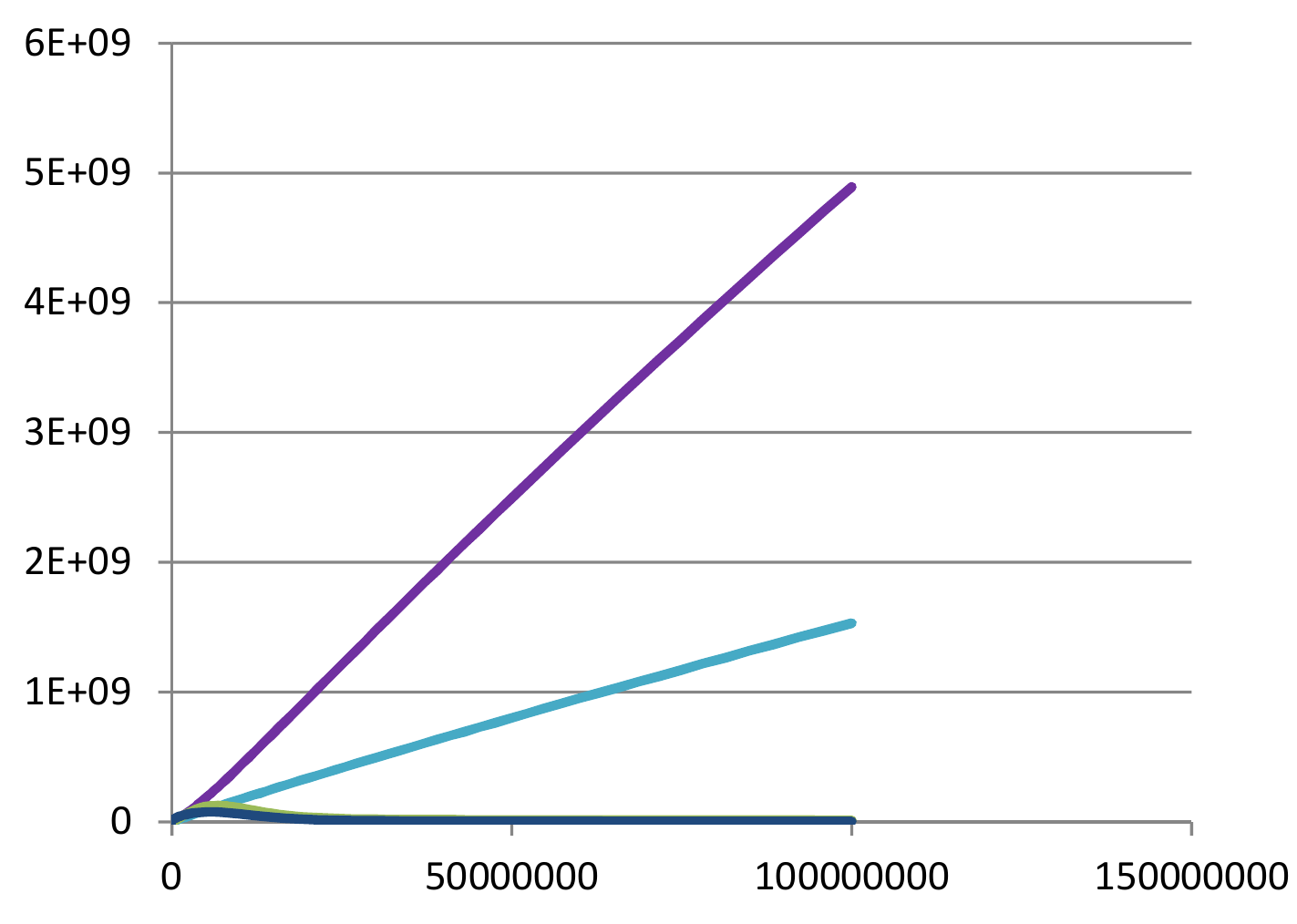} &
\includegraphics[trim=52 250 52 250 ,clip,scale=0.25]{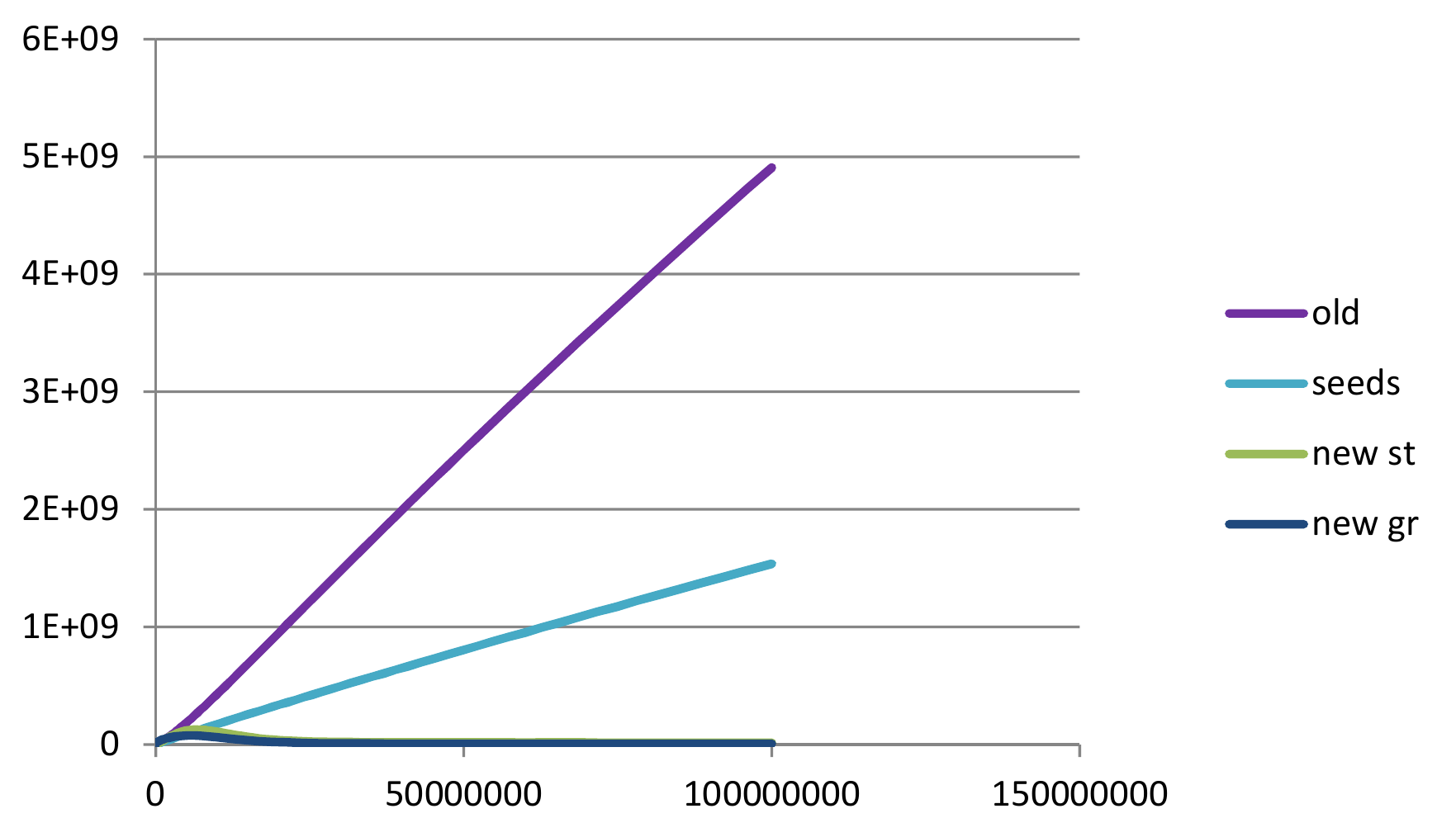}
\\
\includegraphics[trim=52 250 150 250 ,clip,scale=0.25]{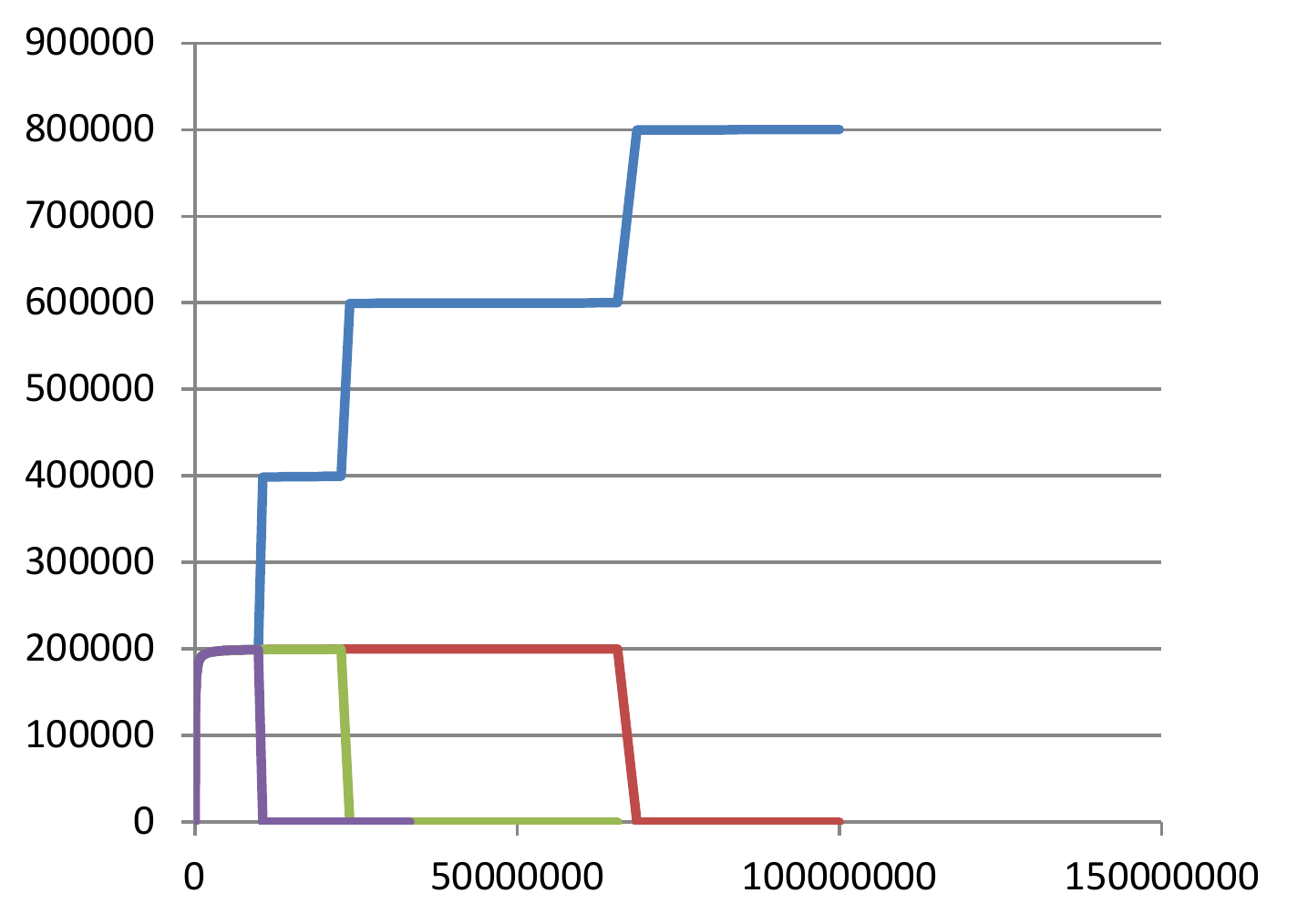} &
\includegraphics[trim=52 250 150 250 ,clip,scale=0.25]{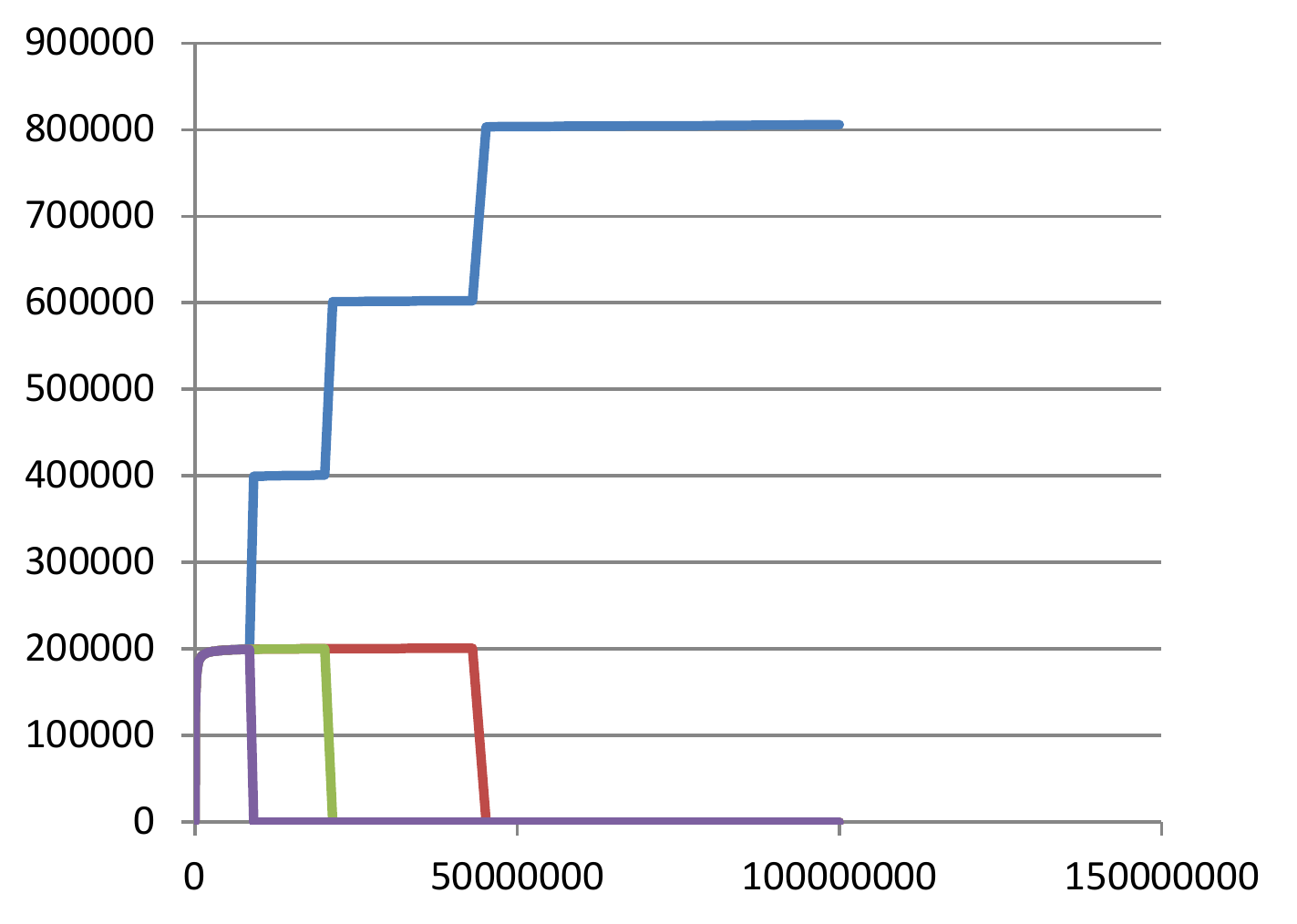} &
\includegraphics[trim=52 250 150 250 ,clip,scale=0.25]{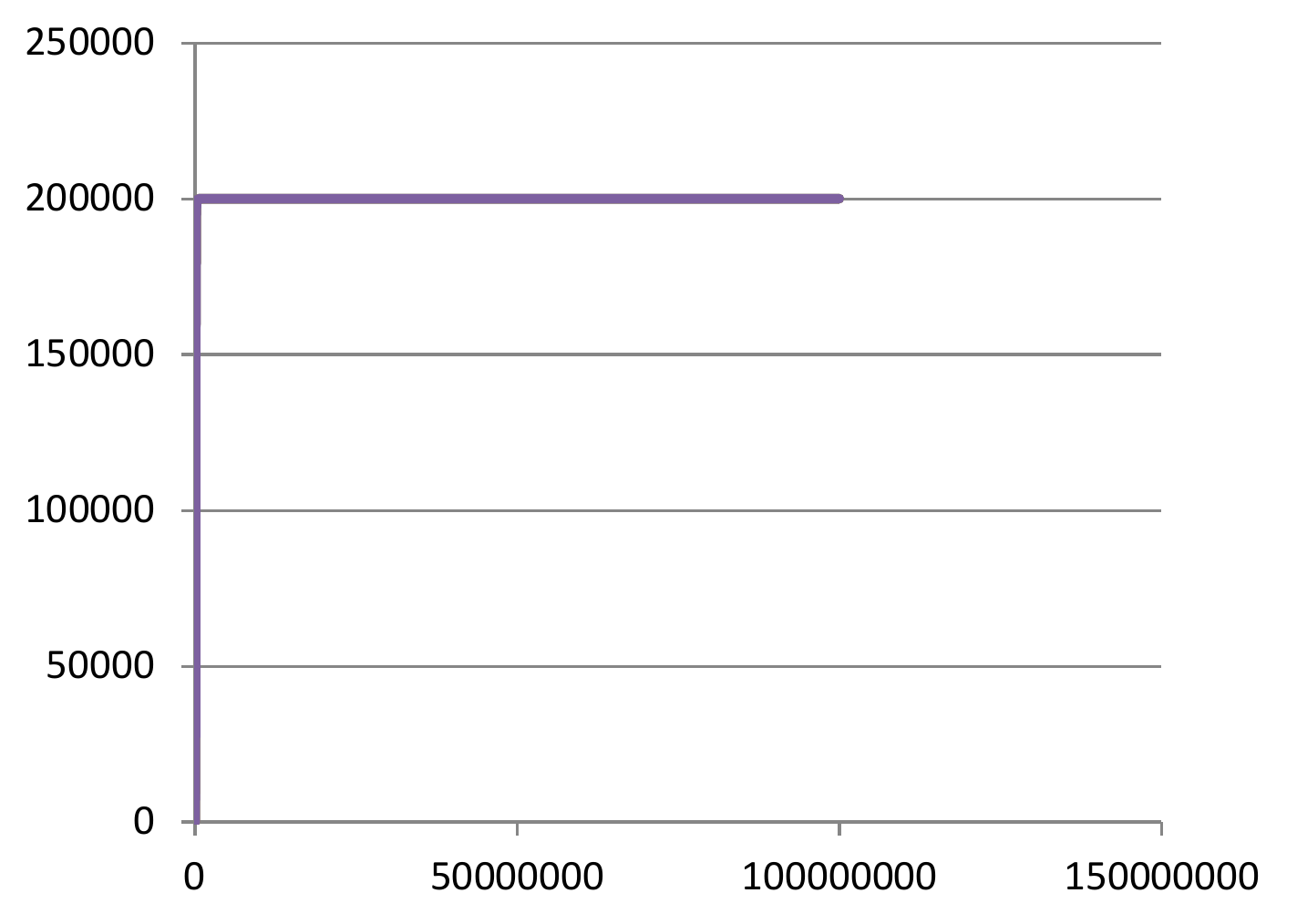} &
\includegraphics[trim=52 250 52 250 ,clip,scale=0.25]{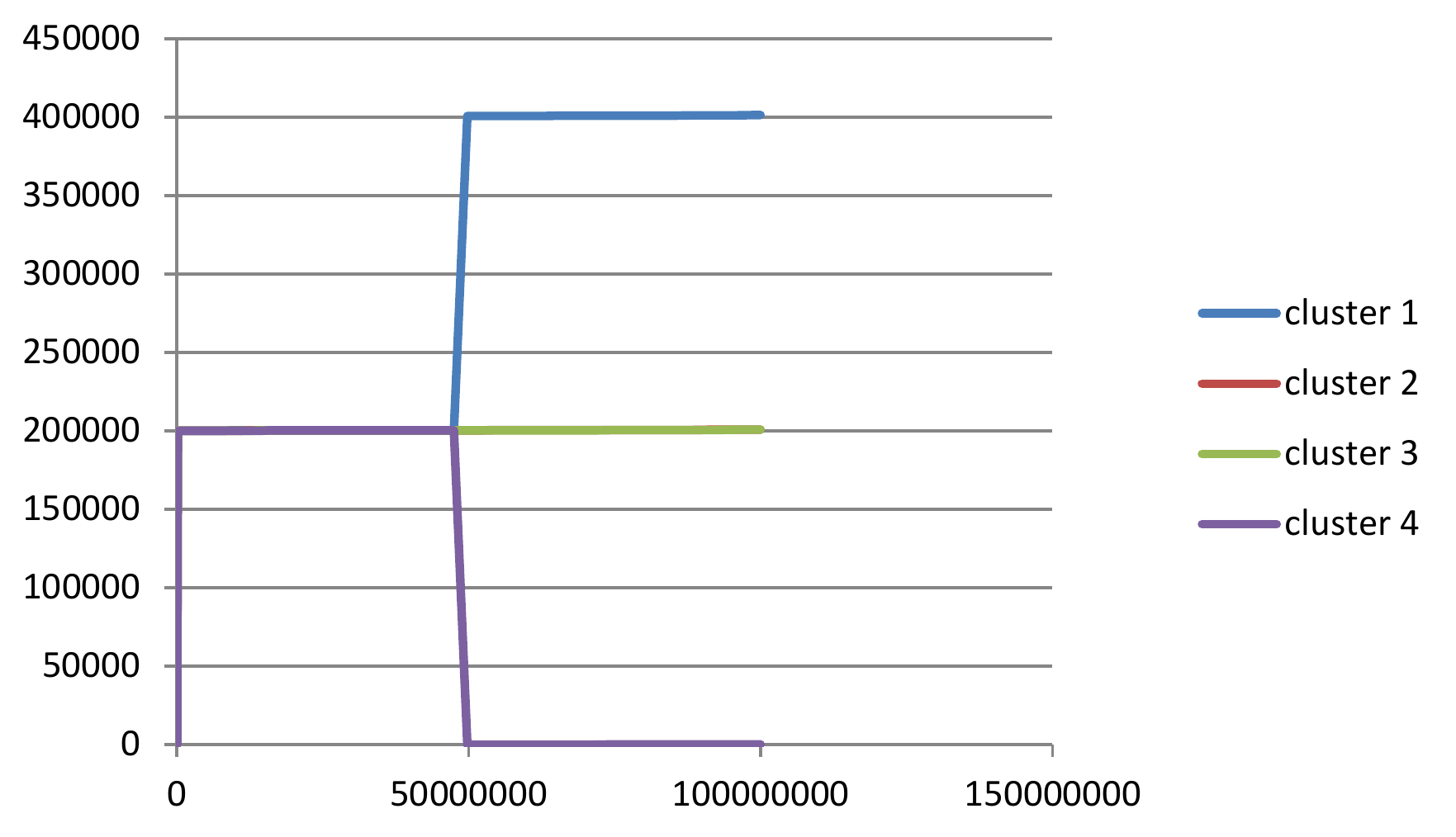}
\end{tabular}
\caption{Results when varying~$\eps$. The $x$-axis denotes $\eps^2$ and the $y$-axis denotes (from top to bottom row for both 2D and 4D) processing time in ms, number of distance calculations, and the number of points in the four largest clusters. The third row thus show for which values of $\eps$ the ``correct''
clustering (with four clusters of 500,000 points for 2D and 200,000 points in 4D) is found. The new method with a strip-based or grid-based construction is denoted by \emph{new st} and \emph{new gr} respectively. Note that in the second row these lines overlap.}
\label{fig:2D-fixed-n}
\end{figure}

\mypara{Result for fixed density.}
In our second set of experiments we use data sets of different sizes, where for
each data set we pick $\eps$ such that the density (as defined above) remains constant.
For each type of data we generated data sets ranging from 20,000 to 500,000 points per cluster
in 2D, and from 20,000 to 300,000 points per cluster for~4D. We ran these experiments for two different values of~$\eps$: one that is roughly the smallest value needed to find the four clusters,
and one that is roughly the highest value for which the four clusters do not start to merge. The results are given in Fig.~\ref{fig:2D-fixed-density}.

\begin{figure}[p]
\centering
\begin{tabular}{@{}c@{}c@{}c@{}c@{}}
\multicolumn{4}{c}{2D}\\
Gaussian & Gaussian + noise & uniform & uniform + noise\\
\includegraphics[trim=52 250 150 250 ,clip,scale=0.25]{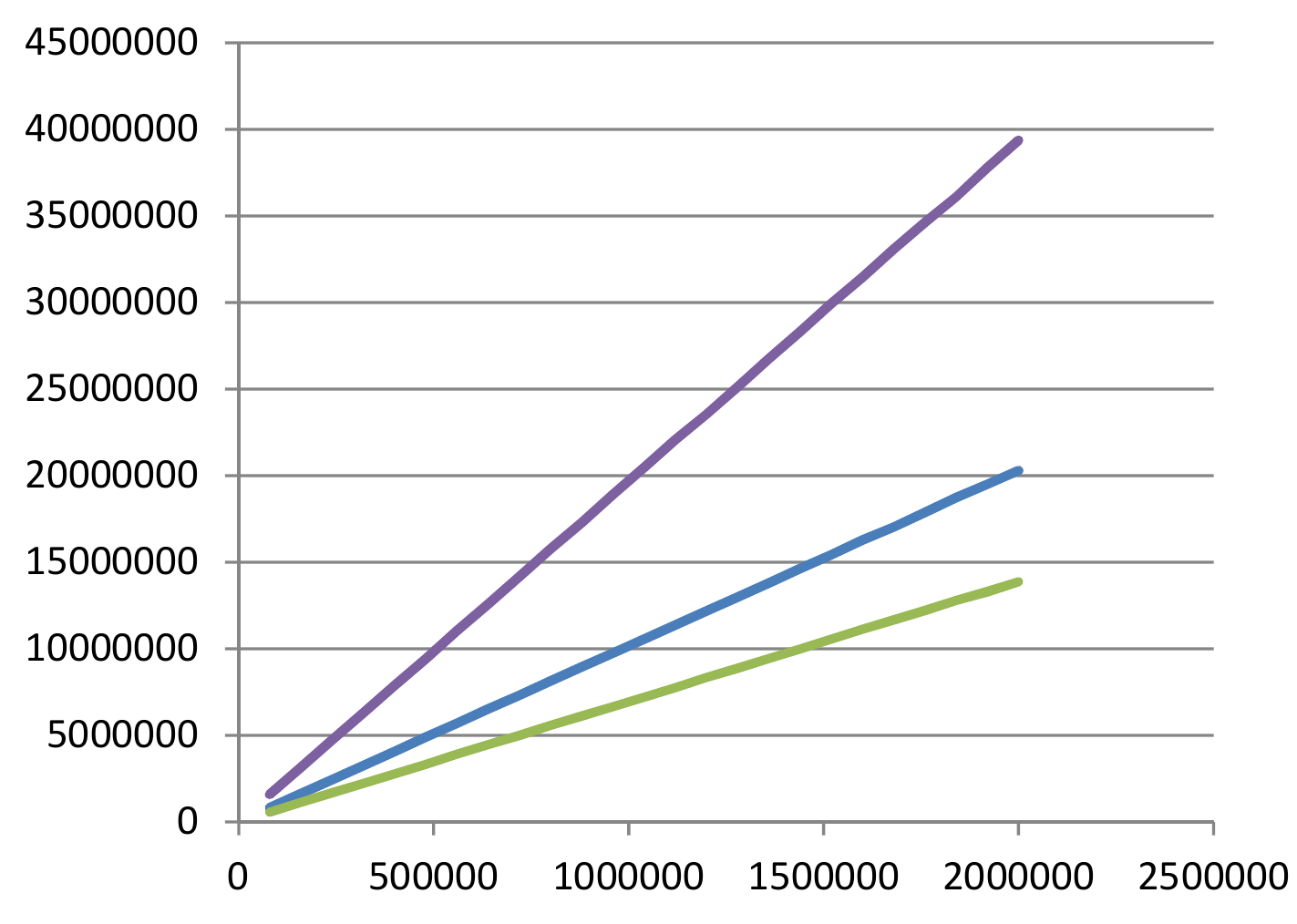} &
\includegraphics[trim=52 250 150 250 ,clip,scale=0.25]{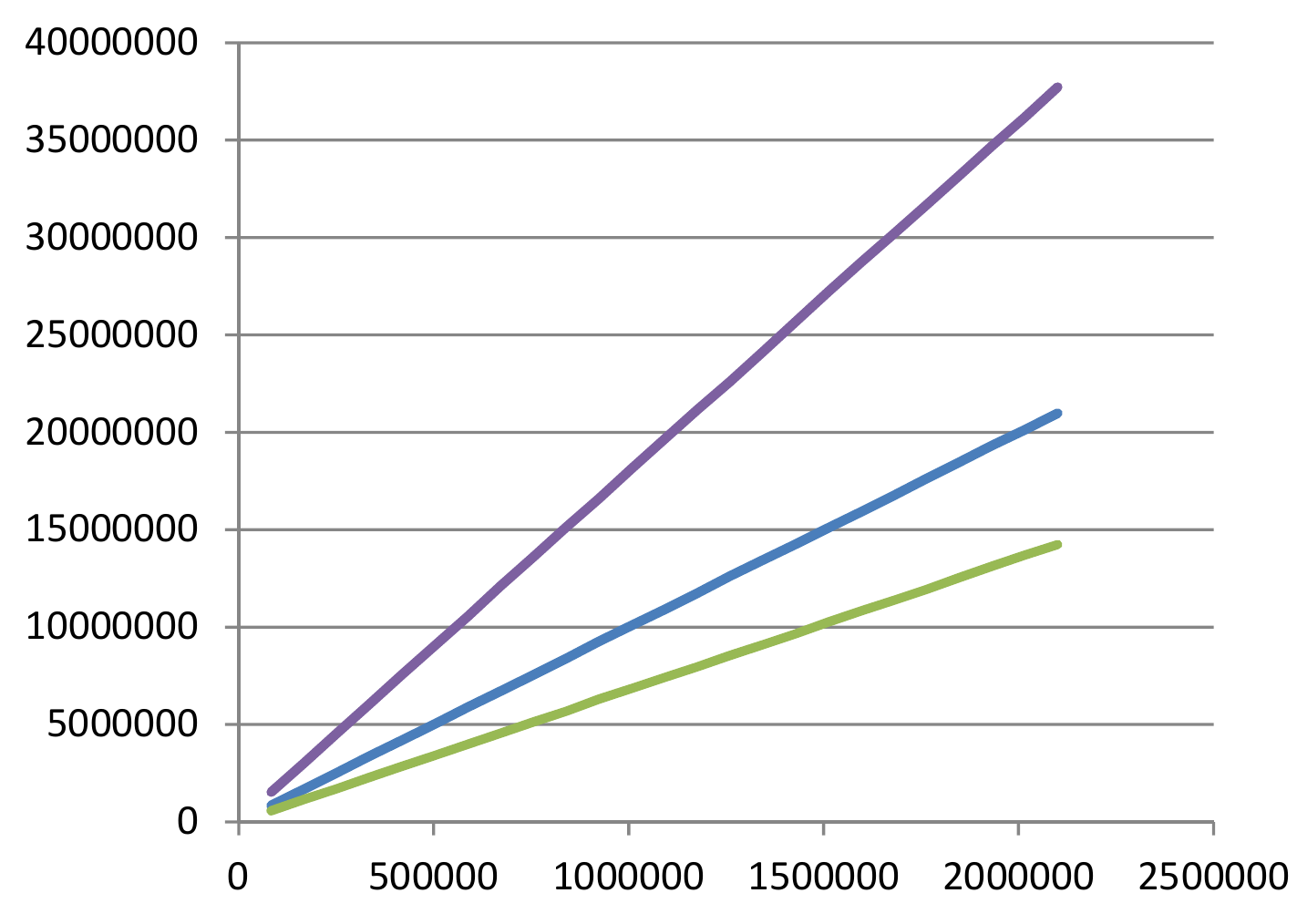} &
\includegraphics[trim=52 250 150 250 ,clip,scale=0.25]{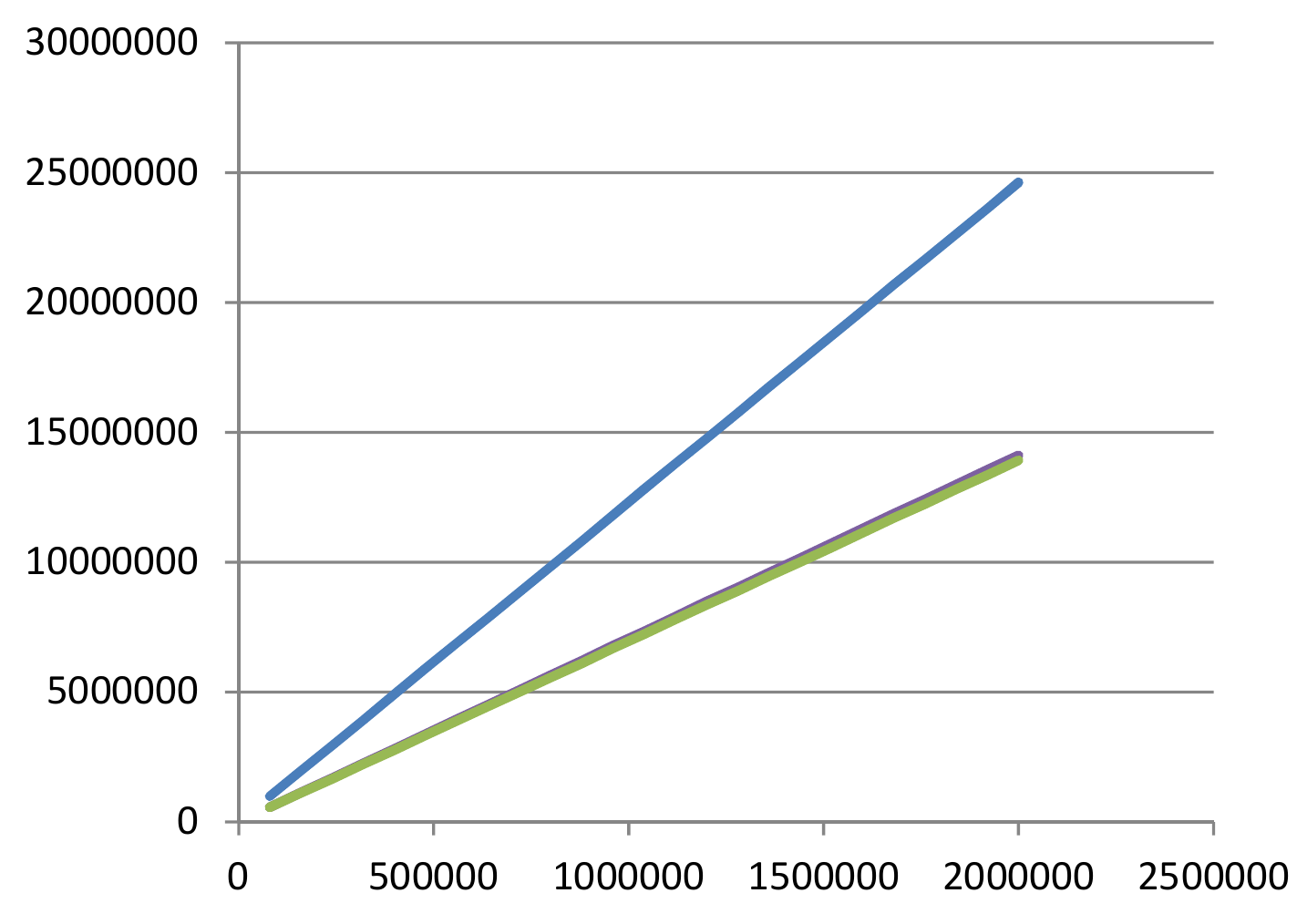} &
\includegraphics[trim=52 250 52 250 ,clip,scale=0.25]{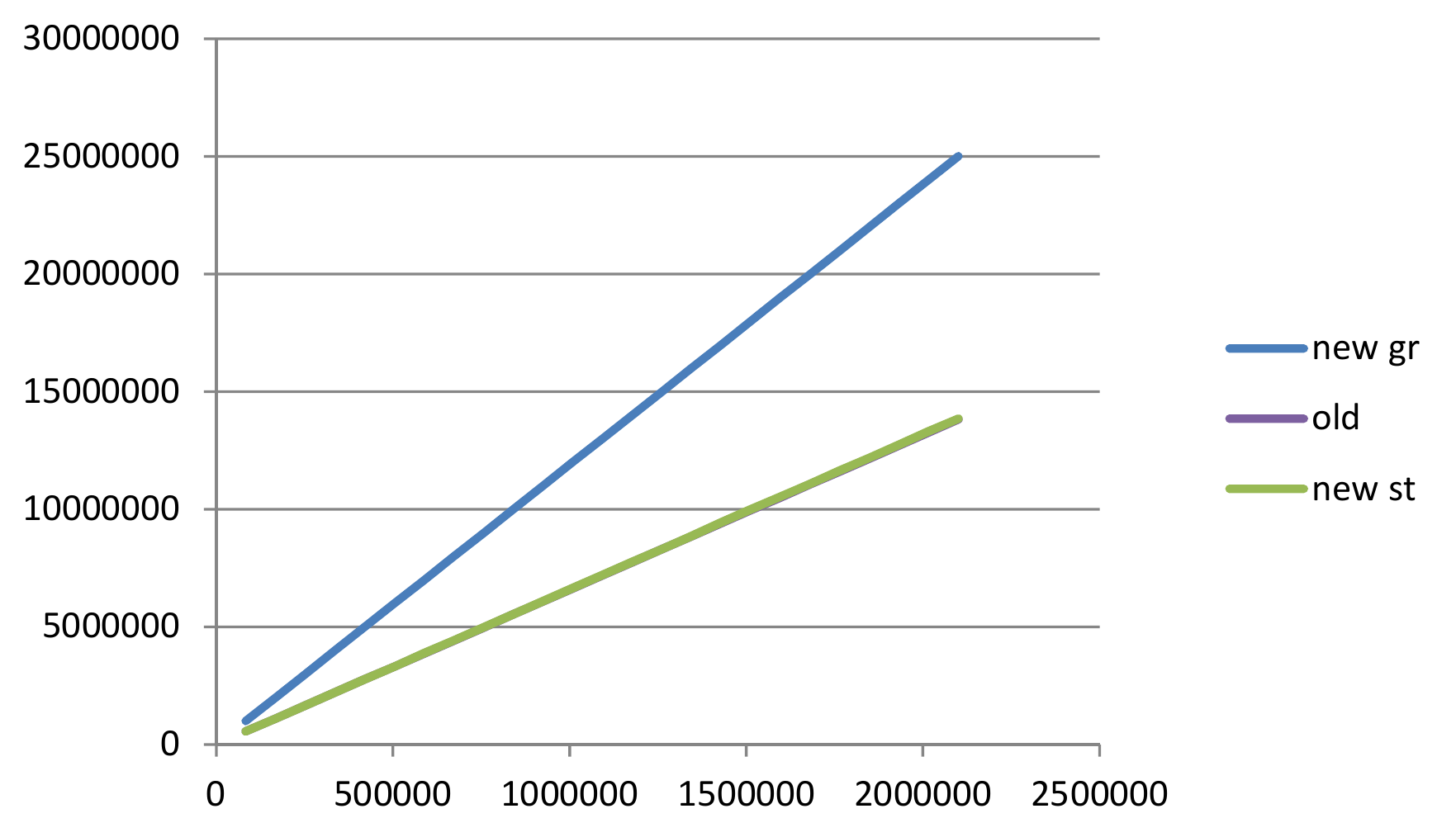}
\\
\includegraphics[trim=52 250 150 250 ,clip,scale=0.25]{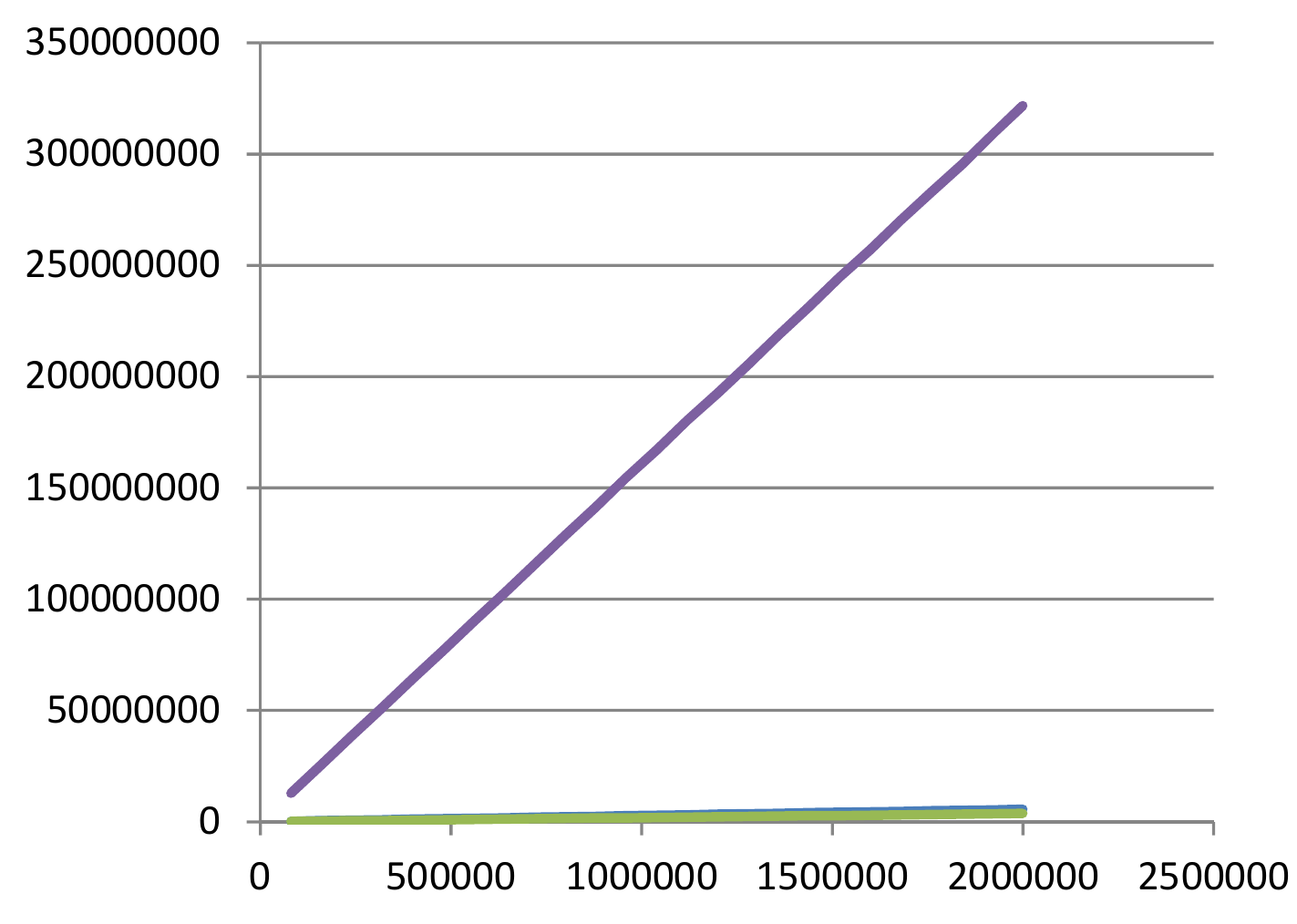} &
\includegraphics[trim=52 250 150 250 ,clip,scale=0.25]{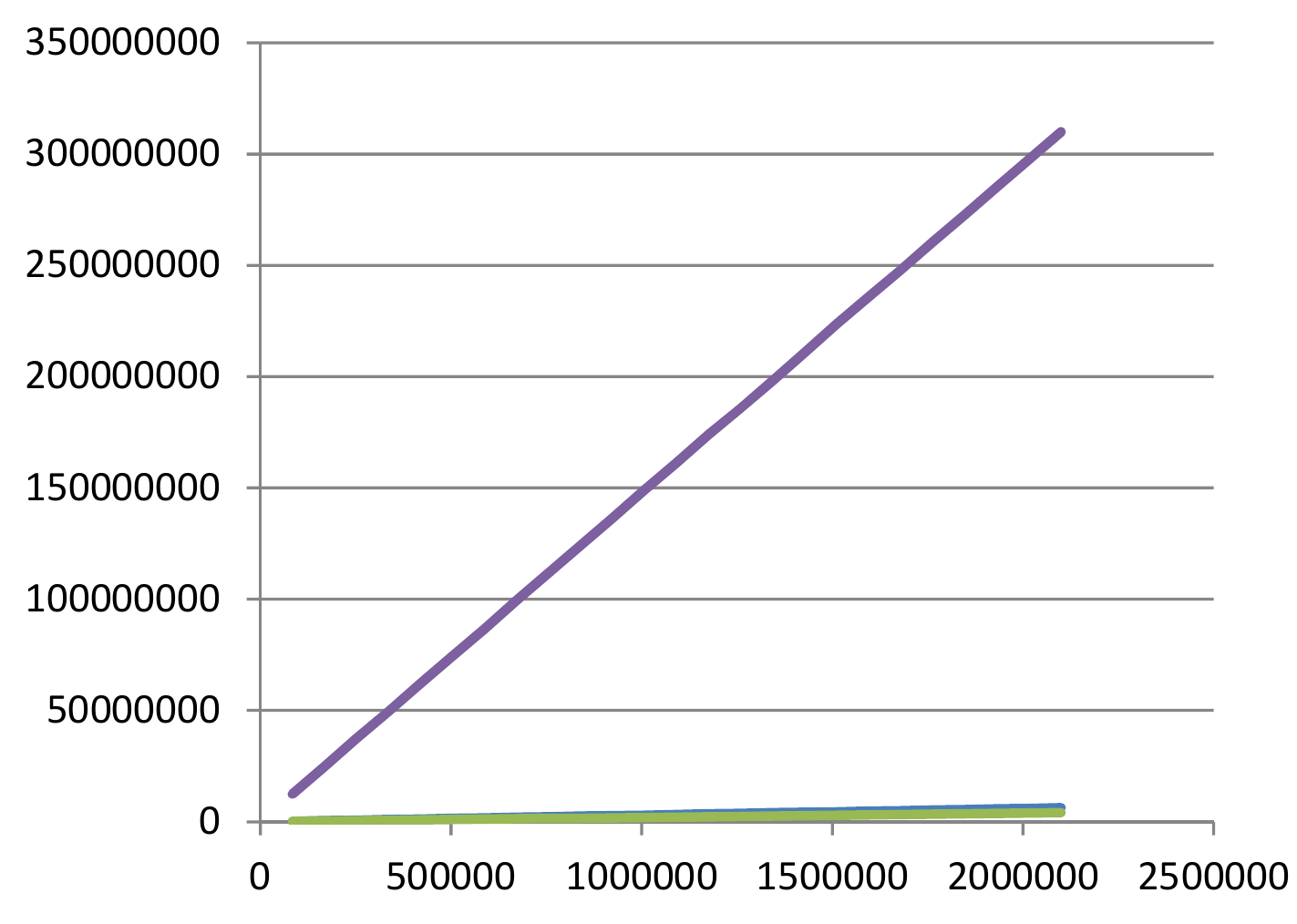} &
\includegraphics[trim=52 250 150 250 ,clip,scale=0.25]{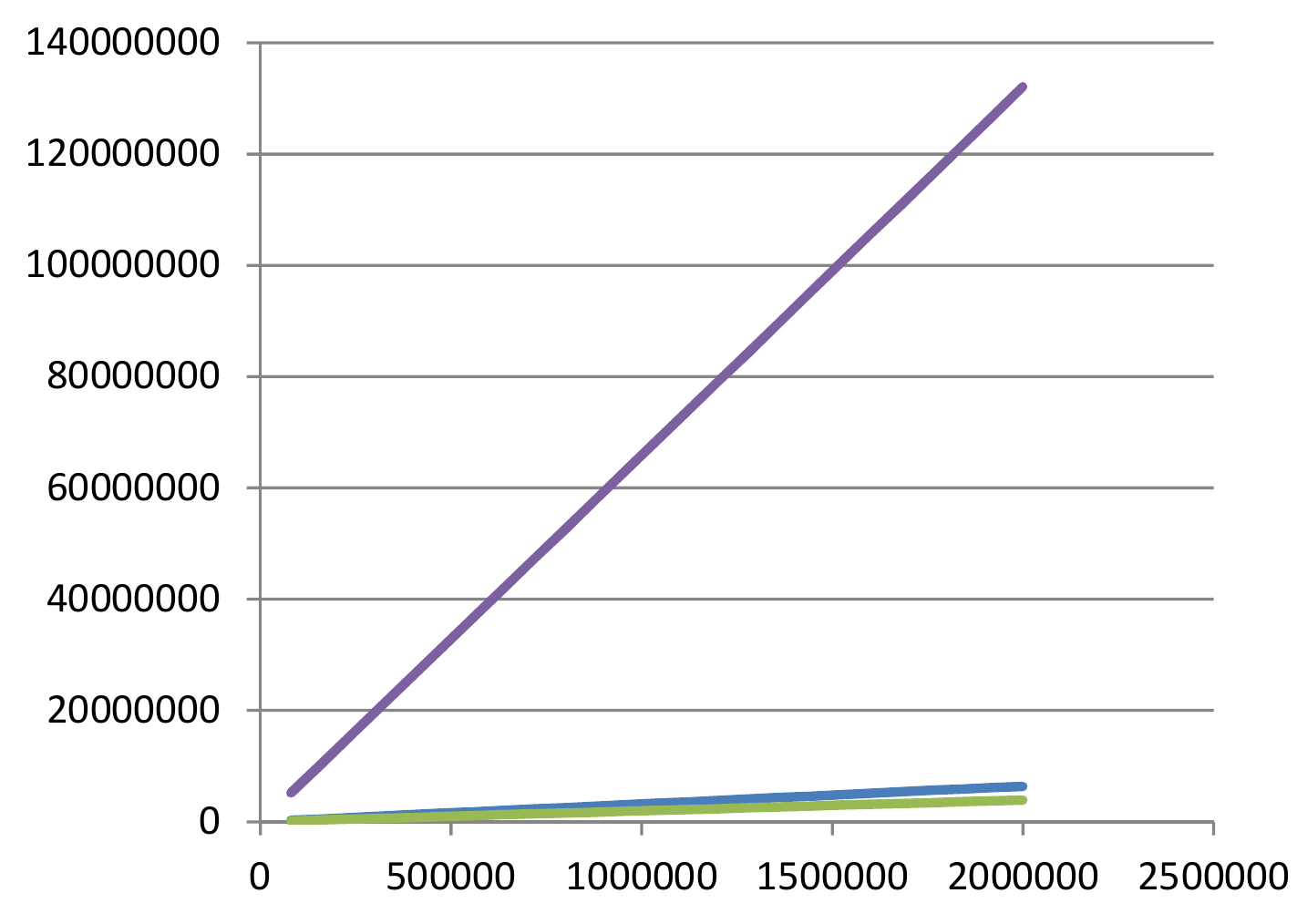} &
\includegraphics[trim=52 250 52 250 ,clip,scale=0.25]{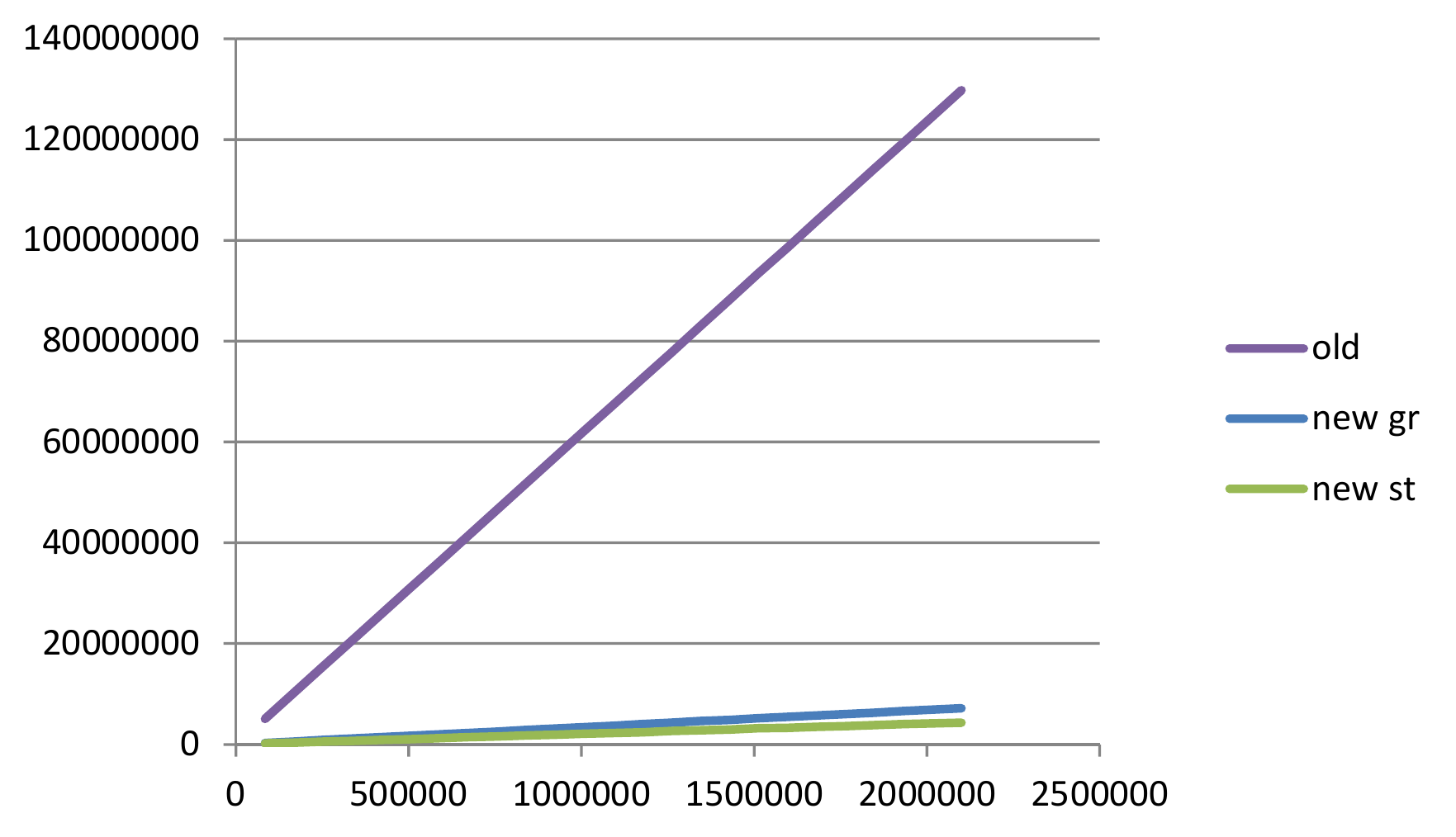}
\\ \\
\multicolumn{4}{c}{4D} \\
\includegraphics[trim=52 250 150 250 ,clip,scale=0.25]{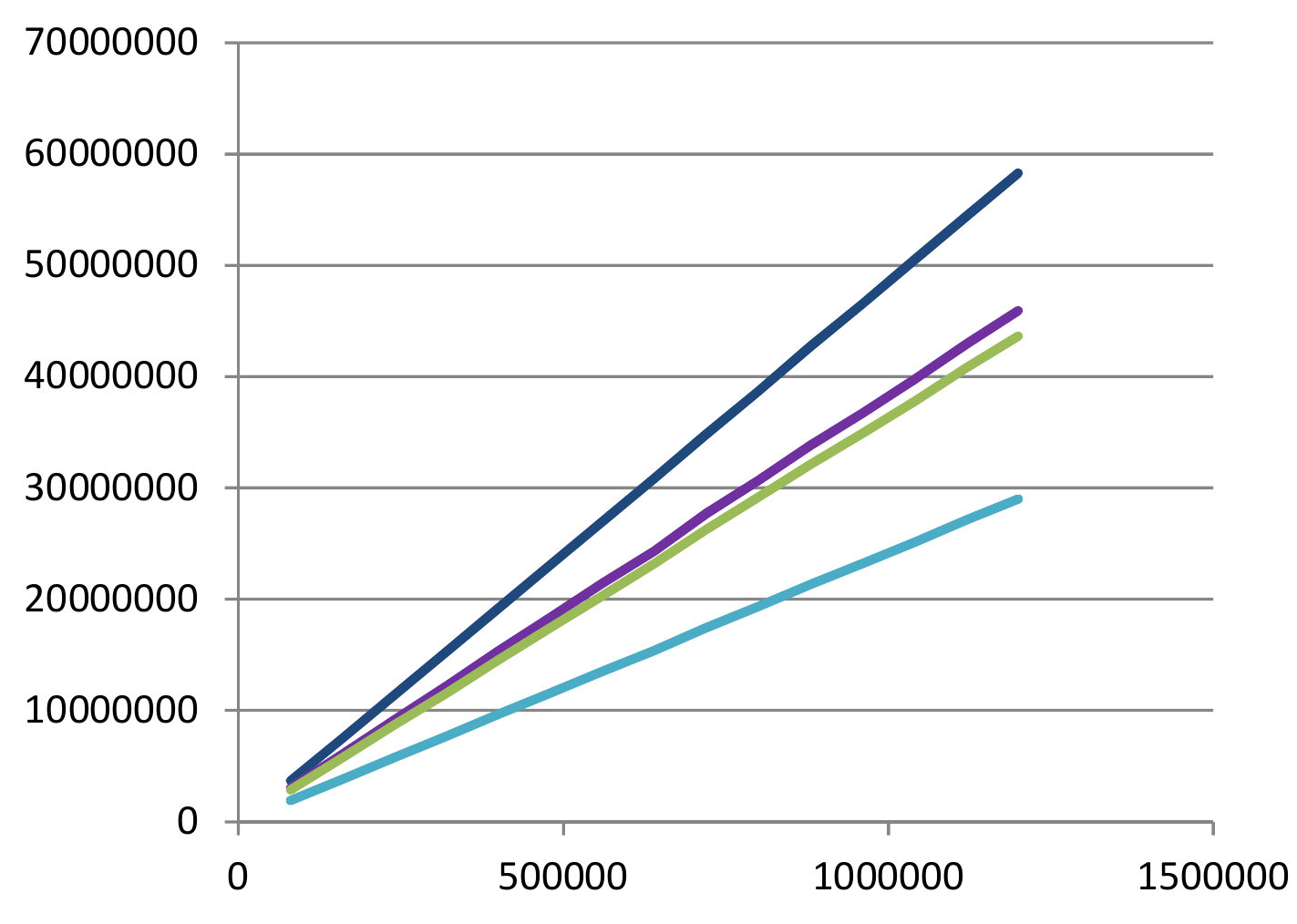} &
\includegraphics[trim=52 250 150 250 ,clip,scale=0.25]{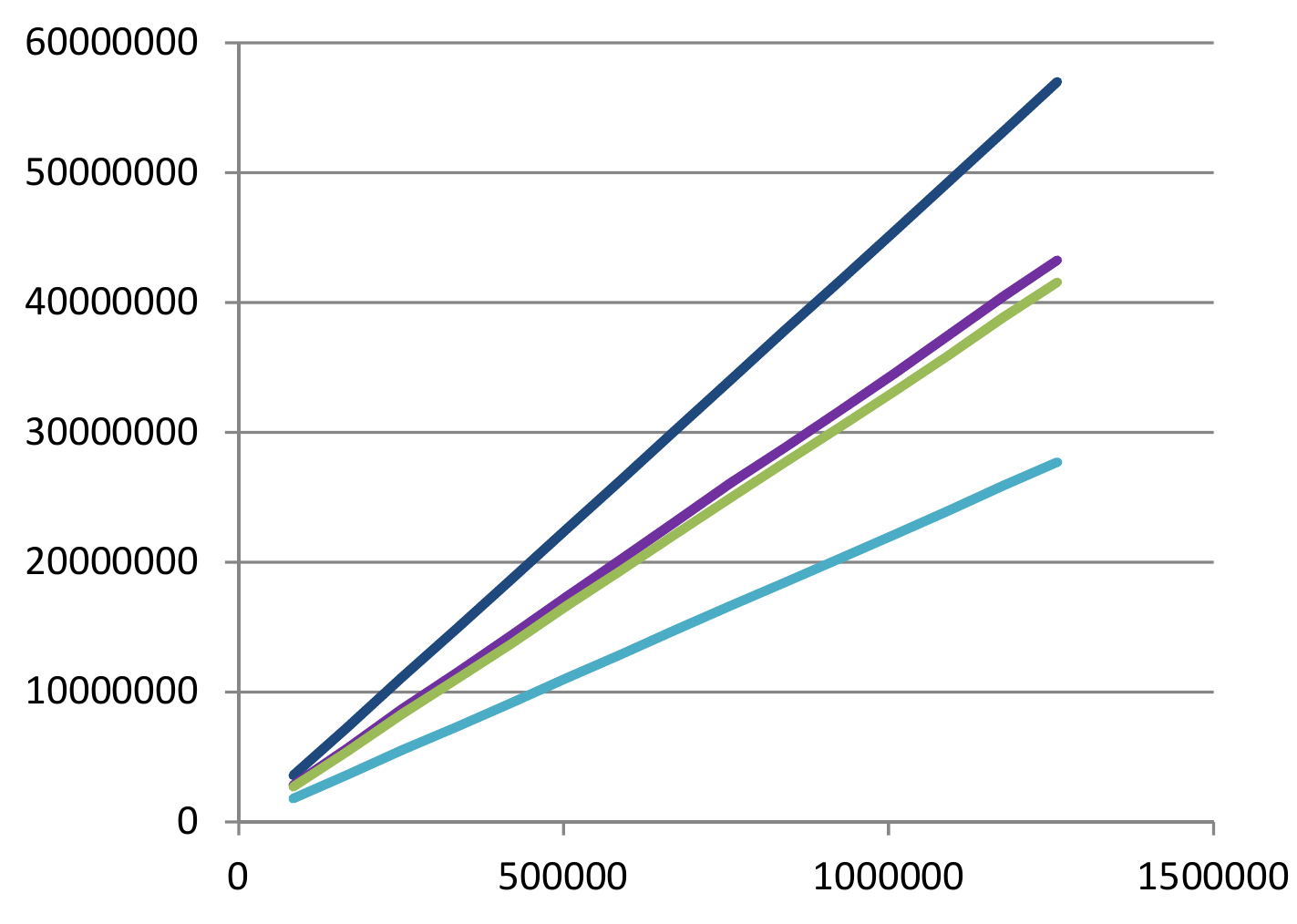} &
\includegraphics[trim=52 250 150 250 ,clip,scale=0.25]{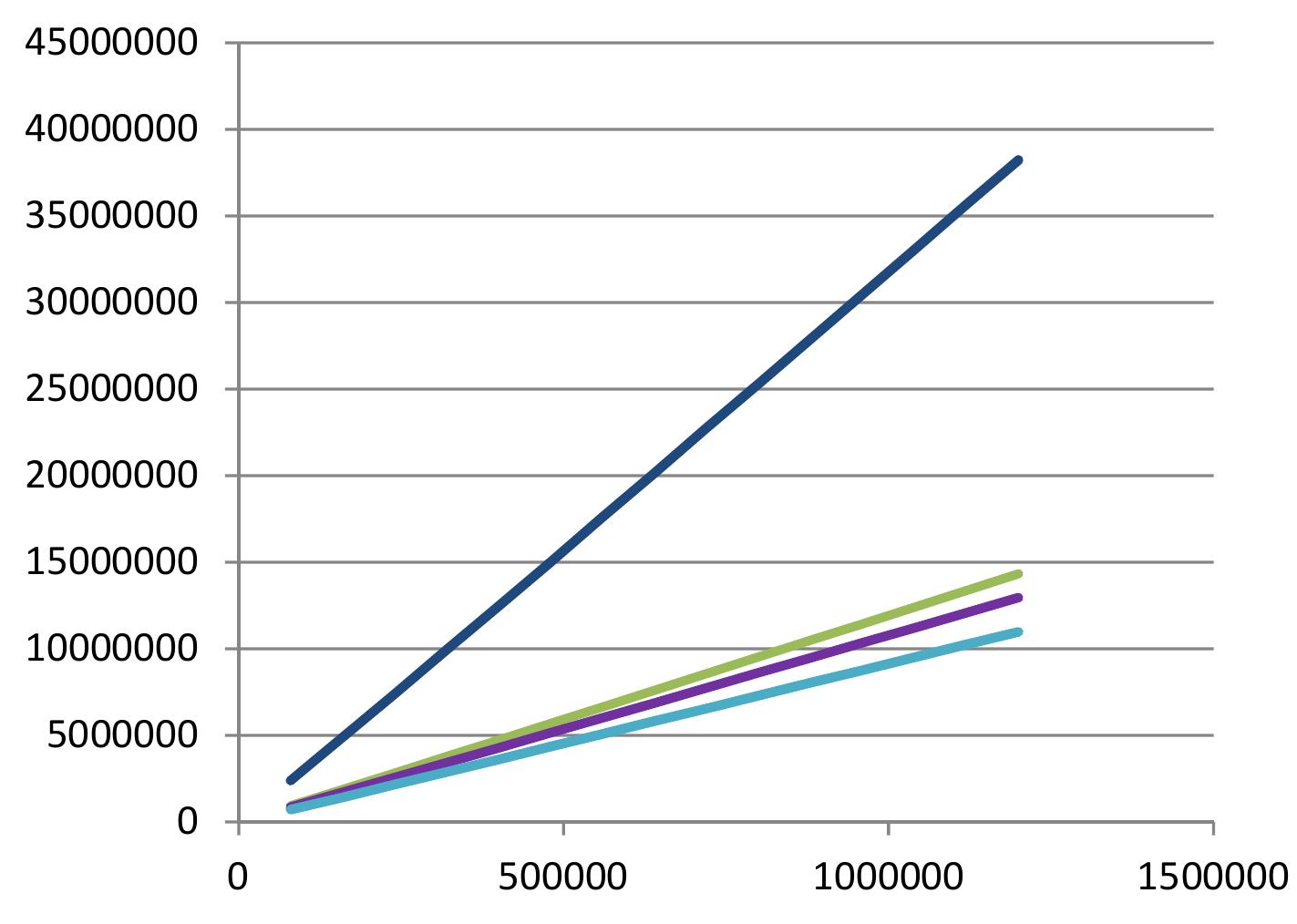} &
\includegraphics[trim=52 250 52 250 ,clip,scale=0.25]{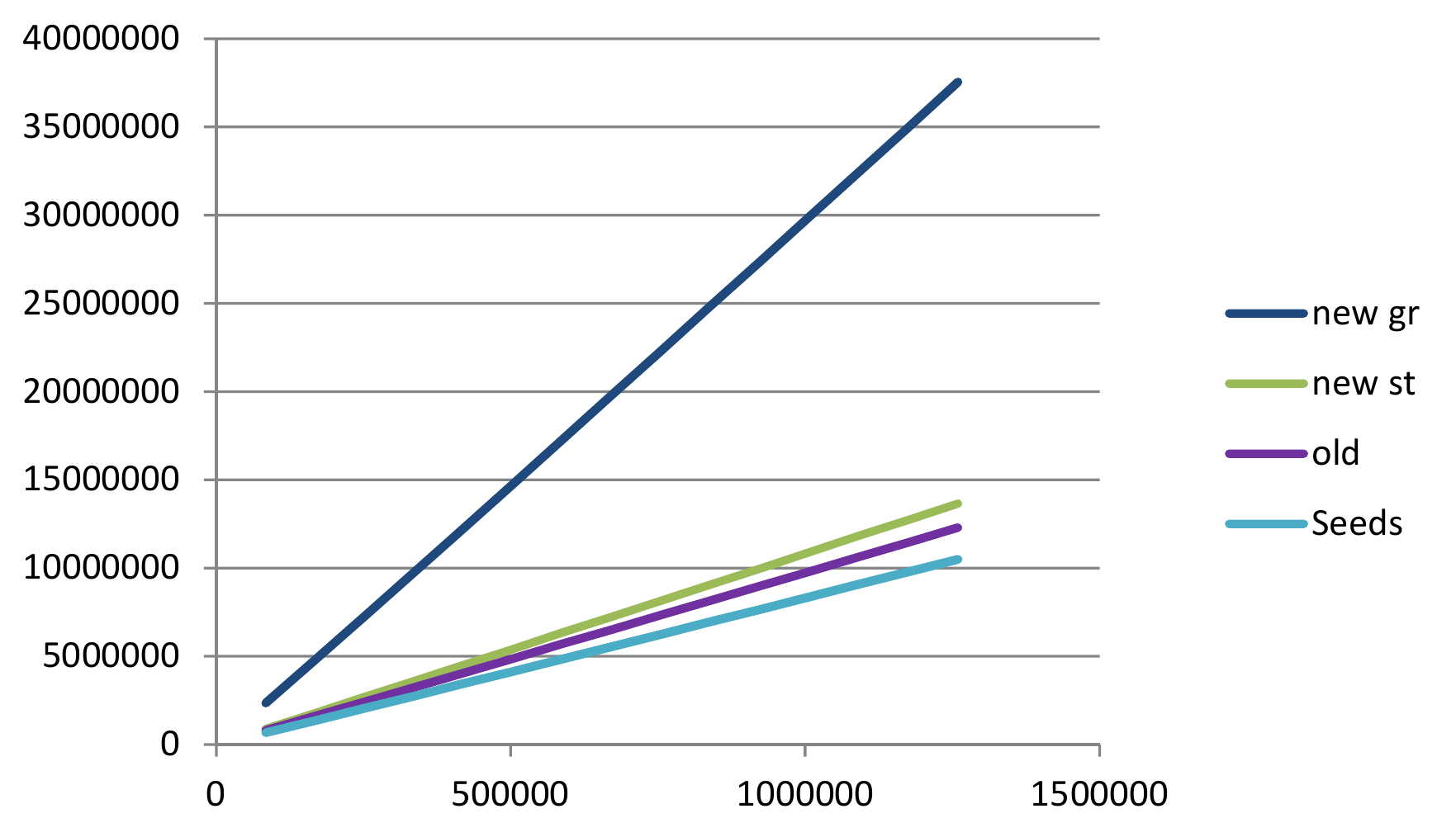}
\\
\includegraphics[trim=52 250 150 250 ,clip,scale=0.25]{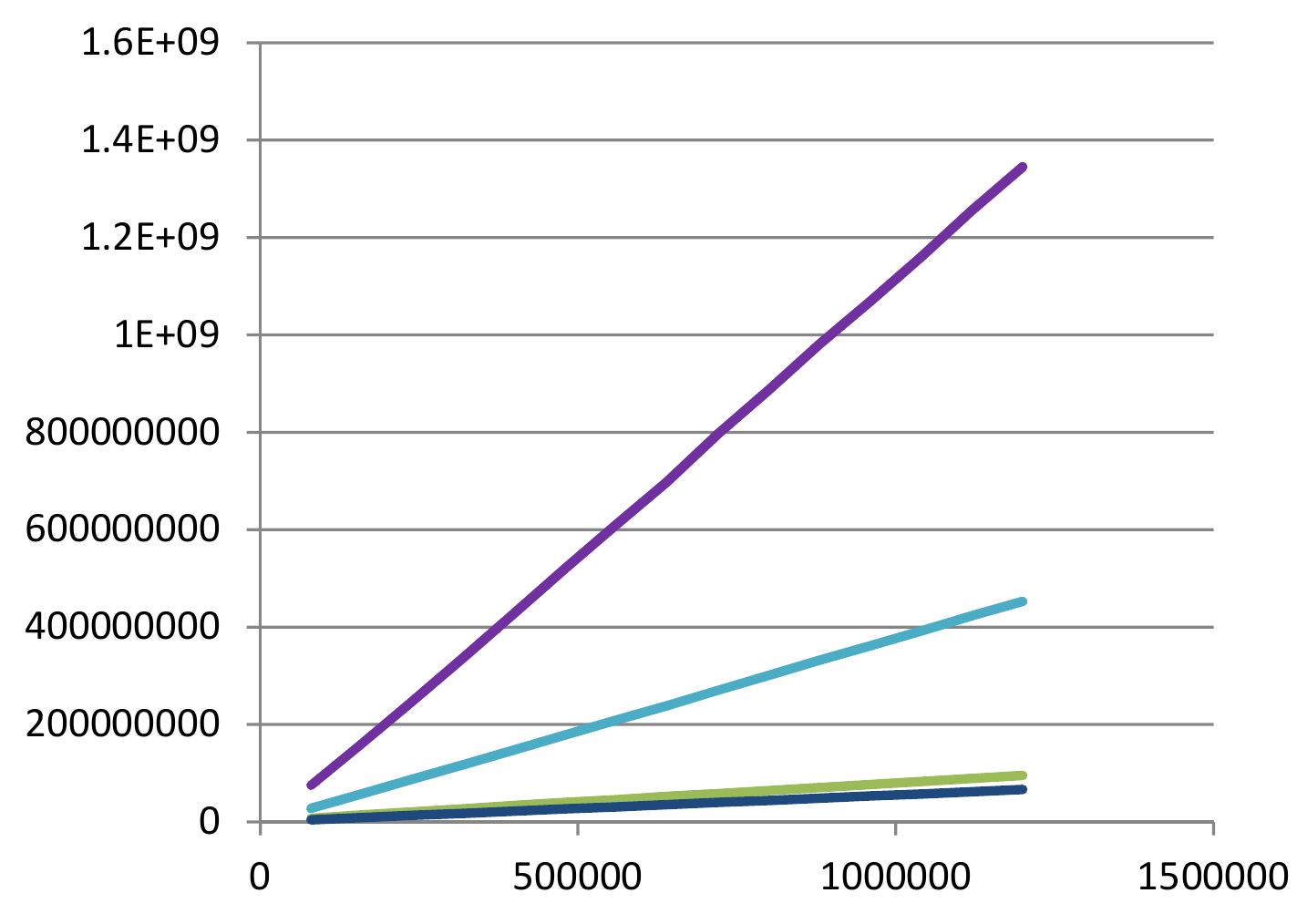} &
\includegraphics[trim=52 250 150 250 ,clip,scale=0.25]{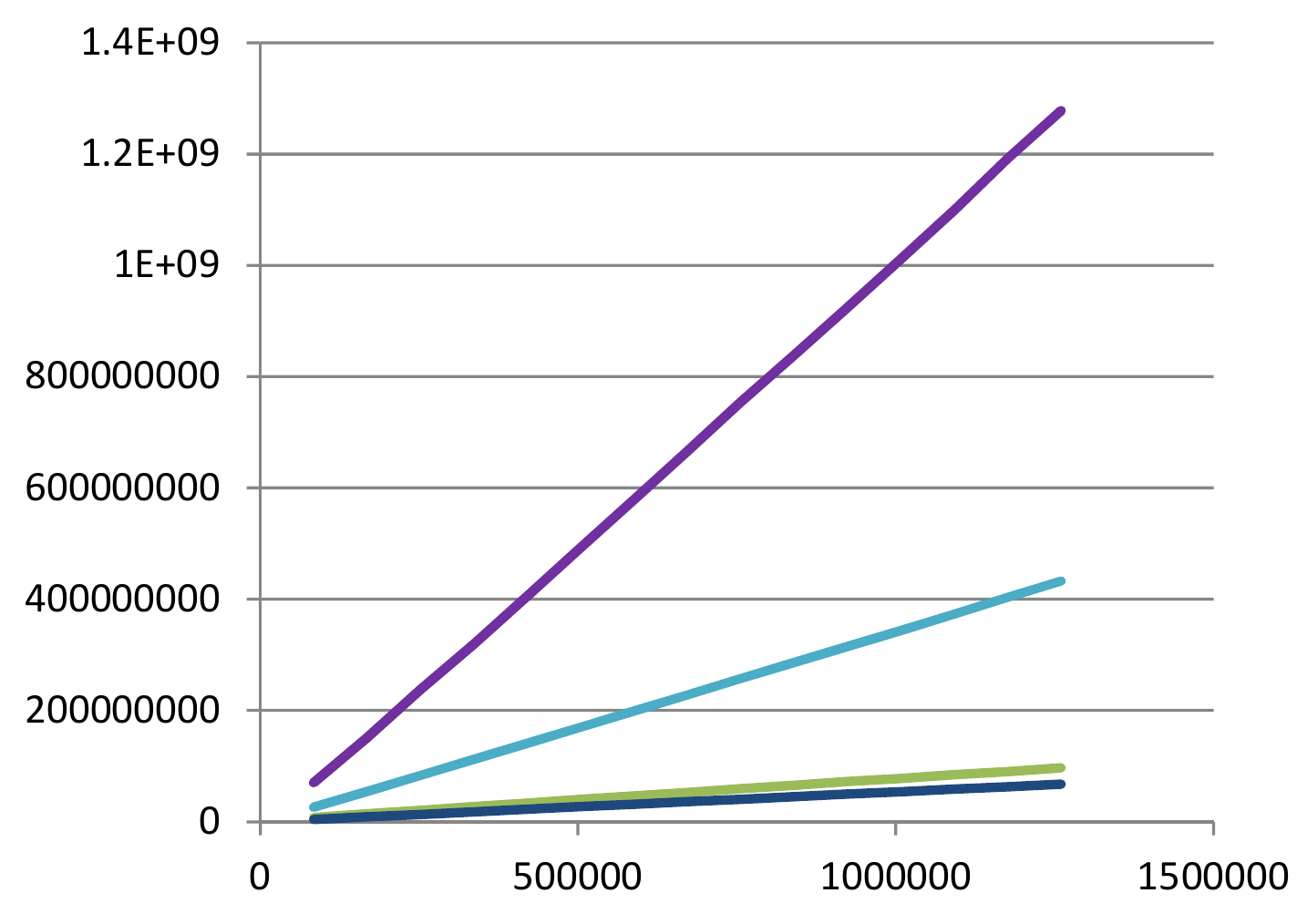} &
\includegraphics[trim=52 250 150 250 ,clip,scale=0.25]{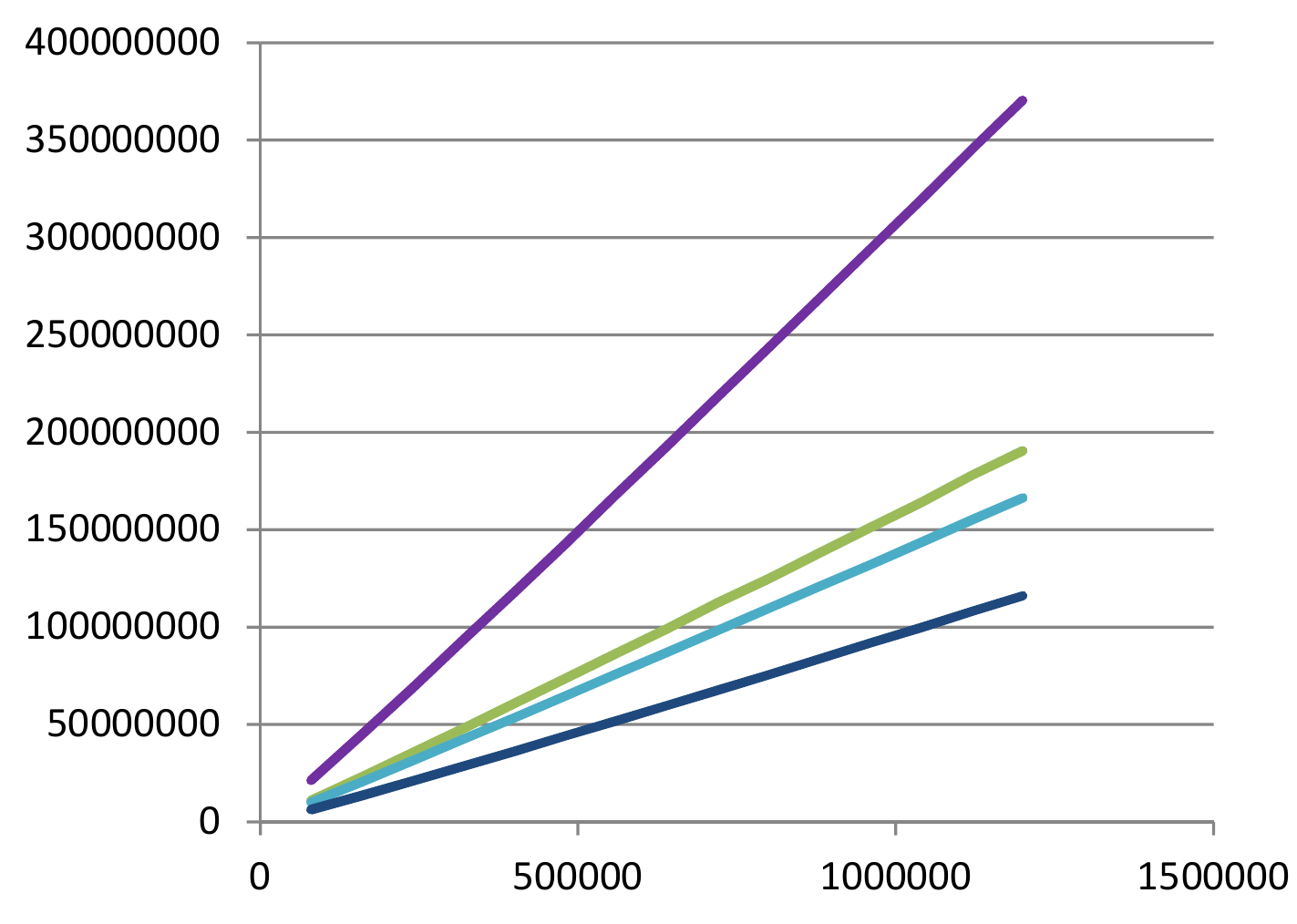} &
\includegraphics[trim=52 250 52 250 ,clip,scale=0.25]{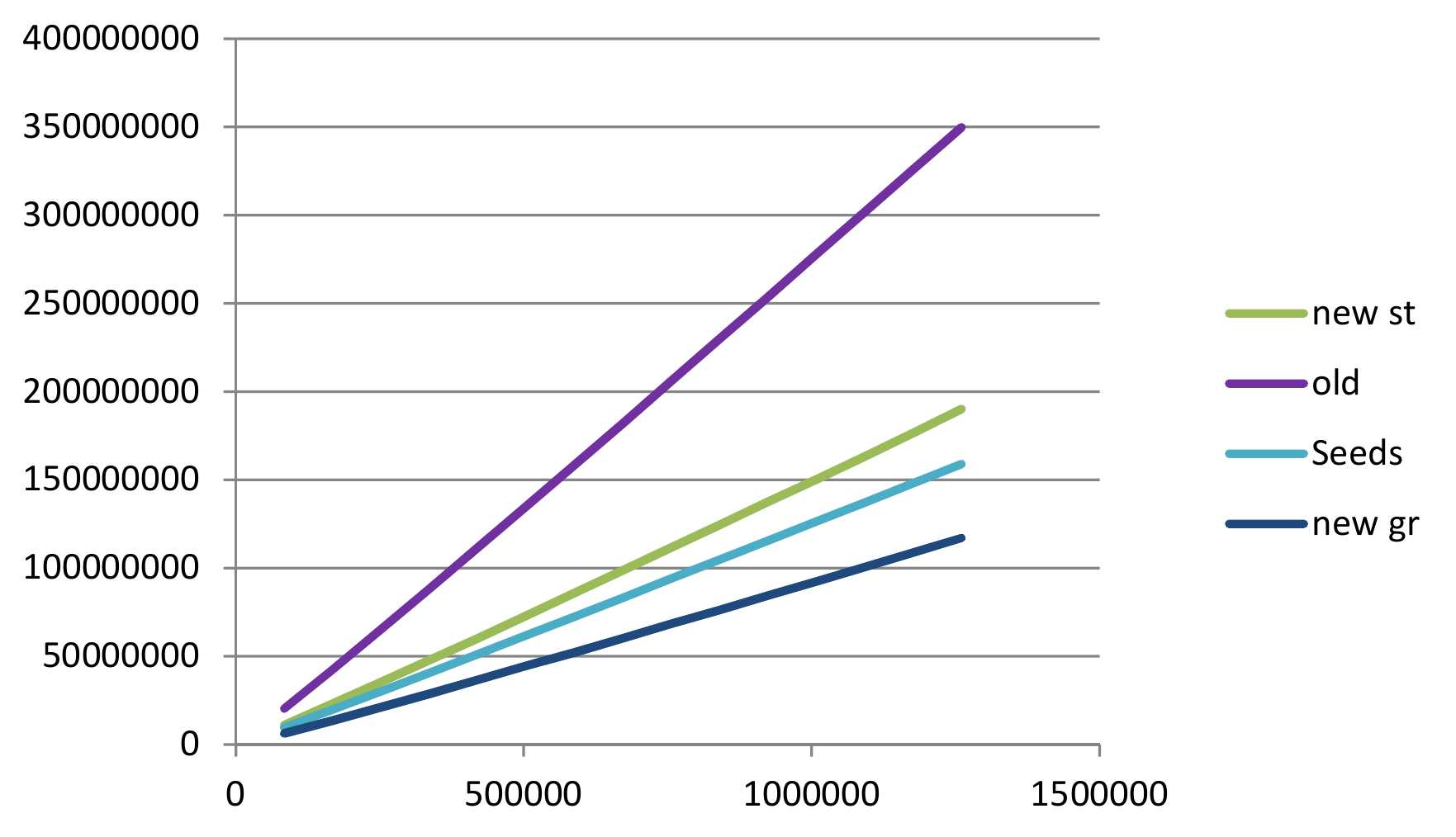}
\end{tabular}
\caption{Results of running the algorithms on data-sets of varying sizes, but with a fixed density. The first and third rows are from experiments with a low density and the second and fourth rows with a high density. The $x$-axis denotes the total number of points in the point set and the $y$-axis the number of distance calculations. The new method with a strip-based or grid-based construction is denoted by \emph{new st} and \emph{new gr} respectively. Note that in the high-density experiments (second and fourth row), the lines for the two new methods almost overlap. However, in the low-density setting, the line for the old method and strip-construction (\emph{new st}) overlap.}
\label{fig:2D-fixed-density}
\end{figure}

\mypara{Discussion.}
The running times in Fig.~\ref{fig:2D-fixed-n} (top row)
show a lot of fluctuation, while the results on the number of distance computations
are very stable. We suspect the fluctuation in running time is related
to memory issues, or to other processes claiming resources.
Hence, most of our conclusions will rely on the number of distance calculations and the sum of neighborhood-sizes.

The results from Fig.~\ref{fig:2D-fixed-n} clearly show the different dependency
on density between the two algorithms. For the original algorithm the number of
distance calculations (shown in the second row) grows linearly with the density.
The sum of the neighborhood sizes (seeds) is slightly over 50\% of the number
of distance computations, so the linear dependency is inherent in the
original algorithm, which reports the complete neighborhood of all points.
The new algorithm, on the other hand, is insensitive to the density.
Note that the bottom row in Fig.~\ref{fig:2D-fixed-n} shows that there is a
fairly large range of values of $\eps$ at which the correct clusters are found.
(The smallest range occurs for Gaussian clusters with noise, and even there the four
clusters are still found for $\eps^2$ roughly between 4 and 20.)
The original algorithm can compete with the new algorithm only
when $\eps$ is chosen very near the lower end of the correct range
of $\eps$; for a ``safe'' value of $\eps$ in the middle of the range, our
algorithm is significantly faster. This is important, as determining the
right value of $\eps$ is hard and may require running the algorithm several times
with different values of $\eps$. Moreover, different clusters may
have different densities. Choosing a value for $\eps$ that is near the low
end of the range for some clusters may then either fail to find some other clusters
or it may lead to an explosion in running time for the other clusters.

Fig.~\ref{fig:2D-fixed-density} shows that the strip-based approach tends to make fewer
comparison than the grid-based approach, especially in the low-density
setting shown in the upper row in the figure. (In Fig.~\ref{fig:2D-fixed-n} the low-density
setting is in the far left in each graph, and therefore less visible.) This is caused
by two things: the boxes are created in a data-driven way in the strip-based approach,
and the boxes are smaller (being bounding boxes of the points in it). The former
means that points within distance $\eps$ end up in the same box more often (in which
case they are not compared), while the latter means that some pairs of bounding boxes
can be at distance more than $\eps$, while the ``corresponding'' grid cells have distance
just below~$\eps$. Except for the low-density setting, the actual computation time
(see the top row in Fig.~\ref{fig:2D-fixed-n}) of the grid-based approach is better in 2D (even though
the number of comparisons is not), due to the smaller constant factors in the approach.

\section{Concluding remarks}
We presented a new algorithm for \dbs in $\Reals^2$,
which runs in $O(n\log n)$ time in the worst case. We also presented an $O(n\log n)$
algorithm for \hdbs in $\Reals^2$---this is the first subquadratic algorithm for this
problem---and we presented near-linear algorithms for approximate \hdbs.
It will be interesting to do a more thorough experimental evaluation of our new algorithms.
Specifically to compare the approach to other implementations that are current available.
It would also be interesting to implement the approximation algorithms and compare them to the exact ones.
From the theoretical side, a main open problem is to see if we can compute
the exact \hdbs hierarchy in subquadratic time in dimensions $d\geq 3$.


\bibliographystyle{plain}                       

\begin{thebibliography}{99}

\bibitem{ac-ohrr-09}
  P.~Afshani and T.M.~Chan.
  Optimal halfspace range reporting in three dimensions.
  In \emph{Proc. 20th ACM-SIAM Symposium on Discrete Algorithms}, pages 180--186, 2009.

\bibitem{aes-emstbcp-91}
   P.K.~Agarwal, H.~Edelsbrunner, and O.~Schwarzkopf.
   Euclidean minimum spanning trees and bichromatic closest pairs.
   \emph{Discr. Comput. Geom.}~6:407--422 (1991).

\bibitem{abks-optics-99}
  M.~Ankerst, M.M.~Breunig, H.-P.~Kriegel, and J.~Sander.
  OPTICS: ordering points to identify the clustering structure.
  \emph{SIGMOD Rec.}~28:49--60 (1999).

\bibitem{ac-bed-14}
  S.~Arya and T.M.~Chan.
  Better $\varepsilon$-dependencies for offline approximate nearest-neighbor search,
  Euclidean minimum spanning trees, and $\varepsilon$-kernels.
  In \emph{Proc. 30th Symposium on Computational Geometry}, pages 416--425, 2014.

\bibitem{bb-idbscan-04}
  B.~Borah and D.~Bhattacharyya.
  An improved sampling-based DBSCAN for large spatial databases.
  In \emph{Proc. Int. Conf. on Intelligent Sensing and Information Processing}, pages 92--96, 2004.

\bibitem{bcko-cgaa-08}
  M.~de~Berg, O.~Cheong, M.~van~Kreveld and M.~Overmars.
  \emph{Computational Geometry: Algorithms and Applications (3nd edition)}.
  Springer-Verlag, 2008.

\bibitem{cms-dbchde-13}
  R.J.G.B.~Campello, D.~Malouvi and J. Sander.
  Density-based clustering based on hierarchical density estimates.
  In \emph{Proc. 17th Pacific-Asia Conference on Knowledge Discovery and Data Mining},
  LNCS~7819, pages 160--172, 2013.

\bibitem{csx-gadbc-05}
  D.Z.~Chen, M.H.~Smid and B.~Xu.
  Geometric Algorithms for Density-based Data Clustering.
  \emph{Int. J. Comput. Geometry Appl.} 15:239--260 (2005)

\bibitem{clrs-ia-09}
  T.H.~Cormen, C.E.~Leiserson, R.L.~Rivest and C.~Stein.
  \emph{Introduction to Algorithms} (3rd edition), MIT Press, 2009.

\bibitem{e-rcgp-95}
   J. Erickson.
   On the relative complexities of some geometric problems.
   In \emph{Proc. 7th Canadian Conf. Comput. Geom. (CCCG)}  pages 85--90, 1995.


\bibitem{eksx-dbadc-96}
  M.~Ester, H.-P.~Kriegel, J.~Sander, and X.~Xu.
  A density-based algorithm for discovering clusters in large spatial databases with noise.
  In \emph{Proc.~2nd Int. Conference on Knowledge Discovery and Data Mining (KDD)}, pages 226--231, 1996.

\bibitem{gt-dbs-15}
  J.~Gan and Y.~Tao.
  DBSCAN revisited: Mis-claim, un-fixability, and approximation.
  In \emph{Proc. 2015 ACM SIGMOD Int. Conf. on Management of Data}, pages 519--530.

\bibitem{ghk-hodt-02}
  J.~Gudmundsson, M.~Hammer and M.~van Kreveld.
  Higher order Delaunay triangulations.
  \emph{Computational Geometry: Theory and Applications}~23: 85--98 (2002).

\bibitem{g-fadbs-13}
  A.~Gunawan.
  A faster algorithm for DBSCAN.
  Master's thesis, Technische University Eindhoven, March 2013.

\bibitem{l-fdbscan-06}
  B.~Liu.
  A fast density-based clustering algorithm for large databases.
  In \emph{Proc. Int. Conf. on Machine Learning and Cybernetics}, pages 996--1000, 2006.


\bibitem{mm-uga-08}
  S.~Mahran and K.~Mahar.
  Using grid for accelerating density-based clustering.
  In \emph{8th Int. Conf. on Computer and Information Technology}, pages 35--40, 2008.

\bibitem{m-rphs-92}
  J.~Matou\v{s}ek.
  Reporting points in halfspaces.
  \emph{Computational Geometry: Theory and Applications}~2: 169--186 (1993).

\bibitem{ns-gsn-07}
  G.~Narasimhan and M.~Smid.
  \emph{Geometric Spanner Networks}.
  Cambridge University Press, 2007.

\bibitem{ppalmc-poptics-13}
  M.M.A.~Patwary, D.~Palsetia, A.~Agrawal, W.-K.~Liao, F.~Manne, and A.~Choudhary.
  Scalable parallel OPTICS data clustering using graph algorithmic techniques.
  In \emph{Proc. Int. Conf. on High Performance Computing, Networking, Storage and Analysis}
  pages~49:1--49:12, 2013.

\bibitem{tsk-idm-06}
  P.~Tan, M.~Steinbach and V.~Kumar.
  \emph{Introduction to Data Mining.}
  Addison-Wesley (2006).

\bibitem{v-annp-89}
  P.M.~Vaidya.
  An $O(n \log n)$ algorithm for the all-nearest-neighbor problem.
  \emph{Discr. Comput. Geom.}~4:101--115 (1989).

\end{thebibliography}

\end{document}